\documentclass[prx,twocolumn,showpacs,superscriptaddress,preprintnumbers,amssymb]{revtex4-2} 
\usepackage{graphicx}
\usepackage{latexsym}
\usepackage{amssymb}
\usepackage{amsmath}
\usepackage{mathtools}
\usepackage{amsmath} 
\usepackage{amsfonts} 
\usepackage{upgreek}
\usepackage{bm}
\usepackage{multirow}
\usepackage[shortlabels]{enumitem}
\usepackage{color}
\usepackage[colorlinks, citecolor=blue]{hyperref}
\usepackage{float}
\usepackage{xcolor}
\usepackage{physics}
\usepackage{tcolorbox}
\usepackage{manfnt}
\usepackage{lipsum}
\usepackage[ruled,linesnumbered]{algorithm2e}
\usepackage{tabularx}
\usepackage{nicematrix}
\usepackage{braket}
\usepackage{bm}
\usepackage[normalem]{ulem}
\usepackage{tikz}
\tcbuselibrary{breakable}

\newtcolorbox{axiombox}[1][]{
  breakable,
  colback=gray!15,    
  boxrule=0pt,     
  arc=3mm,           
  left=2mm,          
  right=2mm,         
  top=2mm,           
  bottom=2mm,        
  auto outer arc,
  #1
}

\NiceMatrixOptions{
code-for-first-row = \color{blue} ,
code-for-last-row = \color{blue} ,
code-for-first-col = \color{blue} ,
code-for-last-col = \color{blue}
}

\usepackage{tikz}
\usepackage{tikz-cd}
\usetikzlibrary{arrows}
\usetikzlibrary{intersections}
\usetikzlibrary{shapes.geometric}
\usetikzlibrary{decorations.pathmorphing, patterns,shapes}
\usetikzlibrary{decorations.markings}
\usetikzlibrary{patterns}
\usetikzlibrary{fit}

\tikzset{
	mid arrow/.style={postaction={decorate,decoration={
				markings,
				mark=at position .575 with {\arrow[#1]{stealth}}
	}}},
	near arrow/.style={postaction={decorate,decoration={
				markings,
				mark=at position .275 with {\arrow[#1]{stealth}}
	}}},
	far arrow/.style={postaction={decorate,decoration={
				markings,
				mark=at position .800 with {\arrow[#1]{stealth}}
	}}},
}
\pgfmathsetmacro\MathAxis{height("$\vcenter{}$")}

\usepackage{amsthm}

\theoremstyle{plain} 
\newtheorem{theorem}{Theorem}

\newtheoremstyle{axiomwithbreak} 
  {3pt}   
  {3pt}   
  {} 
  {}      
  {\bfseries} 
  {.}     
  {\newline }   
  {}      

\theoremstyle{axiomwithbreak}
\newtheorem{axiom}{Axiom}

\theoremstyle{definition}               
\newtheorem{definition}[theorem]{Definition}

\newtheorem{exmp}{Example}
\makeatother

\theoremstyle{plain}
\newtheorem{corollary}{Corollary}[theorem]
\newtheorem{lemma}[theorem]{Lemma}
\newtheorem{Proposition}[theorem]{Proposition}
\newtheorem{Conjecture}[theorem]{Conjecture}

\definecolor{Mathematica1}{rgb}{0.368417, 0.506779, 0.709798}
\definecolor{Mathematica2}{rgb}{0.880722, 0.611041, 0.142051}
\definecolor{Mathematica3}{rgb}{0.560181, 0.691569, 0.194885}

\newcommand{\secref}[1]{Sec.\,\ref{#1}}
\newcommand{\appref}[1]{Appendix.\,\ref{#1}}
\newcommand{\eqnref}[1]{Eq.\,\eqref{#1}}

\newcommand{\figref}[1]{Fig.\,\ref{#1}}

\newcommand{\vdagger}{{\vphantom{\dagger}}}

\hypersetup{colorlinks=true, 
linkcolor=green!50!black,
citecolor=red!40!orange!90!black,
filecolor=OliveGreen, 
urlcolor=Mathematica2!80!orange!95!black,
filebordercolor={.8 .8 1}, 
urlbordercolor={.8 .8 0}}

\usepackage{multirow} 

\usepackage{wrapfig}
\usepackage{setspace}

\definecolor{BSorange}{RGB}{140,50,0}

\newcommand{\calA}{{\mathcal A}}

\newcommand{\calC}{{\mathcal C}}

\newcommand{\calE}{{\mathcal E}}

\newcommand{\calH}{{\mathcal H}}

\newcommand{\calO}{{\mathcal O}}
\newcommand{\calQ}{{\mathcal Q}}
\newcommand{\calS}{{\mathcal S}}

\newcommand{\calR}{{\mathcal R}}
\newcommand{\calM}{{\mathcal M}}
\newcommand{\calP}{{\mathcal P}}
\newcommand{\rd}{{\partial}}

\newcommand{\Cgh}{C_{g,h}}
\newcommand{\Eg}{E_{g}}

\newcommand{\LW}{{\rm LW}}
\newcommand{\conv}{{\rm conv}}

\newcommand{\Id}{{\textrm{Id}}}

\usepackage{bbm}
\newcommand{\bmflux}[1]{{\mathbf{\color{blue!50!cyan!80!black}#1}}}

\newcommand{\bmcharge}[1]{{\mathbf{\color{red!50!orange!80!black}#1}}}
\newcommand{\bmstate}[1]{{\mathbbm{#1}}}

\newcommand{\rAngle}[1]{\rangle \hspace{-2pt} \rangle }

\begin{document}

\title{Topological Mixed States:  Phases of Matter from Axiomatic Approaches}

\author{Tai-Hsuan Yang}
\affiliation{Department of Physics, Illinois Quantum Information Science and Technology Center, and Institute of Condensed Matter Theory, The Grainger College of Engineering, University of Illinois at Urbana-Champaign, Urbana, Illinois 61801, USA}
\author{Bowen Shi}
\affiliation{Department of Physics, Illinois Quantum Information Science and Technology Center, and Institute of Condensed Matter Theory, The Grainger College of Engineering, University of Illinois at Urbana-Champaign, Urbana, Illinois 61801, USA}
\affiliation{Department of Computer Science, University of California, Davis, CA 95616, USA}
\author{Jong Yeon Lee}
\affiliation{Department of Physics, Illinois Quantum Information Science and Technology Center, and Institute of Condensed Matter Theory, The Grainger College of Engineering, University of Illinois at Urbana-Champaign, Urbana, Illinois 61801, USA}
\affiliation{Korea Institute for Advanced Study, Seoul 02455, South Korea}

\date{\today}
\begin{abstract}
For closed quantum systems, topological orders are understood through the equivalence classes of ground states of gapped local Hamiltonians. The generalization of this conceptual paradigm to open quantum systems, however, remains elusive, often relying on operational definitions without fundamental principles. Here, we fill this gap by proposing an approach based on three axioms: ($i$) local recoverability, ($ii$) absence of long-range correlations, and ($iii$) spatial uniformity. 
States that satisfy these axioms are fixed points; requiring the axioms only after coarse-graining promotes each fixed point to an equivalence class, i.e. a phase, presenting the first step toward the axiomatic classification of mixed-state phases of matter: the \emph{mixed-state bootstrap program}.

From these axioms, a rich set of topological data naturally emerges; importantly, these data are \emph{robust} under relaxation of axioms.  
For example, each topological mixed state supports locally indistinguishable classical and/or quantum logical memories with distinct responses to topological operations.
These data label distinct mixed-state phases, allowing one to distinguish them. 
We further uncover a hierarchy of secret-sharing constraints: in non-Abelian phases, reliable recovery---even of information that looks purely classical---demands a specific coordination among spatial subregions, a requirement different across non-Abelian classes. This originates from non-Abelian fusion rules that can stay robust under decoherence.
Finally, we performed large-scale numerical simulations to corroborate stability: weakly decohered fixed points respect the axioms once coarse-grained. Surprisingly, the axiom violations at different length scales exhibit an exact crossing at the critical point, enabling its identification without any rescaling or data-collapse ambiguities.
These results lay the foundation for a systematic classification of topological states in open quantum systems.
\end{abstract} 
\maketitle

\tableofcontents

\section{Introduction}

Entanglement is one of the fundamental principles underpinning the classification of phases of matter. Building on this concept, it has been proposed that two pure states belong to the same phase if they can be continuously deformed into one another via finite-depth local unitary operations (FDLU)~\cite{Hasting_quasiadiabaticcontinuation_2005,xiechen_localunitary2010, PhysRevB.84.165139, Chen_2011}. Extending this idea to mixed states, however, is nontrivial. A straightforward generalization--substituting finite-depth local unitary operations with finite-depth local quantum channels (FDLC)--overlooks a crucial requirement: 
the smooth deformation of the underlying entanglement structure. While finite-depth local unitary operations modify entanglement only gradually, local quantum channels can induce abrupt changes, potentially destroying the topological correlation across the system. 
This subtle oversight has contributed to some ambiguity in the study of mixed-state topological phases.
This paper aims to clarify this point and provide a more robust and consistent framework.

The classification of phases of matter has been a central theme in quantum many-body physics, traditionally focused on identifying distinct ground states of local gapped Hamiltonians. Over the past few decades, entanglement has emerged as a pivotal tool in this endeavor, offering a perspective that goes beyond Landau's spontaneous symmetry breaking. In particular, quantum phases can be broadly categorized by their entanglement structure into short-range entangled (SRE) or long-range entangled (LRE) states~\cite{xiechen_localunitary2010,Zeng2015book}, which can be further enriched under the presence of symmetries. LRE states that are robust against local perturbations are called \emph{topological}, and in two dimensions, these states are described by anyon theories \cite{Kitaev2005,Levin2005,Rowell2007}, which characterize the emergent quasiparticles and their braiding statistics.

A foundational principle in defining and classifying gapped phases is the existence of an energy gap in the parent local Hamiltonian \cite{Hasting_quasiadiabaticcontinuation_2005,Levin2005,Kitaev2005,Fidkowski2010}. Traversing from one phase to another typically necessitates gap closing that marks the phase boundary as in \figref{fig:1}(a). The entanglement bootstrap program \cite{Shi2019fusion,Shi2020domainwall,Huang2021,Braiding2023,Lin2023,confrmal2025,2024strict-A1-string-net,strict-J-2024} offers a complementary perspective. 
It begins with a set of axioms designed to strip away all short-range structures, retaining only long-range entanglement. 
Two fixed point wavefunctions that satisfy these axioms but cannot be connected by an FDLU circuit define distinct topological phases, see \figref{fig:1}(b).
Conversely, relaxing the axioms around a chosen fixed point and allowing FDLU deformations generates the entire phase anchored to that point.  
This construction yields a robust definition of topological order, independent of explicit Hamiltonian constructions~\footnote{For chiral topological orders, we refer to \cite{strict-J-2024} for more in-depth discussions. }. Despite these advances in the pure-state regime, much less is known about how to classify mixed states in a similarly rigorous manner, particularly when the guiding notion of a spectral gap is unavailable.

\begin{figure}[!t]
     \centering 
     \includegraphics[width=0.96\linewidth]{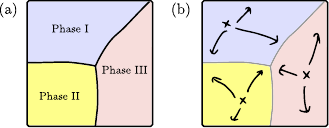}
     \caption{ {\bf Phases of matter} defined through {\bf (a)} equivalence relations (adiabatic connection without closing gaps of a local Hamiltonian), which specify boundaries as well, and {\bf (b)} fixed points and allowed deformations, which are finite-depth (quasi)-local operations.}
     \label{fig:1}
\end{figure}

From the information-theoretic perspective, understanding which aspects of topological order survive after decoherence is both conceptually intriguing and practically important.
Recent studies have shown that decohered topological states--intrinsic or symmetry-protected can still encode nontrivial quantum or classical logical information across various regimes~\cite{lee2022symmetry, fan2023diagnostics, Lee_2023, bao2023mixedstate, Wang_intrinsicmixedstateTO_2023, Zhao_2024, lee2024exact, Sang2024, Niwa_2025}. These findings support the growing consensus that one can meaningfully define mixed-state generalizations of topological phases~\cite{lee2022symmetry, fan2023diagnostics, Lee_2023, bao2023mixedstate, Wang_intrinsicmixedstateTO_2023, Zhao_2024, lee2024exact, Sang2024, Niwa_2025, Ma2023, sang2023mixedstate, lyons2024understanding, 
Coser_2019,
xue2024tensornetworkformulationsymmetry,
ma2024symmetryprotectedtopologicalphases, Sohal2024,Ellison2025,
PhysRevB.111.115141,
chen2023separability,
chen2023symmetryenforced,
guo2023twodimensional,
su2024tapestry,
Peter-Lu2024,sala2024stability,
li2024entanglementneededemergentanyons,
PhysRevB.111.125128,
li2024replicatopologicalorderquantum,
chen2024unconventionaltopologicalmixedstatetransition,
PRXQuantum.6.010347, 
Lu2024,
Negari2024, Ogata2025}. 
Substantial progress has been made along these lines. The doubled space formalism offers a qualitative organizing framework for decoherence-driven transitions~\cite{lee2022symmetry, bao2023mixedstate}, while generalizations of strange correlators~\cite{lee2022symmetry, PhysRevB.111.115141}, information-theoretic quantities~\cite{fan2023diagnostics}, and separability conditions~\cite{chen2023separability} have emerged as promising order parameters.
In parallel, the divergence of the Markov length has been proposed as a hallmark of transitions between phases defined via local Lindbladian reversibility~\cite{Sang2024}. 

However, a comprehensive framework is still missing. The main hurdle is to specify principled constraints on mixed states that allow a meaningful classification. Equally uncertain is which properties of pure-state topological order would survive under decoherence. 
Do anyonic braiding statistics persist, and can generalized fusion rules--or viable analogues of the modular $S$ and $T$ matrices in two dimensions--still be defined? Finally, a complete theory must explain how information storage capacity--whether classical, quantum, or hybrid--naturally emerges across different mixed-state phases.

In this work, we propose a novel systematic framework for topological mixed states.
By generalizing the ideas from the entanglement bootstrap~\cite{Shi2019fusion}, we state a set of axioms that identify fixed points of topological mixed states, and explore invariants to label and classify them. 
Furthermore, by relaxing the axioms, we define mixed-state phases of matter where the equivalence class is defined by axioms under coarse-graining. This new approach naturally provides the notion of topological channel connectivity, which resolves issues in previous attempts to define the equivalence class using two-way connectivity. 
Our framework not only delineates the fundamental constraints that topologically ordered mixed states must satisfy but also clarifies how such states can encode and preserve information--whether quantum, classical, or both. We thus offer the first coherent step toward a classification of topological mixed states.

The paper is organized as follows. Sec.~\ref{sec:mixed-bootstrap} formulates the mixed state bootstrap, including the axioms and definition of fixed point. We introduce the information convex set and its connection to memory capacity, relating its bound to topological entropy. Sec.~\ref{sec:examples} presents explicit fixed point examples and analyzes their classical and quantum memories. For a class of classical fixed points, we discuss the method to prepare them through a local quantum channel. Sec.~\ref{sec:MS-phases} defines mixed state topological phases through coarse-graining. We provide extensive numerical results demonstrating that axioms hold away from fixed points under coarse-graining. We also introduce a refinement of a previous approach based on two-way channel connectivity. Sec.~\ref{Sec:robust-TE} refines this phase equivalence via coarse-graining by demonstrating that topological invariants are robust under approximate axioms.
Sec.~\ref{sec:Secret-sharing} discusses a multipartite secret-sharing property unique to certain non-Abelian mixed-state fixed points, exposing a hierarchy of topological \emph{classical} memories.   
Sec.~\ref{sec:3D} extends the framework to higher dimensions, and Sec.~\ref{sec:conclusion} summarizes with future directions.

\section{Mixed state bootstrap}\label{sec:mixed-bootstrap}

\begin{figure}[!t]
     \centering   \includegraphics[width=1\linewidth]{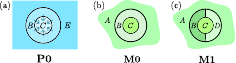}
     \caption{ {\bf Bootstrap axioms}. Given a mixed state $\sigma$, we impose axioms (a) {\bf P0}, (b) {\bf M0}, and (c) {\bf M1} on it. Axiom {\bf P0} is imposed on the entire system, while $BC$ remains local. The axioms {\bf M0} and {\bf M1} are stated on local patches indicated by green regions.}
     \label{fig:axioms}    
\end{figure}

\subsection{Axioms}

We study bosonic lattice systems whose total Hilbert space factorizes into a tensor product of finite-dimensional local Hilbert spaces. To isolate intrinsically topological features, we assume the absence of any global symmetry.

Our goal is to formulate mixed-state bootstrap axioms for the fixed-point density matrix of topological mixed states, as an analog of the \emph{reference state} in the pure-state bootstrap framework~\cite{Shi2019fusion}. The relevant information-theoretic measures are the mutual information $I(A{\,:\,}C)$ and conditional mutual information $I(A{\,:\,}C|B)$ defined for non-overlapping regions $A,B,C$ as
\begin{align}
    I(A{\,:\,}C)_\rho &= S(\rho_A) + S(\rho_C) - S(\rho_{AC}) \\
    I(A{\,:\,}C|B)_\rho &= I(A{\,:\,}BC)_\rho - I(A{\,:\,}B)_\rho,
\end{align}
where $S(\rho):=-\Tr(\rho \log \rho)$ is the von Neumann entropy of density matrix $\rho$. These quantities capture elementary forms of multipartite correlations~\cite{nielsen2010quantum}. 
Given a physical system, consider a large ball partitioned in a way described in \figref{fig:axioms}. A density matrix $\sigma$ is called a fixed point of topological mixed states if it satisfies the following axioms:

\begin{axiombox}
\begin{axiom}[{\bf P0}: Local recoverability] 
    $I(E{\,:\,}C|B)_\sigma =0$ for \figref{fig:axioms}(a), where $E$ is the complement of $BC$. The condition is the saturation of strong subadditivity~\cite{lieb1973proof}, which is the condition for the existence of a map that only acts on $B$ to recover $C$ when $C$ is corrupted or removed~\cite{PETZ_2003,Fawzi2014,Junge2015}, i.e.,  $\rho_{BCE}= \calE_{B\to BC}(\rho_{BE}$) for some quantum channel $\calE_{B\to BC}$.
    The condition enforces the state to have a zero  Markov length mixed state.  
\end{axiom}    
\end{axiombox}

{\bf P0} implies an important aspect of a topological mixed state: the non-topological (local) structure of the underlying mixed state can always be recovered locally, and thus local behavior should not contribute toward the properties we care about. 
This requirement is consistent with the fact that $I(E:C|B)$ diverges at a critical point separating two mixed state phases~\cite{Sang2024}. By imposing the axiom, we focus our attention on fixed points located far away from the phase boundary.

Furthermore, this condition enforces the absence of the novel type of spontaneous symmetry breaking that can occur in the mixed state systems~\cite{Lee_2023, StW_U(1)}, which was recently dubbed as strong to weak spontaneous symmetry breaking (SWSSB)~\cite{sala2024spontaneousstrongsymmetrybreaking, lessa2024strongtoweakspontaneoussymmetrybreaking, gu2024spontaneoussymmetrybreakingopen, kim2024errorthresholdsykcodes, zhang2025strongtoweakspontaneousbreaking1form}. Analogous to an ordinary spontaneous symmetry breaking, 
in SWSSB a quantum channel that breaks a strong symmetry to a weak one \cite{Bu_a_2012, PhysRevA.89.022118}---whether applied either infinitesimally across the entire system or finitely on a local patch---thermalizes the conserved-charge sector and induces a global change in the density matrix state, behavior incompatible with {\bf P0}. More details for the motivation of this axiom is explained in Appendix~\ref{app:info-robustness}.

Formally, if we consider a channel $\calE_C$ acting on a local patch $C$, the condition together with the monotonicity of the conditional mutual information implies that $0 = I(E:C|B)_{\rho} \geq I(E:C|B)_{\calE_C[\rho]} = 0$. This, in turn, implies that the difference between the original and perturbed density matrices of the entire system, captured by the relative entropy~\cite{PETZ_2003}, is given as
\begin{align}
    D(\rho_{EBC} \Vert \calE_C[\rho_{EBC}]) = D(\rho_{BC} \Vert \calE_C[\rho_{BC}]),
\end{align}
which is only a function of the compact region $BC$.
Accordingly, with {\bf P0}, the local channel acting on $C$ cannot induce a change that is intrinsically global in the density matrix. More details can be found in Appendix~\ref{app:relative-S}.

The next two axioms can be stated on local patches of the density matrix. 
\begin{axiombox}
    \begin{axiom}[{\bf M0}: No long-range order] 
    $I(A{\,:\,}C)_\sigma =0$ for \figref{fig:axioms}(b), where $A$ is the complement of $BC$ on the local patch. This condition implies that there is neither classical nor quantum correlation between separated regions.  
    \end{axiom} 
\end{axiombox}

{\bf M0} is physically natural: it eliminates states that exhibit long-range order, such as those with spontaneous symmetry breaking.
Furthermore, the condition precludes some pathological situations that arises in a mixed state setup, such as the classical mixture $\rho = \sum_i p_i \rho_i$ with $p_i > 0$, where each $\rho_i$ realizes a different topological phase. Therefore, through {\bf P0} and {\bf M0}, we can exclude pathological density matrices.

\begin{axiombox}
    \begin{axiom}[{\bf M1}: Uniformity and Smoothness] 
    $I(A{\,:\,}C|B)_\sigma \,{=}\,0$ for \figref{fig:axioms}(c), where $A$ is the complement of $BCD$ on the local patch. The condition excludes domain walls between two different phases, which can satisfy both {\bf P0} and {\bf M0}. Thus, it should be included if we wish to say the given state is uniform enough to represent a single mixed-state phase. Without it, one risks including pathological cases where two different phases are simply juxtaposed.  
\end{axiom}
\end{axiombox}

{\bf M1} is quite subtle, yet extremely important. 
For pure states, it implies the absence of defects in the wavefunction, which allows one to discuss the physics from the topology of the manifold rather than local details~\footnote{By defects, we mean both point-like defects in the sense of Refs.~\cite{Bombin2008,Brown2013} and line defects, such as domain walls~\cite{Beigi2011,Shi2021-DWTEE,Buican2025}; we discuss these further in Sec.~\ref{sec:MS-phases}.}. 
In other words, the axiom requires a \emph{smoothness} of the fabric of many-body correlation or entanglement, the quintessential premise when it comes to the discussion of topological states.  
On the other hand, a finite-depth local channel can violate {\bf M1} since channels can break the entanglement. As we shall see, the violation of {\bf M1} also implies that the topological entropy (TE) can change abruptly.

One potential concern is that highly fine-tuned finite-depth local unitaries (FDLUs) can generate \emph{spurious} entanglement that violates {\bf M1} on specially chosen partitions~\cite{Zou2016,Santos2018,Williamson2019,Stephen2019,Kato2020}. Two points address this. $(i)$ A generic FDLU does not violate {\bf M1} on sufficiently large local patches. $(ii)$ Even when such violations occur, they do not obstruct phase identification: we pass to a \emph{clean} representative obtained by a shallow-depth local unitary that removes the axiom violations and then evaluate phase diagnostics on the clean state. A soft-invariant theorem~\cite{Bowen2025} ensures that, within a given phase, all clean representative states related by FDLUs yield identical diagnostic values. In particular, the Levin-Wen topological entropy $\gamma_{\textrm{LW}}$ of any clean representative is a global minimum over its FDLU orbit~\cite{Kim2023-TEE-bound}. Hence, for any two states $\rho_0$ and $\rho_1$ in the same phase, their clean counterparts $\rho'_0$ and $\rho'_1$ have the same topological entropy $\gamma_{\textrm{LW}}(\rho'_0)\,{=}\,\gamma_{\textrm{LW}}(\rho'_1)$; for relevant discussions, see \secref{Sec:robust-TE}. Therefore, spurious contributions do not impede a well-defined notion of mixed-state phases.

\begin{axiombox}
    \begin{definition}[Topological mixed state fixed point]\label{def:fp_topological}
        A density matrix satisfying three axioms {\bf P0}, {\bf M0}, and {\bf M1} for any partition above length scale $l$ is called a \emph{fixed point} of a topological mixed state. In particular, if these axioms hold for all shapes of partitions at length scale $l_0$, they hold for length scale $l\geq l_0$~\cite{Shi2019fusion}. 
    \end{definition}
\end{axiombox}

We emphasize that the novelty of this work does not lie in the individual entropy conditions {\bf P0}, {\bf M0}, or {\bf M1}; as we elaborate above, these are well-understood physical conditions one expect from fixed points of topological mixed states. 
Rather, the novelty is to show that the three conditions taken together furnish a coherent and self-contained axiomatic framework for mixed-state topological phases. In particular, this framework underlies our definitions of fixed-point states [Def.~\ref{def:fp_topological}] and phase equivalence [Def.~\ref{def:p.f.-equivalence}], and supports a consistent classification theory for mixed states.

\vspace{5pt} It is worth comparing these mixed state axioms with those of the entanglement bootstrap program for gapped ground states~\cite{Shi2019fusion}, where we have two axioms {\bf A0}: $S(BC) + S(C) - S(B)\,{=}\,0$ and {\bf A1}: $S(BC) + S(CD) - S(B) - S(D)\,{=}\,0$, for the partition of a ball into $BC$ or $BCD$ as shown in Fig.~\ref{fig:axioms}(b) and (c) respectively. The first two axioms {\bf P0} and {\bf M0} follow from {\bf A0}~\footnote{Using purification and strong subadditivity, one can prove that $S(BC) + S(C) - S(B) \geq I(A\,{:}\,C), I(A\,{:}\,C|B)$.}, and {\bf M1} follows from {\bf A1}. However, these mixed-state conditions are strictly weaker because they may hold on volume law states, whereas the entanglement bootstrap axioms for gapped ground state imply an area law. 
The Axiom {\bf A0} is known as local purity condition, this is because it implies that the state in $C$ can be purified with a subspace of $B$; see Appendix~\ref{app:local-pure}. From now on, we shall refer to entanglement bootstrap with {\bf A0} and {\bf A1} as ``pure-state bootstrap".  

Furthermore, these three axioms are independent: 
$(i)$ {\bf P0} and {\bf M0} together do not imply {\bf M1}, e.g. a state with non-Abelian anyons located at $C$. 
This is because adding an anyon $a$ to a disk will increase the von Neumann entropy of the disk by $\log d_a$ \cite{KitaevPreskill_2006}, where $d_a$ is the quantum dimension of $a$. Suppose we create an anyon pair with a unitary string operator, one anyon in $A$ and another in $C$. Then the violation of {\bf M1} will be $2 \log d_a$ because $S_{AB}$ and $S_{BC}$ are increased by $\log d_a$ each without affecting $S_B$ and $S_{ABC}=S_D$. Note that only non-Abelian anyons give a violation, because $d_a=1$ for Abelian anyons.
$(ii)$ {\bf P0} and {\bf M1} together do not imply {\bf M0} by the example of an incoherent mixture of distinct SSB states, such as $\rho = \frac{1}{2}( |\bar{0} \rangle \langle \bar{0} | + |\bar{1} \rangle \langle \bar{1} |)$, where $|\bar{a} \rangle = |a \rangle^{\otimes N}$. $(iii)$ {\bf M0} and {\bf M1} together do not imply {\bf P0} by the example of the state $(1+\prod_{i=1}^N X_i)/2^{N}$. This state reduces to the maximal mixed state after tracing out one site.

These mixed-state axioms are naturally satisfied by pure states satisfying pure-state axioms, and thus pure-state fixed points are naturally incorporated into our framework. Furthermore, volume-law entangled decohered states can satisfy mixed-state axioms, as we will see in the example section (Sec.~\ref{sec:examples}). More physically, we allow small errors to the axioms as long as the errors decay as we zoom out to larger length scales. For a more careful discussion, see Sec.~\ref{sec:MS-phases}.

\subsection{Information convex set}

Given two density matrices $\rho_1$ and $\rho_2$ satisfying mixed-state axioms, how can we distinguish them?
The first step is to characterize their information-theoretic properties, which can be achieved by investigating the structure of an information convex set, a set of mixed states locally equivalent to each density matrix satisfying the axioms. 
Information convex sets can be defined on the entire system as well as on pieces~\cite{Shi2019fusion,Shi2020domainwall,Huang2021knots} (including embedded and locally embedded regions \cite{Shi2024-figure8}). 
For pure states, the information convex set of a closed manifold simply corresponds to the ground state manifold for the corresponding parent Hamiltonian~\cite{Bravyi2010,Shi2019ICS,strict-H-2024}. For mixed states, the information convex set corresponds to the set of mixed states that can be prepared in a similar manner yet encode different topological information (\figref{fig:Bloch_Polytope}).

\begin{figure}[!t]
     \centering   \includegraphics[width=1\linewidth]{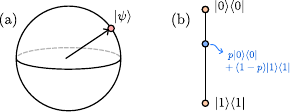} 
     \caption{ {\bf Information convex set}. For topological mixed states encoding (a) a single qubit (quantum memory) and (b) a single bit (classical memory). While quantum memories reside on the surface of the Bloch sphere, classical memories reside in the linear combination of extreme points. Mathematically, the convex sets in (a) is the convex hull $ \conv (\{|\psi\rangle \langle \psi|\, |\, |\psi\rangle \in \mathcal{H}\})$, and the convex set in (b) is $\conv(\{|0\rangle \langle 0|, |1\rangle \langle 1|\})$.} 
     \label{fig:Bloch_Polytope}    
\end{figure}

\begin{definition}[Information convex set] \label{def:ICS-mixed}
    Given a fixed-point state $\sigma$ under Def.~\ref{def:fp_topological} defined on a physical system $N$. The information convex set on a manifold $\Omega\,{\subseteq}\,N$, possibly with boundary, is defined as
 \begin{equation} \label{eq:ICS}
     \Sigma(\Omega|\sigma) = \{ \rho_\Omega \,| \, (1), (2) \}, \quad \text{where}
 \end{equation}
 \begin{itemize}[leftmargin=15pt]
     \item [(1)] \emph{Local indistinguishability}. For any bounded radius ball $b\subset \Omega$, we require $\rho_{b}\,{=}\,\sigma_b$.  
     \item [(2)] \emph{Local recoverability}. {\bf P0} is satisfied for $\rho_\Omega$. When $\Omega$ has a boundary, we require $I(E\,{:}\,C|B)_\rho =0$ for configurations in Fig.~\ref{fig:ICS-def-short}. This allows the partition to be located near the boundary.
\end{itemize}
These two conditions imply that distinct elements of the information convex set cannot be converted into each other through local operations. Furthermore,  
$\rho \in \Sigma$ must satisfy {\bf M0} and {\bf M1} because the state is identical to the fixed-point state locally.
This indicates the topological nature of the information that we can encode using states in the information convex set.

When $\Omega$ is a closed manifold, the proof that $\Sigma(\Omega|\sigma)$ is a compact convex set can be found in appendix~\ref{app:convex}. The role of $\mathbf{P0}$ in the definition of information convex set is explained in appendix~\ref{app:info-robustness}.

When $\Omega$ is a manifold with boundary embedded in a larger manifold $N$, $\Sigma$ is still a compact convex set; see \appref{app:ICS-Omega}. More importantly, given a fixed-point state $\sigma$ on $N$, one can show that $\Sigma(\Omega|\sigma) \cong \Sigma(\Omega'|\sigma)$ whenever $\Omega$ can be smoothly deformed into $\Omega'$ on the background $N$. This isomorphism theorem, further elaborated in \eqnref{eq:iso-thm}, is central to the advanced application of our framework.
\end{definition}

 \begin{figure}[!t]
     \centering    
     \includegraphics[width=0.94\linewidth]{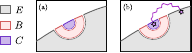}
     \caption{ {\bf Manifold with boundary}. (a) Partitioning near the boundary for the definition of information convex set (Def.~\ref{def:ICS-mixed}). (b) For a subsystem with a boundary, a nontrivial entanglement between boundary degrees of freedom can exist. Imposing $I(E:C|B)\,{=}\,0$ precisely exclude this complication.}
     \label{fig:ICS-def-short}
 \end{figure}

\subsection{Channel}
 
In the study of mixed states, the generalization of local perturbations or unitary transformations is a local quantum channel.
\begin{definition}[Finite depth local channel]
    A quantum channel $\mathcal{E}$ acting on a lattice system is called a finite-depth local channel if it can be written as a finite composition of layers of local quantum channels, each consisting of commuting, spatially disjoint operations. More precisely, $\mathcal{E}$ is said to be of depth $D$ if \begin{align} 
        \mathcal{E} = \mathcal{E}^{(D)} \circ \mathcal{E}^{(D-1)} \circ \cdots \circ \mathcal{E}^{(1)}, 
    \end{align} 
    where each $\mathcal{E}^{(j)} = \bigotimes_{x \in \Lambda_j} \mathcal{E}_{x}^{(j)}$ is a layer of local channels supported on disjoint regions (e.g., finite balls around lattice sites), and $\Lambda_j$ is a subset of lattice sites such that the supports of $\mathcal{E}_{x}^{(j)}$ and $\mathcal{E}_{x'}^{(j)}$ are disjoint for $x \neq x'$. See \figref{fig:channel} for illustration.
\end{definition}

Finite-depth local channels (FDLCs) provide a physically motivated model for local quantum processes, such as interactions with an environment, local noise, and decoherence protocols. Since local channels can \emph{break} entanglement~\cite{Horodecki_2003}, the action of an FDLC can, though need not, drive a state toward classicality. From the perspective of state preparation, FDLCs correspond to a restricted set of operations: they allow only local operations and local classical communication. This is more stringent than the standard LOCC (local operations and classical communication) framework, which permits unrestricted classical communication and plays a central role in entanglement theory. In this regard, we introduce the following notion:

\begin{figure}[!t]
\centering
\begin{tikzpicture}
\definecolor{myred}{RGB}{240,83,90};
\definecolor{myblue}{RGB}{73,103,189};
\definecolor{myturquoise}{RGB}{83,195,189};
\definecolor{mydo}{RGB}{204, 204, 255};
\definecolor{mychannel}{RGB}{255, 204, 204};
\filldraw[fill=mydo, draw=white] (-0.6, -3.1) rectangle (5.6, -2.3);
\draw[black!80, line width=1.0] (-0.6,-2.7) -- (5.6,-2.7);
\foreach \x in {0,1,...,5}{
\draw[black!80, line width=1.0] (\x+0.1,-2.5) -- (\x+0.1, -0.3);
\draw[black!80, line width=1.0] (\x-0.1,-2.5) -- (\x-0.1, -0.3);
\filldraw[fill=white, draw=black!80, line width=1.0] (\x-0.25,-2.95) rectangle (\x+0.25,-2.45);
\foreach \x in {0,2,4}{
\filldraw[fill=white] (\x-0.2,-1.7) rectangle (\x+1.2,-2.1);
\filldraw[fill=mychannel, draw=black, line width=1] (\x-0.2,-1.7) rectangle (\x+1.2,-2.1);
\filldraw[fill=white] (\x-0.2,-0.5) rectangle (\x+1.2,-0.9);
\filldraw[fill=mychannel, draw=black, line width=1] (\x-0.2,-0.5) rectangle (\x+1.2,-0.9);
}
\foreach \x in {-1,1,3,5}{
\filldraw[fill=white] (\x-0.2,-1.5) rectangle (\x+1.2,-1.1);
\filldraw[fill=mychannel, draw=black, line width=1] (\x-0.2,-1.5) rectangle (\x+1.2,-1.1);
}
\filldraw[fill=white, draw=white] (-1.3, -1.6) rectangle (-0.5, -0.7);
\filldraw[fill=white, draw=white] (5.5, -1.6) rectangle (6.3, -0.7);
\node[black] at (-0.8, -2.7) {$\rho$};
\draw[<->,>=stealth] (5.7,-2.1) -- (5.7,-1.3) node[right]{${\cal E}_{\textrm{FDLC}}$} -- (5.7,-0.5);
}
\end{tikzpicture}
\caption{{\bf Finite depth local channel (FDLC)}. 
For each site, we have two physical indices corresponding to bra and ket. The red block represents a local quantum channel acting on two sites. The entire channel has a finite depth of three. }
\label{fig:channel}
\end{figure}

\begin{definition}[Classically preparable]
Consider a state $\sigma$ such that  ${\calE}(\rho_0) = \sigma$ where ${\calE}$ is a FDLC. If $\rho_0$ is a classical mixture of product states, i.e., 
\begin{align}
    \rho_0 = \sum_a p_a (\otimes_i \rho_i^a),
\end{align}
then $\sigma$ is classically preparable. Here the probability $\{ p_a \}$ mediates classical correlations.  
\end{definition}
 
Any product state $\rho = \otimes_i \rho_i$ is a classically preparable state, where the channel $\calE$ is given as the tensor product of an on-site channel $\calE_i$ that resets to a specific local density matrix $\rho_i$. 
On the other hand, the toric code ground state is not classically preparable because starting from a product state, one has to perform local measurements and non-local communication for error corrections.

\subsection{Memory capacity and its bound by the topological entropy}

How much classical or quantum information can be stored in a closed manifold $M$? In the context of pure state topological order, this is measured by the dimension of the degenerate ground state.
To quantify the ability to store (classical or quantum) topological information of the system $M$, we define the \emph{memory capacity} of a mixed state $\sigma$ as

\begin{equation}\label{memory_cap}
    \calQ(M|\sigma) \equiv \hspace{-5pt} \max_{\rho,\lambda \in \Sigma(M|\sigma)} \hspace{-5pt} S(\rho) - S(\lambda).
\end{equation}

Since two density matrices from the same information convex set are locally indistinguishable, the entropy difference can only arise from non-local topological information. The memory capacity of the sphere $S^2$ is always zero, as we show in Appendix~\ref{app:S2}. Thus, we take the torus $M=\mathbb{T}^2$, as the minimal example.

To better connect this idea with a well-known quantity, we consider the topological entropy defined by
\begin{equation} \label{eq:TE}
  \gamma_{\textrm{LW}} \equiv  \frac{1}{2} I(A\,{:}\,C|B)_\sigma \quad \text{for} \quad \vcenter{\hbox{\includegraphics[width=0.08\textwidth]{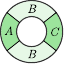}}} 
\end{equation}
where the annulus $ABC$ is embedded in a disk. Since this topological contribution can be non-zero even for entirely separable mixed states, we call this more general quantity \emph{topological entropy} (TE) following Ref.~\cite{Castelnovo_2007}. 
For pure topological states, $\gamma_{\textrm{LW}}$ is exactly the Levin-Wen topological entanglement entropy~\cite{LevinWen2006}~\footnote{In the context of the mixed-state phase, the Levin-Wen topological entropy is well-defined due to axiom {\bf M1}. On the other hand, the Kitaev-Preskill topological entanglement entropy~\cite{KitaevPreskill_2006} stops being a topological invariant in the mixed-state context; see Appendix~\ref{app:non-topo-invariant}}.  


While we shall mainly talk about case $M=\mathbb{T}^2$, we make an important remark on the annulus. The memory capacity of the information convex set of a contractable annulus $X$ (i.e., annulus contained in a disk) equals twice the topological entropy. Namely,
\begin{equation}
    \calQ(X|\sigma)= 2 \gamma_{{\rm LW}}, 
\end{equation}
for $X=ABC$ being the Levin-Wen partition.
We shall explain this exact property in the Appendix~\ref{app:annulus-and}, Prop.~\ref{prop:annulus-TE}.

\begin{Proposition}[torus capacity]\label{prop:torus-capacity}
    The capacity of a torus $M=\mathbb{T}^2$ is upper bounded by four times the Levin-Wen topological entropy as
    \begin{equation}\label{eq:bound-mixed}
        \calQ(\mathbb{T}^2) \le 4 \gamma_{\textrm{LW}}.
    \end{equation}
\end{Proposition} 
The proof of this (and a more powerful one on any closed manifold) will be given in Appendix~\ref{app:bound-C-TEE}. We will see examples of mixed state fixed points saturating this bound in Sec.~\ref{sec:examples}. The bound \eqref{eq:bound-mixed} should be contrasted with a bound by Kim~\cite{Kim2013storage} for pure state topological order, which says $\calQ(M) \le 2\gamma_{\textrm{LW}}$. The violation of the pure-state bound for some mixed-state fixed points originates due to the absence of the non-degeneracy braiding statistics~\cite{Braiding2023}, which is a quantum mechanical effect. 

\subsection{Comparison to pure-state bootstrap}

The bootstrap axioms developed in this section encompass topological mixed and pure states. While three ``mixed-state axioms'' {\bf P0}, {\bf M0}, and {\bf M1} are strictly weaker than ``pure state axioms'' {\bf A0} and {\bf A1} introduced in \cite{Shi2019fusion}, many of the intuitions and machinery carry over. 
\begin{itemize}[leftmargin=13pt]
    \item {\bf Fixed point interpretation}. These axioms, defined on (or centered at) local balls, imply that the same conditions hold true on larger length scales by strong subadditivity~\cite{lieb1973proof,PETZ_2003}. Think about coarse-graining as the renormalization group (RG) flowing to large length scales. This intuition agrees with the fixed point interpretation of such states. This leads to the formal definition of fixed point in Def.~\ref{def:fp_topological}.
    
    \item {\bf Uniformity}. The axiom {\bf M1} (which is an analog of {\bf A1}) implies an emergent sense of uniformity, which is, in the deep level related to the isomorphism property of the information convex set under a smooth deformation of region, as we review in Appendix~\ref{app:ICS-Omega}.

    \item {\bf Topological entropy}. For both cases, the Levin-Wen topological entropy can be shown to be invariant under the deformation of regions. {We showed this is true for states with decaying errors of the axioms; see Sec.~\ref{Sec:robust-TE}}.

    \item {\bf From mixed to pure:}  A mixed state fixed point which also satisfies the local purity condition {\bf A0} will be a pure state fixed point (see Appendix~\ref{app:local-pure}):
    \begin{equation}
        \text{\bf A0, M1} \quad \Rightarrow \quad \text{\bf A0, A1}.
    \end{equation}  

\end{itemize}
\noindent 
However, there are several distinctions.
\begin{itemize}[leftmargin=13pt]
\item The pure state axioms \emph{guarantee} the information convex set on a closed manifold to host pure quantum states on some Hilbert space; in other words, such states are fully coherent, can host topological quantum information.  On the other hand, the mixed state axioms do not; in particular, the mixed state axioms allow the information convex set to host classical states as illustrated in \figref{fig:Bloch_Polytope}(b), implying that there exists \emph{classical topological information}. See Prop.~\ref{prop:convexset} for a concrete example.

\item When the pure state axioms are satisfied, both the Levin-Wen and Kitaev-Preskill topological entanglement entropy can be shown to be quantized \cite{Shi2019fusion} according to the fusion rules.  
In contrast, we do not yet know if Levin-Wen TE is quantized following the mixed state axioms. Furthermore, the Kitaev-Preskill topological entropy is no longer a topological invariant for mixed states fixed points; see Appendix~\ref{app:non-topo-invariant}. 

\item The information convex set of the annulus can be rigorously shown to be a simplex with orthogonal extreme points, under the pure state axioms~\cite{Shi2019fusion}. It is an open question if this ``simplex theorem'' is still true for an arbitrary mixed state fixed point.  

\item The braiding non-degeneracy~\footnote{It means that the only anyon that has trivial braiding statistics (`transparent') with all non-trivial anyons is the vacuum.} of anyon systems can be derived under the pure state axioms \cite{Braiding2023,Shi2019Verlinde}. However, when the mixed state axioms are satisfied, braiding may become degenerate, as happens in examples; see Ref.~\cite{Ellison2025}. This is related to a special property of these mixed state fixed point wave functions detected by an immersed punctured torus; see Appendix~\ref{app:phase-equiv-coherent}. 

\item A modular commutator~\cite{Modular-commutator-Gapped} is not an invariant anymore in the mixed state. This is because there exists a fixed point density matrix with a non-uniform value of modular commutator. In particular, a nonzero modular commutator can appear in a state in the same phase as the trivial phase; see \appref{app:non-topo-invariant}.

\item Finally, classification of phases is not only about finding fixed points, but also about their equivalence; the phases are fixed points modulo the equivalence condition. Due to the distinction of the axioms, the meaning of equivalence will change accordingly. See Def.~\ref{def:p.f.-equivalence} for our sense of equivalence. 
\end{itemize}

\section{Examples of fixed points}\label{sec:examples}

In this section, we identify various solvable fixed points by Def.~\ref{def:fp_topological} and calculate the corresponding information convex set $\Sigma$, which demonstrates their classical/quantum memories.

 A useful way to show whether the state satisfies axioms is by showing the existence of recovery channel ${\cal R}: B({\cal H}_B) \mapsto B({\cal H}_{BC})$ such that $(\textrm{Id}_A \otimes {\cal R})[\rho_{AB}] = \rho_{ABC}$, which is equivalent to the condition $I(A\,{:}\,C|B)\,{=}\,0$~\cite{PETZ_2003}. Here $B(\calH)$ is a set of bounded operators acting on $\calH$.

In general, showing the existence of a Petz recovery map can be cumbersome. But for stabilizer states, we have a simple argument. Consider a stabilizer state $\rho_B = \Tr_C \rho_{BC}$ on some region $B$:
\begin{equation} 
    \rho_B = \prod_{S_i \in \calS_B} \frac{1+S_i}{2}, 
\end{equation} 
where $\calS_B\,{=}\,\{S_i \}$ is a set of stabilizer generators supported on $B$. Assume $\rho_{BC}$ supports additional stabilizer generators $\{ \bar{S}_j \}$ that commute with $\{ S_i \}$. We can have a recovery map as follows. For simplicity, assume we have a single extra stabilizer $\bar{S}$ on $BC$, where its eigenvalue is $1$. Let $\calP_{\bar{S}}^\pm$ be the projector onto the eigenspace of $\bar{S}$ with eigenvalue $\pm 1$. Then the recovery channel $\calR$ is given as: 
\begin{align} \label{eq:recoverychannel}
    \calR[\rho_B] &= {\cal P}_{\bar{S}}^+(\rho_B \otimes \frac{\Id_C}{d_C} ) {\cal P}_{\bar{S}}^+ + U^\vdagger {\cal P}_{\bar{S}}^- (\rho_B\otimes\frac{\Id_C}{d_C} ) {\cal P}_{\bar{S}}^- U^{\dagger} \nonumber \\ 
    & = (1+{\bar{S}})(\rho_B \otimes \frac{\Id_C}{d_C} ) = \rho_{BC},
\end{align}
where $U^\vdagger$ is some unitary acting on extra qubits such that $U^\vdagger {\bar{S}} U^\dagger = -{\bar{S}}$ and $d_C$ is the Hilbert space dimension associated with region $C$. It is straightforward that its Kraus representation satisfies the normalization condition. Even if we have multiple commuting stabilizers on $BC$, it should always be possible to apply a certain unitary that makes their signs correct since they are independent.
Finally, if $\rho_{ABC}$ does not contain any stabilizer that extends across $A$ and $C$, then this recovery map $\calR$ would recover $\rho_{ABC}$ from $\rho_{BC}$.

\begin{exmp}[Product state]
Consider a generic product state $\rho = \otimes_i \rho_i$; here $\rho_i$ can be either a pure or mixed local state. In this case, the subadditivity is saturated and we have $S(\rho_{AB})\,{=}\,S(\rho_A) + S(\rho_B)$. Accordingly, all mutual information vanish and the three axioms are satisfied. Maximally mixed states with volume-law entropy belong to this scenario. Because product states have zero topological entropy $\gamma_{\textrm{LW}}=0$, the bound in Eq.~\eqref{eq:bound-mixed} is saturated, and the memory capacity $\calQ(M|\rho)=0$ for any closed manifold $M$. This further implies that the information convex set $\Sigma(M|\rho)=\{\rho\}$, i.e. no information can be stored.
\end{exmp}

\begin{exmp}[Ground states of topological order]
It is well-known that the ground state of nonchiral topological orders has fixed points that exactly satisfy the axioms. Such representative wave functions can be the string-net wave functions \cite{Levin2005,LevinWen2006}.
In particular, on the torus, $\Sigma(\mathbb{T}^2|\sigma)$ is the state space of a finite-dimensional Hilbert space, where the Hilbert space dimension equals the number of anyons of the system. \begin{equation}\label{eq:gs-ICS}
    \Sigma(\mathbb{T}^2|\sigma) = \conv \big\{ |\psi\rangle\langle \psi| \,| \,|\psi\rangle  = \sum_a c_a |\psi_a\rangle \big\},
\end{equation}
where $|\psi_a\rangle $ are minimal entangled states \cite{Zhang2012} on the torus, in one-to-one correspondence with anyon types, and $\sum_a |c_a|^2 =1$. One such model is the ground state of the toric code \eqref{eq:t.c.}, for which $\Sigma(\mathbb{T}^2|\sigma)$ is identical to the 4-dimensional Hilbert space. Ground state of chiral topological orders can be treated as fixed points in the same manner, if we are allowed to neglect finite size errors of the area law that decays with subsystem sizes~\cite{confrmal2025,strict-J-2024,2024chiral-Vira}.
\end{exmp}

\begin{exmp}[Dephased toric code]\label{exmp:toric-code-fixed-point}
The discussion here immediately generalizes to $\mathbb{Z}_n$ toric code. 
For the toric code, the density matrix for maximally mixed logical states is given as
\begin{align}\label{eq:t.c.}
    A^{\bm{1}}_v &:= (1+\prod_{e \ni v} X_e)/2 \\
    B_p^{\bm{1}} &:= (1+\prod_{e \in p} Z_e)/2 \\
    \rho^{\rm tc} &= \prod_v A^{\bm{1}}_v \prod_p B^{\bm{1}}_p.
\end{align}
Consider a mixed state dephased under maximal $X$ error (note that the other option is mathematically equivalent due to $e\leftrightarrow m$ exchange duality):
\begin{equation}\label{eq:rho^e-T.C.}
    \rho^{e} = \prod_v A^{\bm{1}}_v.
\end{equation}
This model has been considered in the literature of mixed state phases~\cite{fan2023diagnostics, Lee_2023, bao2023mixedstate, Wang_intrinsicmixedstateTO_2023, chen2023separability, Sohal2024}. We show that the state satisfies {\bf P0}, {\bf M0}, and {\bf M1} and thus it is a valid fixed-point under the bootstrap axioms.
The regions $A,B,C$ (and $D$) of interest are recalled in Fig.~\ref{fig:toric-code-cases}.  

\begin{figure}[!h]
    \centering
    \includegraphics[width=0.99\linewidth]{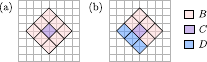}
    \caption{Cases to check for (a) {\bf M0} and {\bf P0}, and (b) {\bf M1}. In both cases, a local region $C$ is separated from the rest by an annulus ($B$ or $BD$). Here we only show one of the several configurations to check {\bf M1}. See Appendix~\ref{app:discrete-axioms} for a complete list.} 
    \label{fig:toric-code-cases}
\end{figure}

\vspace{5pt} \noindent (1) Proof of {\bf M0}: Consider tripartite region $ABC$ such that $B$ separates $A$ and $C$. Here a region $B$ can be denoted using a collection of vertices and qubits supported on the edges connecting within this collection. Then, we have
\begin{equation} \label{eq:traceout}
    \rho_{AC} = \hspace{-13pt} \prod_{\{v|{N(v)\notin B}\}} \hspace{-13pt}  A^{\bm{1}}_v\left(\frac{1+\bm{X}_{\partial B_A} }{2}\right) 
    \left(\frac{1+\bm{X}_{\partial B_C} }{2}\right), 
\end{equation}
where $N(v) = \left\{e\,|\, e \ni v \right\}$, and $\partial B$ denotes a set of edge qubits whose one end belongs to $B$ but not the other end; note that $\partial B$ decomposes into two disjoint sets $\partial B_A$ and $\partial B_C$ each of which is supported on $A$ and $C$, separately.
Note that the last factor on $\bm{X}_{\partial B_C}$ is redundant as it can be absorbed into $\prod A_v^{\bm{1}}$. $\rho_{AC}$ decomposes as $\rho_{AC} = \rho_A\otimes \rho_C$, implying $I(A\,{:}\,C)\,{=}\,0$.  

\vspace{5pt} \noindent (2) Proof of {\bf P0}: Note that if the region $C$ in \figref{fig:toric-code-cases}(a) is traced out, we lose stabilizers that share links with $C$ but keep only the product of them which is entirely supported on $B$. Now, using the aforementioned strategy in \eqnref{eq:recoverychannel}, we can construct local channels just acting on $B$, which recover individual stabilizers. More precisely, we prepare ancillas, measure lost stabilizers, and perform error correction on outcomes. Because $\prod_{e \in \partial A} X_e\,{=}\,1$, we obtain only the desired outcomes.

\vspace{5pt} \noindent (3) Proof of {\bf M1}: 

If we trace out the region $CD$ we get
\begin{equation}  \label{eq:traceout-2}
    \rho_{AB} = \hspace{-13pt} \prod_{\{v|{N(v)\notin B}\}} \hspace{-13pt}  A^{\bm{1}}_v\left(\frac{1+\bm{X}_{ \partial (CD)} }{2}\right)
\end{equation}

By using the protocol explained in {\bf P0}, we can recover $\rho_{ABC}$ by applying a channel that only acts on the support of $B$, which proves $I(A\,{:}\,C|B)\,{=}\,0$. 
\end{exmp}
 
\begin{exmp}[Fermionic decohered toric code]
    Consider maximally $Z_r X_{r+\delta}$-decohered toric code~\cite{Wang_intrinsicmixedstateTO_2023} in the square lattice, where $\delta=(1/2,1/2)$ be a shift vector. It is written as 
    \begin{align}
        \rho_f \propto \prod_p (1+ W_p), \quad W_p:=  \hspace{-8pt}\prod_{e \ni (p-\delta)} \hspace{-8pt} X_e \prod_{e \in p} Z_e.
    \end{align}
    The proof of {\bf M0} is the same as the toric code case; we can always choose a large enough $B$ so that there is no stabilizer supported on both $A$ and $C$, which leads to $I(A\,{:}\,C)\,{=}\,0$. The proof of {\bf P0} is also similar. After tracing out $C$, we will have an additional stabilizer of $ZX$-string at the boundary of $C$, which enforces the measurement outcome of $W_p$ stabilizer supported on $C$ to obey the desired condition $\prod_{p \in C} W_p = 1$. Accordingly, we can perform error correction on $BC$ to recover the original $\rho_{ABC}$ always. The proof of {\bf M1} is also analogous.

\end{exmp}

\begin{exmp}[Stacked layers of fixed points]
    We can get more fixed-point examples by stacking layers of known fixed-point examples. One can immediately check that stacking does not violate any of the axioms. In particular,  stacking a decohered classical state with a quantum ground state is allowed.  
\end{exmp}

\subsection{Abelian quantum double}

Next, we examine some canonical examples of decohered abelian quantum double models, which generalize the toric code example above. 
To proceed, we use the notation from \cite{Bombin2008}, where $A_v^q$ and $B_p^f$ are commuting  \emph{projectors} on a gauge charge $q$ on $v$ and flux $f$ on $p$, respectively~\footnote{Note that when we generalize this to non-Abelian cases, charge and flux labels correspond to irreducible representation and conjugacy classes, respectively.}. For example, in the case of a  $\mathbb{Z}_n$ toric code state, 
\begin{align} \label{eq:DFT}
    A_v^q &:= \frac{1}{n} \sum_{m=0}^{n-1} q^m \prod_{e \ni v} X_e^{m \epsilon_{e,v}} = \calP_{q, \prod_{e \ni v} X_e^{\epsilon_{e,v}}}  
    \nonumber \\
    B_p^f &:= \frac{1}{n} \sum_{m=0}^{n-1} f^m \prod_{e \in p} Z_e^{m \epsilon_{e,p}} = \calP_{f, \prod_{e \in p} Z_e^{\epsilon_{e,p}}}
\end{align}
where we use $\mathbb{Z}_n$ generalization of Pauli operators and $q,f \in  \{ \omega^m\,|\,m \in \mathbb{Z}_n \}$ label local gauge charge and flux, where $\omega = e^{2\pi i/n}$. The result immediately follows from the discrete Fourier transformation.
The exponent $\epsilon_{e,v}$ and $\epsilon_{e,p}$ are defined in such a way that $\partial e = \sum_{v \in e} \epsilon_{e,v}\, v$ and $\partial p = \sum_{e \in p} \epsilon_{e,p}\, e$.
With these notations, the ground state of an Abelian quantum double model is written as $\prod_v A_v^{\bm{1}} \prod_p B_p^{\bm{1}}$. For mixed states, we define

\begin{definition}[Canonically decohered Abelian quantum double] \label{def:canonical}
    We introduce the canonically decohered states of the $\mathbb{Z}_n$ quantum double by applying the maximal Pauli-$X$  (flux) or Pauli-$Z$ (charge) decoherence channel:
    \begin{align} \label{eq:cdqd}
        \rho^e = \prod_v A_v^{\bm{1}}, \quad \rho^m = \prod_p B_p^{\bm{1}}.
    \end{align}
    Their forms follow from the observation that under the maximal $Z$ or $X$ decoherence, any stabilizer that fails to commute with error operators is removed from the density matrix. 
    Due to $e \leftrightarrow m$ duality in $\mathbb{Z}_n$ quantum double, the studies of $\rho^e$ and $\rho^m$ are equivalent. Furthermore, since any finite abelian group decomposes into a direct product of cyclic groups~\cite{dummit2004abstract}, it is sufficient to study this specific example to understand the generic behavior of the abelian quantum double. As we shall see in \secref{sec:MS-phases}, $\rho^e$ and $\rho^m$ belong to the same phase as well. 
\end{definition} 

Following the arguments similar to Example~\ref{exmp:toric-code-fixed-point}, it is straightforward to show that $\rho^{e/m}$ satisfies mixed-state axioms. With this, the next step is to understand the information convex set $\Sigma(\mathbb{T}^2|\rho^e)$.

\begin{Proposition} \label{prop:convexset}
    The information convex set $\Sigma(\mathbb{T}^2|\sigma)$ for the canonically decohered $\mathbb{Z}_n$ quantum double $\rho^e$ on the torus is a simplex with $n^2$ mutually orthogonal extreme points with the same entanglement entropy.

    \begin{proof}
    Consider a closed manifold $M=\mathbb{T}^2$. First, we consider a fixed-point state $\sigma$ reduced to a ball $\bullet \subset M$:
    \begin{align}\label{eq:ball-condition-TC}
         X_e\sigma_{\bullet} & = \sigma_{\bullet} X_e, \nonumber \\
            A_v^{\bm{1}} \sigma_{\bullet} &=  \sigma_{\bullet} A_v^{\bm{1}}  = \sigma_{\bullet}, \\
            B_p^f \sigma_{\bullet}  &= \sigma_{\bullet} B_p^f.\nonumber
    \end{align} 
    For any $\rho \in \Sigma(M|\sigma)$, its reduced density matrix is indistinguishable from the fixed-point state on any ball $\bullet \subset M$, i.e., $\rho_\bullet = \sigma_\bullet$. Thus, we can replace $\sigma_\bullet$ with $\rho_\bullet$ in Eq.~\eqref{eq:ball-condition-TC}. 
    Does the same hold for any \emph{un-reduced} global state $\rho$? Without loss of generality, we partition $\bullet$ into $BC$ (as in {\bf P0}) such that the operator in question (either $A^{\bm{1}}_v$, $X_e$ or $B^f_p$) is supported on $C$,  and  $A = M \setminus BC$. For this partition, $I(A\,{:}\,C|B)_\rho =0$ by  condition {\bf P0} on $\rho$. This implies the existence of a quantum channel $\Phi_{B\to AB}$ such that
    \begin{equation}
        \rho = \Phi_{B\to AB} (\rho_{\bullet}) = \Phi_{B\to AB} (\sigma_{\bullet}).
    \end{equation}
    Applying the channel $\Phi_{B\to AB} $ to both sides of each equation in \eqref{eq:ball-condition-TC}, we show that the global density matrix $\rho$ commutes with $B_p^f$ and $X_e$ while absorbing $A_v^{\bm{1}}$.
    
    The rest of the proof is done by expanding $\rho$ in the generalized Pauli basis. Write 
    \begin{align}
        \rho = \sum_{P \in \calP_n^{\otimes N}} c_P P 
    \end{align}
    where $\calP^{\otimes N}_n$ is a set of $n$-generalized Pauli operators on $N$ sites.
    The condition $[X_e, \rho_M]\,{=}\,0$ and $[B_p^f, \rho_M]\,{=}\,0$ for any $i$ and $p$ implies that $\rho$ commutes with an operator algebra generated by $\{ X_e, B_p^f \}$. Therefore, $c_P =0$ for Pauli strings involving $Y$ and $Z$, and even for Pauli strings in terms of $I$ and $X$, they should commute with $B_p^f$s--this implies that they must be \emph{closed} $X$-strings on the dual lattice. 
    The condition $A^1_v  \rho \,{=}\, \rho$ implies that $\rho$ can be written as a product of $A^1_v$ and additional terms. 
    Since any contractible closed $X$-strings can be absorbed into $\prod_v A_v^{\bm{1}}$, only remaining degrees of freedom are non-contractible $X$-strings. For $\calM = \mathbb{T}^2$ with two non-contractible strings, we get
    \begin{align}\label{eq:logical-states-ab}
        \rho = \bigg[ \sum_{a,b=0}^{n-1} c_{ab} \bar{X}_1^a \bar{X}_2^b \bigg] \cdot  \prod_v A_v^{\bm{1}}.
    \end{align}
    where $\bar{X}_{1,2}$ are generators of logical $X$-operators along the $x$ and $y$ direction. Since $\rho$ is diagonal in the joint eigenspaces of $\bar{X}_1$ and $\bar{X}_2$, it can be written as a convex combination of the corresponding joint-eigenspace states. That is,
    \begin{equation*}
        \rho = \sum_{a,b=0}^{n-1}p_{ab}\rho^e_{ab}, \qquad p_{ab}\geq 0, \qquad \sum_{a,b=0}^{n-1}p_{ab}=1.
    \end{equation*}
    The following $n^2$ density matrices are mutually orthogonal extreme points:
    \begin{align} \label{eq:extreme_points}
         \rho^e_{ab} = \frac{1}{\mathcal N_{ab}}\prod_v A_v^{\bm{1}} \cdot
         \left[\frac{1}{n}\sum_{j=0}^{n-1}(\omega^{-a}\bar{X}_1)^j\right]
         \left[\frac{1}{n}\sum_{k=0}^{n-1}(\omega^{-b}\bar{X}_2)^k\right],
    \end{align}
    where $\mathcal N_{ab}$ is chosen so that $\Tr\rho^e_{ab}=1$, and the factors next to $\rho^e = \prod_v A_v^{\bm{1}}$ project $\rho^e$ onto an eigenspace of $\bar{X}_{1,2}$ with eigenvalues $\omega^a, \omega^b$.
    Therefore, $\Sigma(\mathbb{T}^2|\sigma)$ is forming a simplex with $n^2$ orthogonal extreme points.  
    \end{proof}
\end{Proposition}

\begin{corollary}[memory capacity] The memory capacity $\calQ(\mathbb{T}^2)$ of the state $\rho^e$ ($\rho^m$) is
\begin{equation}
    \calQ(M) = 2 \log n.
\end{equation}
\end{corollary}
This follows immediately from the counting of the extreme points, the orthogonality, and the fact that all the extreme points have identical entropies. It is also easy to compute that $\gamma_{\textrm{LW}}= \frac{1}{2} \log n$. Thus, the class of abelian models saturates the bound \eqref{eq:bound-mixed}.

\begin{lemma} \label{lemma:fdp} The extreme points of $\Sigma(\mathbb{T}^2|\rho^{e,m})$  can be prepared by a finite-depth local channel from the product state.

    \begin{proof}
        Without loss of generality, we consider $\rho^e$ only. The extreme points of $\Sigma(\mathbb{T}^2|\rho^{e})$ are labeled by two integers $a,b \in \mathbb{Z}_n$.
        Define the product state $\lambda_{a,b}$ as 
        \begin{align}
        \lambda_{ab} &:= L^Z_{a,b} \, (|+ \rangle \langle + |)^{\otimes N} (L^Z_{a,b})^\dagger
        \end{align}
        where $L^Z_{a,b} = \prod_{e \in C_1} Z_e^a \prod_{e \in C_2} Z_e^b$ is a logical $Z$-string along noncontractible loops $C_1$ (horizontal) and $C_2$ (vertical) of the torus as shown in Fig.~\ref{fig:preparation_abelian}. By performing a positive-operator-valued measure (POVM) using Kraus operators $\{ B_p^f\}$ (which commute with $L^Z$), we obtain
        \begin{align}\label{eq:B-fp}
            \rho^e_{ab}  &:= \sum_{ \{ f_p \} } \prod_p B_p^{f_p} \lambda_{ab} B_p^{f_p},
        \end{align}
        which is the extreme point labeled by $(a,b)$.  
        To see this, note that the Pauli expansion of the density matrix $\lambda_{ab}$ contains all possible dual $X$-string operators $X_{\bm{l}^\perp}$ with coefficients given by $L^Z_{a,b} X_{\bm{l}^\perp} (L^Z_{a,b})^{\dagger} X^{\dagger}_{\bm{l}^\perp}$.
        Next, under the POVM by $B_p^f$, the only surviving terms are dual \emph{closed} strings. Defining $\omega = e^{2\pi i/n}$, the final density matrix can be rewritten in the form in \eqnref{eq:extreme_points}. 
    \end{proof}
\end{lemma}

\begin{figure}[!t]
    \centering   
    \includegraphics[width=0.97\linewidth]{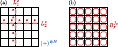}
    \caption{ {\bf Preparation of dephased Toric code}. Consider a square lattice on a torus. (a) First, we prepare a product state and apply logical-$Z$ operators $L^Z_{a,b}$. (b) Second, we perform $Z$-stabilizer measurements by POVM $\{B_p^{f_p}\}$, and throw away measurement outcomes. The final state has an eigenvalue $(a,b)$ under $X$-logical operator $(\bar{X}_1,\bar{X}_2)$.}
    \label{fig:preparation_abelian}
\end{figure}

Let us comment on the separable nature of the states. While we formulated as applying a channel with Kraus operators being projective measurement, this is equivalent to applying randomly unitary operators $B_p^m = \prod_{e \in p} Z_e^{m \epsilon_{e,p}}$ obtained by inverse Fourier transformation (c.f. \eqnref{eq:DFT}). 
Thus, the state in Eq.~\eqref{eq:B-fp} can be re-written as
\begin{equation}
    \rho_{ab}^e= \sum_{m_p} \prod_p B_p^{m_p} \lambda_{ab} B_p^{-m_p},
\end{equation}
which makes it manifest that $\rho_{ab}^e$ is a classical separable state. The fact that such a state has $\gamma_{\textrm{LW}} =\frac{1}{2} \log n$ means the source of topological entropy is \emph{not} from the entanglement. This is why we call $\gamma_{\textrm{LW}}$ the topological entropy rather than topological entanglement entropy in the mixed states.

\section{Mixed state phases of matter}\label{sec:MS-phases}

Thus far, we have proposed an axiomatic starting point for understanding topological mixed states. 
By focusing on axioms that isolate topological features, we characterized nontrivial behaviors of fixed points without being distracted by non-universal details. However, a fundamental question remains: how does this framework contribute to the broader understanding of mixed-state phases of matter? To define the equivalence of two topological mixed states satisfying the bootstrap axioms, one natural approach is to propose that two mixed states belong to the same phase if they are connected by finite-depth local channels (FDLC) in both directions~\cite{Coser_2019, xue2024tensornetworkformulationsymmetry}; also see Ref.~\cite{Sang2024}. However, our axiomatic analysis reveals a tension: for example, the trivial product state and the maximally dephased toric code represent distinct fixed points, with different classical topological memory capacities $\calQ$. Yet under the FDLC criterion, they appear to belong to the same phase. This discrepancy highlights the need for a more refined phase definition in the mixed-state setting.

In this section, we show that our approach provides a systematic framework to define mixed-state topological phases. 
First, we define the equivalence relation between two fixed-point states satisfying the axioms (Def.~\ref{def:p.f.-equivalence}). 
Through this definition, we propose a new way to define an equivalence relation between generic topological mixed states.  
Second, we discuss how to relax axioms to explore \emph{phases} away from the fixed points, and show numerical results to justify our relaxation conditions. 
Finally, we investigate novel physics that arise at the domain wall between distinct topological mixed states.

\subsection{Mixed state phase equivalence}

A reference state satisfying our axioms is a fixed point of topological mixed state (Def.~\ref{def:fp_topological}).
Two reference states $\sigma^\alpha$ and $\sigma^\beta$ that can be sewn together along a boundary satisfying those axioms necessarily belong to the same phase. This is the consequence of the isomorphism theorem (see Appendix \ref{app:ICS-Omega}), which implies that whenever a topological invariant (e.g. $\gamma_{\textrm{LW}}$ in \eqnref{eq:LW-TE-G}) is evaluated on any sufficiently large local patch, its value should be the same everywhere. This indicates that the topological correlations can be smoothly extended from a local patch to further neighborhood. Therefore, one can use the absence of an interpolating boundary satisfying axioms as a probe to distinguish two topological mixed states.
This result is the conceptual foundation for the bootstrap approach, even for the pure state setup~\cite{Shi2019fusion}.

To extend this approach to explore the mixed-state \emph{phase diagram}, we may relax the axioms to hold only under \emph{coarse-graining}. 
In practice, this means that a given axiom may be violated on short length scales, provided that such violations vanish under successive rescaling of the region and partition size. 
These considerations motivate the following formal definition. 

\begin{axiombox}
    \begin{definition}[Mixed state phase equivalence]\label{def:p.f.-equivalence}
Two representative mixed states $\sigma^\alpha$ and $\sigma^\beta$ are equivalent (i.e., in the same phase) if there exists a state containing a bubble of $\beta$ in the background of $\alpha$, such that {\bf P0}, {\bf M0} and {\bf M1} are preserved everywhere including the bubble domain wall, up to an error that decays under coarse-graining. See Fig.~\ref{fig:domain-wall} for an illustration.
\end{definition}
\end{axiombox}

In other words, the domain wall becomes \emph{transparent} under coarse-graining. 
The transparent domain wall can be formally defined as:
\begin{definition}[Transparent domain wall]\label{def:transparent-wall}
    If a domain wall of a many-body system cannot be detected by checking any {\bf P0}, {\bf M0}, and {\bf M1} centered in its neighborhood (under coarse-graining), we call it a transparent domain wall. See Fig.~\ref{fig:domain-wall} for an illustration.
\end{definition}

\begin{figure}[!t]
    \centering
    \includegraphics[width=0.95\linewidth]{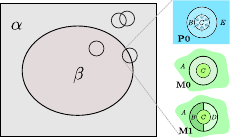}
    \caption{ {\bf Mixed state phase equivalence.} A bubble of phase $\beta$ in the background of $\alpha$. We say the two phases are equivalent, i.e., $\alpha \sim \beta$ if the axioms are preserved everywhere, including the domain wall. The associated domain wall is called a transparent domain wall. 
    }\label{fig:domain-wall}
\end{figure}

\subsubsection{Application of phase equivalence}

The definition of mixed phase equivalence (Def.~\ref{def:p.f.-equivalence}) is a crucial ingredient of our theoretical framework. This is because, a mixed state phase should be classified as fixed point modulo equivalence. 
We first provide a few basic applications.

 \begin{exmp}[Trivial mixed state phase]
 We say a fixed point wave function is in the trivial phase if it is equivalent to the product state $\otimes_j|0\rangle_j$ in the sense of Def.~\ref{def:p.f.-equivalence}. It is easy to see that the tensor product of arbitrary state 
 $\otimes_j \rho_j$ can be joined with $\otimes_j|0\rangle_j$ by a transparent domain wall. We denote this as $\otimes_j \rho_j \sim \otimes_j|0\rangle_j$. Thus $\otimes_j \rho_j$ is in the trivial mixed state phase.
 \end{exmp}

 \begin{exmp}
     The equivalence Def.~\ref{def:p.f.-equivalence} implies that the phases described by $\rho^e$ and $\rho^m$—two canonically decohered toric code ($\mathbb{Z}_2$ quantum doubles)—are in the same phase. While a naive juxtaposition of the two results in a boundary that violates axiom {\bf M1}, one can construct a domain wall that preserves all axioms, as illustrated in \figref{fig:domain_em}. We denote this result as
     \begin{equation}
         \rho^e \sim \rho^m.
     \end{equation}
 \end{exmp}

 \begin{exmp}[Gapped boundary of the toric code]\label{exmp:T.C-gapped-bdy} 
 The toric code and the vacuum state represent different mixed-state phases, namely $\rho^{\rm T.C.} \not\sim \otimes_j|0\rangle_j$.
 For our definition to be consistent with this, the gapped domain wall of the toric code cannot be transparent. We explain why this is the case for the known gapped boundary of the toric code; see Fig.~\ref{fig:toric-code-bdy}. Axioms {\bf P0} and {\bf M0} are satisfied on the gapped boundary.  A version of {\bf M1} is violated by $\log 2$~~\cite{Beigi2011,Shi2021-DWTEE}. We shall find more kinds of non-transparent domain wall in Sec.~\ref{subsec:channel-wet-dry}, especially the dry and wet boundary natural in mixed state setups.
 \end{exmp}

\begin{figure}[!t]
    \centering   
    \includegraphics[width=0.59\linewidth]{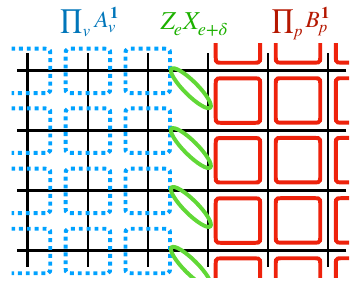}
    \caption{ {\bf Domain walls}.  Smooth sewing of $\rho^e$ and $\rho^m$, which satisfy all axioms. Red and blue squares represent the $Z$ and $X$ stabilizers $\prod_{e \in p} Z_e$ and $\prod_{e \ni v} X_e$, respectively. Green diagonal circles correspond to the stabilizer $Z_r X_{r + \delta}$, where $\delta=(1/2,-1/2)$. }
    \label{fig:domain_em}
\end{figure}

\begin{figure}[h]
    \centering    \includegraphics[width=0.75\linewidth]{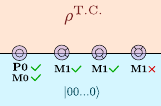}
    \caption{ {\bf Domain walls II}. Verification of the preservation and violation of the axioms on a gapped domain wall between the toric code ground state $\rho^{\rm T.C.}$ and the trivial product state. Axioms {\bf P0} and {\bf M0} are preserved on the wall.  The violation only happens to axiom {\bf M1} in the right-most partition, whether we choose $B$ to be the upper or lower disk contained in the purple annulus that thickens the local ball. The amount of the violation is $\log 2$. This makes the wall non-transparent.} 
    \label{fig:toric-code-bdy}
\end{figure}

How can we tell if two fixed points $\sigma^\alpha$ and $\sigma^\beta$ are in different phases? To show $\sigma^\alpha$ and $\sigma^\beta$ are in different phases, it is \emph{not} enough to identify a domain wall that violates {\bf M1}. As studied in the pure state case~\cite{levin2024physicalprooftopologicalentanglement,Kim2024-A11,Kim2023-TEE-bound}, the violations may be due to spurious contributions. It is necessary to show that it is impossible to construct another domain wall of the bubble, which preserves {\bf M1}. The task may sound challenging. Luckily, there are simpler sufficient conditions. 

Firstly, two phases are distinct if two states have different topological entropy:
\begin{equation}\label{eq:LW-difference}
    \gamma_{\textrm{LW}}(\alpha) \ne \gamma_{\textrm{LW}}(\beta).
\end{equation}
Indeed, if two states can be connected by a transparent domain wall, then $\gamma_{\textrm{LW}}(\alpha) = \gamma_{\textrm{LW}}(\beta) $ because we can smoothly deform the Levin-Wen partition through the domain wall. We demonstrate how this works in Sec.~\ref{Sec:robust-TE}, including the context that the axioms has errors decaying with the subsystem sizes. Since $\gamma_\LW(\rho^{e,m}) = \log \sqrt{2}$ while $\gamma_\LW(\otimes_i \rho_i) = 0$, it immediately follows that the product state and $\rho^{e,m}$ are not in the same phase according to our definition.  

Secondly, our axiomatic approach implies that two mixed states belong to different phases if they possess different information convex sets (ICS). This approach is most useful when the regions are topologically interesting and embedded in the physical system, allowing entanglement boundaries such as the annulus; see Fig.~\ref{fig:immersed}(a). Further interesting and subtly important examples are regions that are only locally embedded, such as the immersed punctured torus in Fig.~\ref{fig:immersed}(b). This is because of an \emph{isomorphism theorem} we explain in Appendix~\ref{app:ICS-Omega}, which guarantees that the information convex sets are isomorphic under any topological deformation as long as axioms are satisfied; one example is moving it across the transparent domain wall.

\begin{figure}[h]
    \centering
    \includegraphics[width=0.71\linewidth]{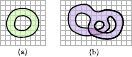}
    \caption{ {\bf Manifolds for ICS.} (a) An annulus embedded in the physical system. (b) A punctured torus that is locally embedded but not globally embedded, known as the immersed punctured torus. This figure is adopted from Ref.~\cite{Braiding2023}.}
    \label{fig:immersed}
\end{figure}

\begin{exmp}

Using the second approach, one \emph{can} distinguish the ground state of chiral semion topological order ($\rho^s$) and $\rho^e$, namely
\begin{equation}
    \rho^s \not\sim \rho^e.
\end{equation}
Note that, they \emph{cannot} be distinguished by their topological entropy, since $ \gamma_{\rm LW}(\rho^s)=\gamma_{\rm LW}(\rho^e)= \log \sqrt{2}$.
The difference can be seen from the information convex set of the immersed punctured torus, as in Fig.~\ref{fig:immersed}. Such an information convex set is a Bloch ball for the chiral semion, but is a simplex with 4 mutually orthogonal extreme points for $\rho^e$. See Appendix~\ref{app:phase-equiv} for details.
    
\end{exmp}

We discuss the robustness of the isomorphism of information convex sets under decaying errors in Appendix~\ref{app:robust-iso}. Although having identical information convex sets is a necessary condition for two fixed-point wave functions to belong to the same phase, it is not sufficient: for example, the toric code and doubled semion topological phases share the same information convex sets, yet they cannot be transformed into one another even by finite-depth local circuits (FDLC) \cite{Ranard2025}. In Appendix~\ref{app:phase-equiv}, we outline an alternative approach that extracts additional insight by exploiting nontrivial automorphisms of the information convex sets.

\subsubsection{Space versus time}

We have chosen our definition of phase equivalence in terms of transparent domain wall separating spatial patches (Def.~\ref{def:p.f.-equivalence}). This is aligned with the entanglement bootstrap philosophy that the property of topological phases can be extracted from a ball-like region that is larger than the correlation (and Markov) length. Such an approach also naturally allow us to describe spatial boundaries of the system.

Alternatively, one may seek a time-based picture of the phase transition. For the traditional Hamiltonian phase transition of ground states, a common wisdom is that one must close the gap and the correlation length diverge. For mixed state topological phases, the nice observation of Ref.~\cite{Sang2024} is that one need to look into the Markov length associated with concentric partitions of a ball, as the correlation length fails to detect such a transition. By a Markov length, we mean a length $\xi_{\bf P0}$ that appears in the violation of axiom {\bf P0} as~\cite{Chen2025}:
\begin{equation}
    I(A\,{:}\,C|B) \sim \textrm{poly}(A,C) \exp(-\frac{\textrm{dist}(A,C)}{\xi_{\bf P0}}),
\end{equation}
where $\textrm{dist}(A,C)$ is the minimum distance between $A$ and $C$. 
This approach is a time based approach in the sense we explain in Fig.~\ref{fig:domain_space-time}, namely, one can interpret the divergence of the Markov length at some critical tuning parameter (error rate or Lindbladian evolution time) as the \emph{temporal} domain wall between two distinct mixed-state phases across time. 

\begin{figure}[!t]
    \centering   
    \includegraphics[width=0.69\linewidth]{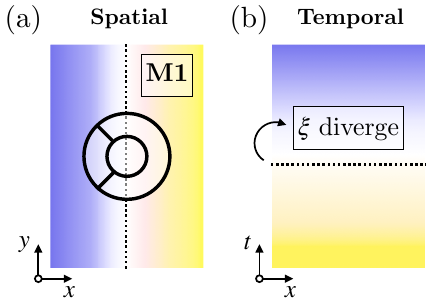}
    \caption{ {\bf Spatial versus temporal domain wall}. (a) Spatial domain wall and the violation of {\bf M1}. (b) Temporal domain wall and the divergence of Markov length or correlation length at a certain time slice. }
    \label{fig:domain_space-time}
\end{figure}

However, a Markov length alone is not enough within our framework, as it is only one of the three pillars. What we propose precisely is that there are three correlations in question, which corresponds to the decaying of the three conditions {\bf P0}, {\bf M0} and {\bf M1}; we call these lengths $\xi_{\bf M1}$, $\xi_{\bf M0}$, and $\xi_{\bf P0}$, respectively. Two of them are Markov lengths with different partitions and one of them is the correlation length.




\begin{axiombox}
\begin{Conjecture}[Phase equivalence with time]\label{conj:t}
Consider a 1-parameter family of states $\rho(t)$  for $t \in [0,1]$. States $\rho(0)$ and $\rho(1)$ are in the same mixed state phases if the three lengths $\xi_{\bf M1}$, $\xi_{\bf M0}$, $\xi_{\bf P0}$ are upper bounded by a finite constant for all $t$.
\end{Conjecture}  
\end{axiombox}

Importantly, we require \emph{all three} length scales to be finite, and the statement does not involve the concept of channels, circuit depth, or Lindbladian. One may wonder whether we can drop the finiteness of one of the lengths in the conjecture. Naive ways of relaxing these requirements lead to statements disproved by examples:
\begin{itemize}[leftmargin=13pt]
    \item The traditional wisdom of phase transition suggests that one should require the finiteness of $\xi_{\bf M0}$. However, we cannot only require the finiteness of $\xi_{\bf M0}$; this observation, together with the observation that one should look at $\xi_{\bf P0}$ are the remarkable insight of Ref.~\cite{Sang2024}.
    
    \item Suppose we require \emph{only} the finiteness of $\xi_{\bf P0}$, but not $\xi_{\bf M0}$, we encounter the counterexample of classical 3D $\mathbb{Z}_2$ gauge theory transition at finite temperature. The transition is not detected by the Markov length $\xi_{\bf P0}$ as it is exactly zero for a thermal state of any stabilizer Hamiltonian~\cite{Hammersley1971MarkovFO}. Thus, the requirement of finite $\xi_{\bf M0}$ cannot be omitted.

    \item Suppose we require both $\xi_{\bf P0}$ and $\xi_{\bf M0}$ to be finite but put no requirement on $\xi_{\bf M1}$. Then one can design a sequential circuit that grows a toric-code patch such that $\rho(0)$ is a trivial product state while $\rho(1)$ is a toric code state. While both states have $\xi_{\bf M1}\,{=}\,0$, in the intermediate time we get a state with a domain wall with diverging $\xi_{\bf M1}$. Such a 1-parameter familar of states all have finite $\xi_{\bf P0}$ and $\xi_{\bf M0}$. However, those in the intermediate time steps do not represent a uniform representative of mixed state phases.  
\end{itemize}

\begin{figure*}[t]
  \centering
\includegraphics[width=\linewidth]{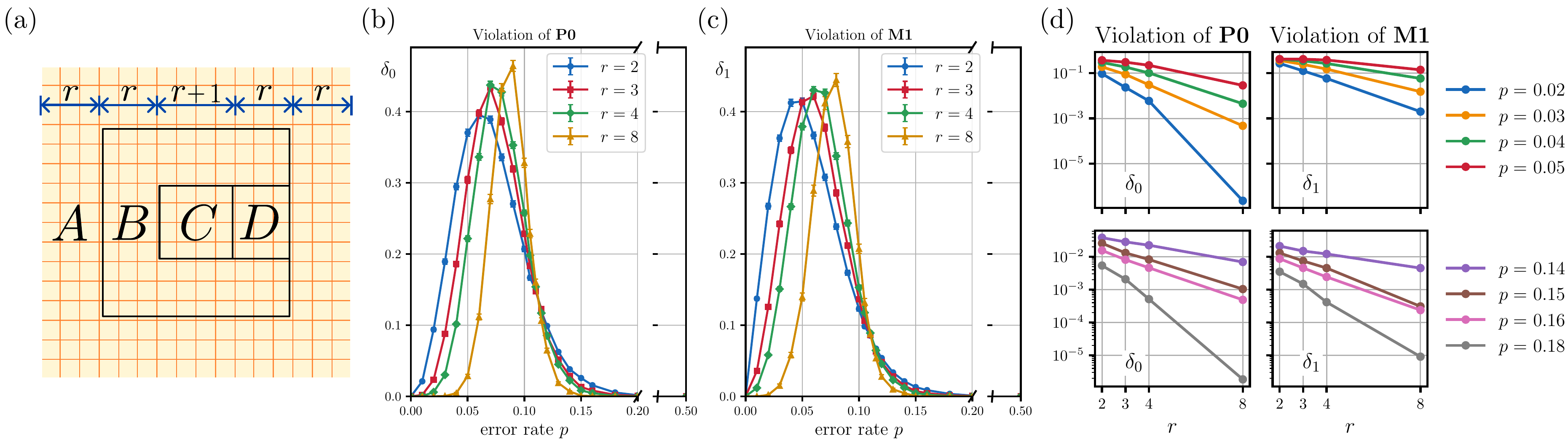}
  \caption{{\bf Violation of P0 and M1 in the bulk.} 
    (a) We partition a bulk disk into $A,B,C$ and $D$ as shown.
    (b) The (local) violation of {\bf P0}~\footnote{Using monotonicity of conditional mutual information under partial traces, we can establish that $I(E:C|B) \geq I(A\,{:}\,C|B)$, $\forall A \subset E$, i.e., the violation of {\bf P0} is lower bounded by the violation of its local version. } 
    $\delta_0 := S_{ABD} + S_{BCD} - S_{BD} - S_{ABCD}$ as a function of error rate $p$ and subsystem sizes $r$. 
    (c) The violation of {\bf M1}, $\delta_1 := S_{AB} + S_{BC} - S_B - S_{ABC}$, as a function of error rate $p$ and subsystem sizes $r$. 
    (d) At $p=[0.01, 0.04]$ and $p=[0.13, 0.16]$ the error of {\bf M1} and {\bf P0} decays exponentially with $r$. All data are shown in units of $\log 2$.}
  \label{fig:bulk_violations}
\end{figure*}



 In summary, we made a conjecture of a time-based mixed state phases equivalence. This statement is largely in the spirit of Ref.~\cite{Sang2024}. Nonetheless, we required three length scales to be finite, instead of just the Markov length for the concentric regions, and have given the physical reason for doing so. Especially, we highlight the need for the special uniformity: {\bf M1} should be satisfied upon coarse-graining. Otherwise, the representative wave function can stop being a uniform representative of a phase, and we run into counterexamples. 
 In fact, even for the ground state of topologically ordered systems the common wisdom of “correlation length remains finite implies no phase transition” suffers from the same counterexample if one discards space uniformity.

\subsection{Coarse-graining}\label{subsec:coarse-graining}

For a generic topological mixed state, we do not expect the three axioms to hold exactly, but only approximately. Given any local patch and associated partition considered in the axioms, one can rescale all the regions with respect to a reference point by a multiplicative factor $r$.

In the case of critical states described by conformal field theory, it is known that the conditional mutual information  
for the partition in {\bf P0} is invariant under such scaling~\cite{Casini2011}. In contrast, for gapped phases of matter, this condition is expected to decay as a function of $r$. Motivated by this behavior, we propose the following operational definition: 
\begin{definition}[Coarse-graining]\label{def:CG}
    If the violations of the three axioms decay towards zero with the multiplicative factor $r$ (by which we rescale all the regions) throughout the state, then we say the state \emph{coarse-grains} into a particular topological mixed state fixed point.  
\end{definition}
In what follows, we perform numerical simulations to study coarse-graining behavior of the axioms away from the fixed-point states.

\subsubsection{Homogeneous system}

Consider decoherence with a uniform error rate $p$. The corresponding error channel maps $\rho^e$ (Eq.~\eqref{eq:rho^e-T.C.}) to a product state at $p=0.5$ and is given by
\begin{align}\label{eq:noise-Z}
    \calE := \prod_i \calE_i, \,\,\,\, \calE_i: \, \rho \mapsto (1-p) \rho + p Z_i \rho Z_i.
\end{align}
When $p$ is $0$, we have $\rho^e$ as the initial state. As we increase the error rate to $0.5$, we get a maximally mixed state. For each $p\in [0,0.5]$, we get a state that can be treated as a homogeneous system.

We investigate the violations of the three axioms for different values of $p$. We focus on {\bf P0} and {\bf M1} since the violation of {\bf M0} is always zero for any $p$ because the errors are onsite and independent.  The partition is given in Fig.~\ref{fig:bulk_violations}(a) where we scale all regions proportionally. As we see in Fig.~\ref{fig:bulk_violations}(b) and (c), the violations of both axioms {\bf P0} and {\bf M1} are exactly $0$ for $p = 0$ and $p = 0.5$ as expected; at $p$ away from the critical value, both axioms have their violations coarse-grain towards $0$. Up to finite size corrections to the violations, the errors decay exponentially with the separation $r$ for both $p<p_c$ and $p>p_c$ as in Fig.~\ref{fig:bulk_violations}(d), establishing the stability of our approach. 

Moreover, in Fig.~\ref{fig:bulk_violations}(b,c) we observe a clear crossing point near $p_c \simeq 0.109$, at which the curves for different length scales $r$ intersect. This provides a simple empirical method for determining $p_c$ with high precision. More broadly, this crossing-point method carries a practical advantage and may be applicable in a wider range of settings.

This is in contrast to the approach of Ref.~\cite{Sang2024}, which relies on the Markov length: there, the size of region $C$ is held fixed while the separation between $A$ and $C$ is rescaled, rather than uniformly rescaling all regions as in our construction. Consequently, identifying $p_c$ in Ref.~\cite{Sang2024} requires an additional curve-by-curve rescaling because the peak height decays with $r$, making the analysis more complicated.

\begin{figure*}[!t]
  \begin{minipage}{\linewidth}
  \centering
\includegraphics[width=\linewidth]{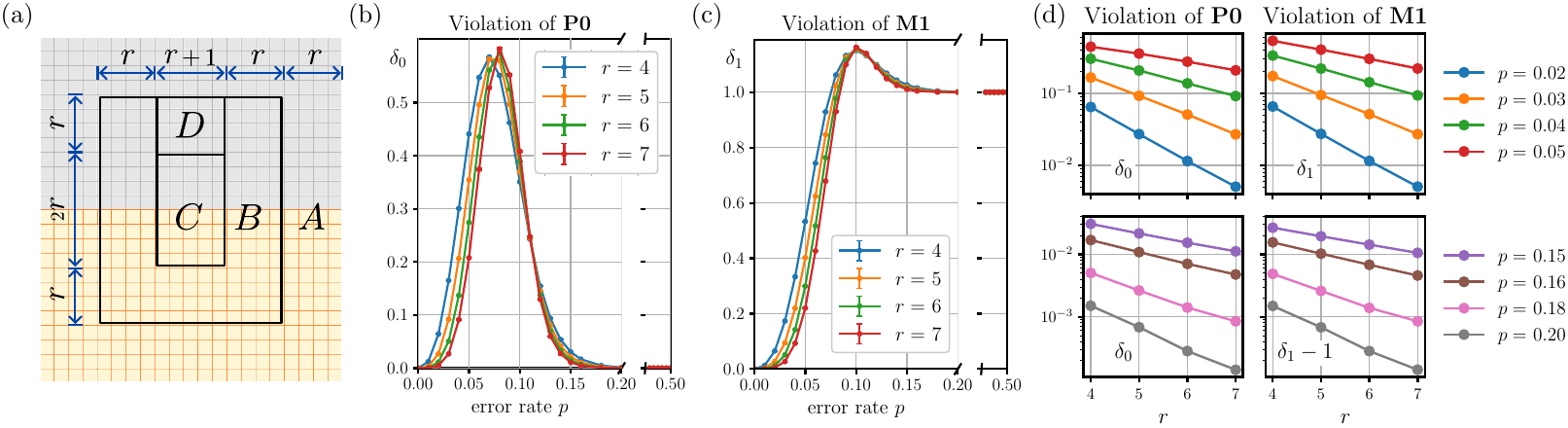}
  \caption{{\bf Violation of P0 and M1 at the boundary.} 
    (a) The partitions near a domain wall. The decohered links are colored orange, located in the lower half plane. $r$ is the length of the subsystems.
    $A$ is an annulus, $BD$ is the annulus that surrounds the disk $C$. 
    (b) The (local) violation of {\bf P0}. 
    (c) The violation of {\bf M1}. 
    (d) At $p=[0.02, 0.05]$ and $p=[0.15, 0.20]$ the error of {\bf M1} and {\bf P0} decays exponentially with $r$. All data are shown in units of $\log 2$.}
  \label{fig:violations}
  \end{minipage}
\end{figure*}

Our way of coarse-graining (Def.~\ref{def:CG}) by rescaling \emph{all} regions by the same multiplicative factor raises another natural question. Is it true that, the height of the peak in the violation of the axioms approaches a constant at large length scales at $p_c$? This appears to be a viable conjecture even though Ref.~\cite{ Sang2024} sees a power law decaying of the violation of {\bf P0} due to the fixed size of $C$. Indeed, if the critical point is described by a scale-invariant (or even conformal) effective theory, it is anticipated that the height of $\delta_0$ and $\delta_1$ at $p_c$ should be a constant that does not depend on the length $r$, consistent with Fig.~\ref{fig:bulk_violations}(b) and (c). It seems natural to look for connections to a disordered conformal field theory~\cite{LUDWIG1990639,Cardy2001,Gurarie_2002}, which opens avenues for further investigations.

\subsubsection{Spatial Boundary: Inhomogeneous decoherence}

Now suppose that we decohere a fixed point state on a half plane using local channels. This will create a domain wall (i.e., a spatial boundary) between two homogeneous states. We are interested in the properties of the domain wall, e.g. whether it is transparent (Def.~\ref{def:transparent-wall}). 
We consider the following two examples. The first is identical to the previous example in the bulk.

\vspace{5pt} \noindent \emph{From $\rho^e$ to product state.} 
Now consider the case that only lower half-plane of $\rho^e$ is decohered under on-site $Z$ errors with a tunable error rate $p$ as in Eq.~\eqref{eq:noise-Z}, see \figref{fig:violations}(a).  At $p\,{=}\,1/2$, the lower half-plane becomes a maximally mixed product state, which has a topological entropy $\gamma_\LW$ different from the upper half-plane.  
Again, we only present the violation of {\bf P0} and {\bf M1} since the violation of {\bf M0} is always zero for every $p$. One can analytically find the violation of {\bf M1} for the partition $I(A\,{:}\,C|B)$ in Fig.~\ref{fig:violations}(a) to be $\log 2$, while the violation of {\bf P0} is $0$ at $p=0.5$, see \appref{app:violation}.

For $p < p_c$, we expect the violation of both {\bf P0} and {\bf M1} to coarse-grain to zero. On the other hand, for $p > p_c$, the violation of {\bf M1} shall coarse-grain to $\log 2$ while the violation of {\bf P0} to zero. This is what we observe in Fig.~\ref{fig:violations}(d) with both violations decay exponentially to the fixed point values. 

At the critical point $p_c$, violations of both {\bf P0} and {\bf M1} in Fig.~\ref{fig:violations}(b,c) exhibit scale invariance with respect to the region size $r$ similar to the homogeneous case.  Again, we emphasize that in computing the violation of \textbf{P0} in both bulk and domain-wall boundary, we scale \emph{all} regions $A$, $BD$, and $C$ (and $A$, $B$, $C$ for {\bf M1}). 

As in the bulk data, we can locate the critical point directly from the crossings of curves at different $r$. We observe that the crossing occurs exactly at the known critical point $p_c\simeq 0.109$~\cite{fan2023diagnostics, Lee_2023, lee2024exact, Zhao_2024}. We detail the numerical calculation in Appendix~\ref{app:num}.

\vspace{5pt} \noindent \emph{From product state to  $\rho^e$.}
We can also consider the opposite channel, which prepares $\rho^e$ starting from the product state. This is the tunable version of the preparation protocol described in \figref{fig:preparation_abelian}, where we perform the decoherence using $B_p^f$ with error rate $p$. Applying this channel on the product state $|+\rangle^{\otimes n}$ gives the following density matrix:
\begin{align} \label{eq:dm_fragile}
    \rho (p) &= \frac{1}{2^n} \sum_{\bm{l}_\perp \in {\cal S}_d }  (1-2p)^{|\partial \bm{l}_\perp|} X^{\bm{l}_\perp}, \quad X^{\bm{l}_\perp} := \prod_{e \in \bm{l}_\perp} X_e.   
\end{align}
where ${\cal S}_d$ is the set of all possible strings on the dual lattice and $| \partial \bm{l}_\perp|$ is the number of boundary points of $\bm{l}_\perp$.

To proceed, we start with $|+\rangle$ state every site and apply this protocol with error rate $p$ on the upper half-plane of a given system. For the fixed point wave function at $p=1/2$, we have $I(A\,{:}\,C|B)=0$ but $I(A\,{:}\,C|D)=\log 2$, near the domain wall Fig.~\ref{fig:violations}(a), which indicates the violation of {\bf M1}. 
However, unlike the previous case, this FDLC is \emph{fragile}, which means that the violation of axioms coarse-grains into zero at $p<1/2$ and two half-planes belong to the same trivial phase, similar fragile behaviors in the bulk are observed. We show the numerical and theoretical details in Appendix \ref{app:fragile}.

We can infer that the state in \eqnref{eq:dm_fragile} with $p<1/2$ exhibits vanishing topological entropy for sufficiently large partitions. 
At an intuitive level, the contribution of independent closed‐loop operators to the density matrix becomes negligible under coarse-graining. This is because a string operator
$X_{\bm{l}_\perp} = \prod_{e \in \bm{l}_\perp} X_e$
maintains a finite expectation value $(1-2p)^2$ regardless of loop length $|\bm{l}_\perp|$---behavior that cannot persist in a genuinely topological mixed state. 
From the statistical mechanics perspective, the physics maps to that of the two-dimensional $\mathbb{Z}_2$ gauge theory. In particular, Renyi-2 entropy would map to the free energy of the 2d Ising gauge theory at the inverse temperature $\beta = -2 \ln (1-2p)$. Since 2d Ising gauge theory is disordered at any finite $\beta$, the (Renyi-2) topological entropy vanishes for $p < 1/2$. 
See \appref{app:fragile} for numerical data.
Following numerous observations~\cite{fan2023diagnostics, Lee_2023, lee2024exact, Zhao_2024, Niwa_2025}, we expect the von Neumann topological entropy would vanish, i.e., coarse-grain to zero for $p<1/2$ since its behavior would follow random-plaquette $\mathbb{Z}_2$ gauge theory, which has no finite‐temperature threshold.

\subsection{Channel and boundary}\label{subsec:channel-wet-dry}

We have proposed an equivalence of mixed state fixed points based on their density matrices (Def.~\ref{def:p.f.-equivalence}). How is this approach different from the existing approach based on the FDLC~\cite{Coser_2019, xue2024tensornetworkformulationsymmetry} or Markov length~\cite{ Sang2024}? In this section, we discuss the connection and distinction between these approaches and ours. In particular, we explain why our definition naturally leads to a refinement of the FDLC approach.

The previous FDLC-based approach to which we compare our framework can be phrased as follows~\footnote{There is another approach closely related to FDLC but uses local Lindbladian~\cite{Sang2024}. It may not be completely equivalent to the local channel condition.}.

\begin{definition}[Channel connectivity]\label{def:CC}
Two mixed states fixed points $\sigma^\alpha$ and $\sigma^\beta$ are channel connected if there exists a pair of FDLCs ($\Phi_{\alpha\to \beta}$ and $\Phi_{\beta\to \alpha}$), which can convert between $\sigma^\alpha$ and $\sigma^\beta$. Namely,
\begin{equation}
   \Phi_{\alpha\to \beta}(\sigma^{\alpha}) = \sigma^\beta,\quad   \Phi_{\beta\to \alpha}(\sigma^{\beta}) = \sigma^\alpha.
\end{equation}   
\end{definition}

Suppose $\alpha$ and $\beta$ are two fixed points that are channel connected as in Def.~\ref{def:CC}. The finite depth quantum channels $ \Phi_{\alpha\to \beta} $ and $ \Phi_{\beta\to \alpha} $ can be truncated onto any (sufficiently large) local region. Such channels can create bubbles of $\alpha$ $(\beta)$ over the background of $\beta$ $(\alpha)$.  Our idea of refining Def.~\ref{def:CC} is that it is possible to find out more by looking at the correlation near the domain wall of the bubble. Based on this, we propose the \emph{topological channel connectivity} below.

\begin{axiombox}
\begin{definition}[Topological channel connectivity]\label{def:TCC}
Two representative mixed states $\sigma^\alpha$ and $\sigma^\beta$ are \emph{topologically channel connected} if there exists a FDLC that can create a disk-like bubble of $\beta$ in the background of $\alpha$ (and vice versa), such that for every layer of the spatially truncated FDLC, {\bf P0}, {\bf M0} and {\bf M1} are preserved everywhere including along the domain wall of the bubble, up to decaying error. 
\end{definition}
\end{axiombox}

It is clear that Def.~\ref{def:TCC} is a refinement of channel connectivity by requiring that the wall of the bubble can be made transparent (Def.~\ref{def:transparent-wall}). 
Let us make a few remarks.
\begin{enumerate}[leftmargin=13pt]
    \item Both channel connectivity (Def.~\ref{def:CC}) and topological channel connectivity (Def.~\ref{def:TCC}) are equivalence relations. The latter is strictly finer, as we will see in examples below.
    \item The idea of having extra properties near the spatial boundary of the FDLC may look unfamiliar (or even strange). Nonetheless, we emphasize that similar considerations already occurred in the study of symmetry-protected topological phases (SPT)~\cite{xiechen_localunitary2010, PhysRevB.84.165139, Chen_2011}. Suppose we want to deform an SPT state to another by finite-depth local unitaries, then only deformations that respect the symmetry \emph{locally} are permitted. For instance, in the 1D cluster state, although the global application of controlled-$Z$ gates (a depth-2 circuit) respects the $\mathbb{Z}_2 \times \mathbb{Z}_2$ symmetry as a whole, a local collection of these gates on local regions breaks the symmetry and is therefore disallowed. This makes the 1D cluster state a nontrivial SPT, which enjoys anomalous behavior at its spatial boundary.
    Analogously, when dealing with topological mixed states \emph{protected} by long-range topological correlations or entanglements, we must restrict to deformations that preserve the underlying topological structure.  
\end{enumerate}
The topological channel connectivity (Def.~\ref{def:TCC}) is seemingly stronger than our definition of equivalent fixed point (Def.~\ref{def:p.f.-equivalence}). However, we do not know any pair of fixed point density matrices that are identical by Def.~\ref{def:p.f.-equivalence}, which do not have topological channel connectivity between them. This leads to the following conjecture:

\begin{axiombox}
\begin{Conjecture}[Phase equivalence with topological channel]\label{conj:topological-channel-equivalence}
Two mixed-state fixed points represent the same phase if and only if these two states are topologically channel-connected. 
\end{Conjecture}  
\end{axiombox}

To make the contrast between the channel connectivity and topological channel connectivity clearer, we consider a fixed point $\alpha$ and a trivial product state $\beta=\bmstate{1}$. Two types of boundaries are particularly natural to consider.

\begin{definition}[Dry and Wet boundaries]\label{def:dry-wet}
Consider a fixed point $\alpha$ and a trivial product state $\beta=\bmstate{1}$.
We call the domain wall of a bubble of $\bmstate{1}$, created on top of $\alpha$, the dry boundary (created by truncated channel $\Phi_{\alpha \to \bmstate{1}} $). Suppose there exists a local channel $\Phi_{\bmstate{1} \to \alpha} $ that converts $\bmstate{1}$ to $\alpha$; we call the domain wall of a bubble of $\alpha$, created on top of $\bmstate{1}$, the wet boundary (created by truncated channel $\Phi_{\bmstate{1} \to \alpha} $).
\end{definition}

Note that any mixed state fixed point $\alpha$ allows a dry boundary with the trivial product state $\bmstate{1}$. This is because we can always find the local reset or maximal depolarization channel that converts $\alpha$ to the trivial state $\bmstate{1}$. On the other hand, the existence of the wet boundary is possible only through channel connectivity (i.e. when $\Phi_{\bmstate{1} \to \alpha}$ exists and can be truncated to create a local bubble), which may or may not satisfy the stronger topological channel connectivity.

To be more explicit, let $\alpha\,{=}\,\rho^e = \prod_v A_v^{\bm{1}}$ be the dephased toric code fixed point~\eqref{eq:rho^e-T.C.} and ${\bmstate{1}}\,{=}\, |+\rangle^{\otimes N}$ be the product state. 
As we have discussed in the previous subsection, $\alpha$ and $\bmstate{1}$ are channel connected. The two channels are
\begin{itemize}[leftmargin=13pt]
    \item $\Phi_{\alpha\to \bmstate{1}}$ is the product of an on-site channel that resets each site to $|+\rangle$. The reset channel in this process is constructed by swapping a system qubit with an ancilla qubit $|+\rangle$ and tracing out the ancilla. 
    \item $\Phi_{\bmstate{1} \to \alpha}$ is the measurement of $\prod_{e \in p} Z_e$ and throw away the outcome on each plaquette.  
\end{itemize}
By restricting these channels to part of the plane, we create bubbles of one phase on the background of another phase. 
We illustrate the ``dry boundary'' and the ``wet boundary'' of the dephased toric code in Fig.~\ref{fig:dry-wet-bdy}.
 Before doing any computation, we know that {\bf M1} must be violated in the neighborhood of (any part of) the boundary. This is because $\gamma_{\textrm{LW}}(\alpha) = \log \sqrt{2}$, whereas  $\gamma_{\textrm{LW}}(\beta)=0$. If {\bf M1} were satisfied on any portion of the boundary, these two values would be the same. (See the discussion around Eq.~\eqref{eq:LW-difference}.)

\begin{figure}[h]
    \centering
    \includegraphics[width=0.85\linewidth]{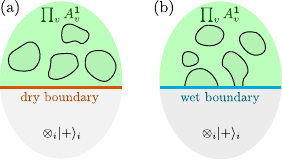}
    \caption{{\bf Dry and wet boundaries (Def.~\ref{def:dry-wet}):} 
    Two types of boundary between $\alpha=\prod_v  A_v^{\bm{1}}$ and $\bmstate{1}=\prod_i |+\rangle^{\otimes N}$  (a) The ``dry" boundary obtained by applying $\Phi_{\alpha \to \bmstate{1}}$ on the lower half-plane of $\alpha$. (b) The ``wet'' boundary  obtained by applying $\Phi_{\bmstate{1} \to \alpha}$ on the upper half-plane of $\bmstate{1}$. 
    These two boundary types differ from each other on whether there exist Pauli-$X$ string operators $\calO_S$ terminating on the boundary, such that $\calO_S \sigma = \sigma $. Such a distinction has physical implications as we explain in Fig.~\ref{fig:bdy-TE}.}  
    \label{fig:dry-wet-bdy}
\end{figure}

Explicit computation tells us which of the axioms  {\bf P0}, {\bf M0}, and {\bf M1} are preserved or violated on these boundaries; see Appendix~\ref{app:violation}. We state the result and comment on the physical interpretation.

\vspace{5pt} \noindent  \emph{Preserved conditions:} Both the dry and wet boundaries have a set of {\bf P0}, {\bf M0}, and {\bf M1}-like conditions preserved:  
    \begin{enumerate}[leftmargin=13pt]
        \item {\bf P0} and {\bf M0} are satisfied on the boundary, namely
        $I(A\,{:}\,C|B)=0 $ and $I(A\,{:}\,C)=0$ for Fig.~\ref{fig:bdy-axiom}(a).
         
    \item {\bf M1} is preserved on the boundary for some orientations. $I(A\,{:}\,C|B)= I(A\,{:}\,C|D) =0$ for Fig.~\ref{fig:bdy-axiom}(b).
    \end{enumerate}
    These conditions guarantee that the boundary is uniform in a sense that generalizes Ref.~\cite{Shi2020domainwall}.  There is no obstruction to transport information along the boundary.  If, instead, the boundary consists of segments jointed at point-like defects, we should expect the violation of some of these conditions. 

    \begin{figure}[h]
    \centering
    \includegraphics[width=0.72\linewidth]{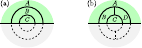}
    \caption{Axioms are preserved for these classes of configurations. (a) for {\bf P0} and {\bf M0}, and (b) for certain classes of {\bf M1}.}
    \label{fig:bdy-axiom}
\end{figure}

\vspace{5pt} \noindent  \emph{Violated conditions:} The dry boundary and the wet boundary both violate a certain type of {\bf M1}. The relevant partitions are shown in Fig.~\ref{fig:bdy-TE}(a). In fact, the violations give a pair of nontrivial topological entropies, and the two boundaries can be distinguished.
    We first define the pair of boundary topological entropies 
\begin{eqnarray}
    \gamma_{\rm dry} &:=& \frac{1}{2} I(A\,{:}\,C|D), \quad \text{for Fig.~\ref{fig:bdy-TE}(a)}. \\
    \gamma_{\rm wet} &:=& \frac{1}{2} I(A\,{:}\,C|B), \quad \text{for Fig.~\ref{fig:bdy-TE}(a)}. 
\end{eqnarray}
The computation result summarized in Table~\ref{tab:TE-dry-wet} justifies the name. Namely, $\gamma_{\rm dry}$ detects a dry boundary and $\gamma_{\rm wet}$ detects a wet boundary.  Note that $\gamma_{\rm dry}$ and $\gamma_{\rm wet}$ are the only forms of violation of {\bf M1} near the boundaries, as $\beta$ is trivial.
\begin{table}[h]
    \centering
    \begin{tabular}{|c|c|c|}
    \hline
    boundary type & $\gamma_{\rm dry}$ & $\gamma_{\rm wet}$\\
    \hline
     dry boundary    & $\gamma_{\textrm{LW}}$   & $0$  \\
     \hline
      wet boundary    & $0$   & $\gamma_{\textrm{LW}}$  \\
      \hline
    \end{tabular}
    \caption{Here $\gamma_{\textrm{LW}}$ is that for the nontrivial phase $\alpha$ being  the $\prod_v A_v^{\bm{1}}$ fixed point of the dephased toric code. Here $\gamma_{\textrm{LW}}= \log \sqrt{2}$ is the Levin-Wen topological entropy of $\alpha$.}
    \label{tab:TE-dry-wet}
\end{table}

\begin{figure}[h]
    \centering
    \includegraphics[width=0.99\linewidth]{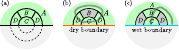}
    \caption{(a) Configurations to check the violation of the {\bf M1}. The resulting values are topological entropies $\gamma_{\rm dry}$ and $\gamma_{\rm wet}$. These regions can be restricted to the upper half plane, as the state is a product in the lower half plane. (b) The intuition why the dry boundary has $I(A\,{:}\,C|D) =\log 2$  and (c) the wet boundary has nontrivial $I(A\,{:}\,C|B) =\log 2$. The gray closed (open) string is the support of a nontrivial string operator $\calO_S$ for the dry (wet) boundary, such that its action on the state with the given boundary condition is $\calO_S \sigma\,{=}\, \sigma$. }
    \label{fig:bdy-TE}
\end{figure}
 
Let us comment on the intuition why the dry boundary gives $I(A\,{:}\,C|D)= \log 2$. 
There exists a $\mathbb{Z}_2$-valued nontrivial operator $\calO_S\,{=}\,\prod_{i \in S} X_i$ supported on a closed string $S\,{\subset}\,ACD$ that stabilizes the state, i.e., $\calO_S \sigma \,{=}\, \sigma$. However, its eigenvalue cannot be determined by either marginal state $\sigma_{CD}$ or $\sigma_{AD}$. 
Therefore, conditioning on $D$ leaves a single $\mathbb{Z}_2$ variable shared between $A$ and $C$, contributing exactly one bit to the conditional mutual information, see Fig.~\ref{fig:bdy-TE}(b).

Similarly, for the wet boundary, there exists a nontrivial operator $\calO_S$ supported on an open string $S\,{\subset}\, ABC$ that stabilizes the state. However, its eigenvalue cannot be determined by the marginal states $\sigma_{AB}$ and $\sigma_{BC}$, as shown in  Fig.~\ref{fig:bdy-TE}(c). Again, conditioning on $B$ leaves a relevant degrees of freedom shared between $A$ and $C$. See \cite{Levin2013boundary,Kim2014-low-energy,Shi2018seeing} for closely related ideas.

 It is amusing to contrast the pure state context.
For a gapped boundary between a pure state topological order and the vacuum, we always have $\gamma_{\rm dry}=\gamma_{\rm wet} = \frac{1}{2} \gamma_{\textrm{LW}}$. Thus, we cannot distinguish any pairs of pure state gapped boundaries with such topological entropies. In pure states, more interesting topological entropies only exist for domain walls between a pair of nontrivial phases \cite{Lan2015-domain-wall,Shi2021-DWTEE}.
For mixed state phases, we have richer physics.

    \begin{figure*}[!t]
  \centering
  \includegraphics[width=0.98\linewidth]{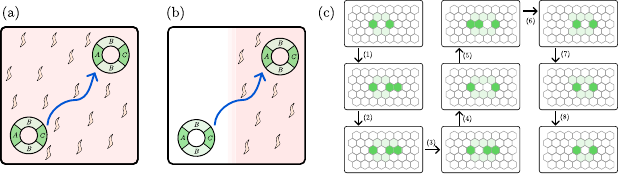}
        \caption{ {\bf Stability of topological invariant.} {\bf (a) Decohered state.} When the three axioms are approximately satisfied due to being away from the fixed point, still topological entropy remains invariant up to an exponentially small correction, no matter where we evaluate.
        {\bf (b) Connection to fixed-points.}  A similar logic can be applied if we have a smooth interpolation of a fixed-point density matrix satisfying three axioms and a decohered density matrix with approximate axioms. Two states must show the same topological entropy up to small corrections that are exponentially suppressed under coarse-graining.
        {\bf (c) Sequence of moves}. On a honeycomb coarse grained lattice at length scale $r$, we consider two Levin-Wen partitions, relatively shifted by 1 lattice, showing in the 1st and the last figure. Their associated values of topological entropies are $\gamma_{\textrm{LW}}^{\rm(left)}$ and $\gamma_{\textrm{LW}}^{\rm(right)}$. The eight steps represent an explicit sequence of deformation of the partition.  
        }
  \label{fig:deformations-explicit}
\end{figure*}

The distinction between the wet and dry boundaries in Table.~\ref{tab:TE-dry-wet} encodes a genuine memory effect: each boundary retains information about how it was created, reminiscent of hysteresis in a magnetic system. At a deeper level, one may ask how difficult it is to convert one boundary type into the other using FDLCs acting only near the interface. This question closely parallels the stability analysis of the FDLC that prepares the bulk phase, discussed in Sec.~\ref{subsec:coarse-graining}.~2, where it was discussed that a slightly imperfect channel ($p \neq 1/2$) fails to prepare the intended phase. More generally, if the preparation channel is noisy, does the resulting state remain in the same phase, or does it instead realize a different phase outside the bubble? In both cases, one encounters an apparent asymmetry. The relationship between these two cases is subtle and will be investigated in future work.

\section{Approximate Axioms and Stability of Topological Invariants}\label{Sec:robust-TE}

Our axioms provide a rigorous foundation for discussing fixed points of non-equilibrium phases of matter under decoherence. To compare states beyond these fixed points, we permit controlled violations of the axioms that decay exponentially with the size of the partition and define phase equivalence accordingly (Definitions~\ref{def:p.f.-equivalence}); see \figref{fig:1}(b). However, for this definition of phase equivalence to be meaningful, the relaxation of axioms should preserve the topological character of invariants---including the information convex set, topological entropy, and related quantities elaborated in \appref{app:ICS-Omega}.

In this section, we establish the \emph{stability} of these topological invariants under such relaxed conditions: if the axioms hold up to exponentially decaying corrections, the invariants are unchanged up to likewise corrections exponentially small in partition size. 
Concretely, we assume the violations of {\bf P0}, {\bf M0}, and {\bf M1} obey the common bound
\begin{equation}\label{eq:violation_approx}
\delta(r) \le \mu \, e^{-r/\xi}
\end{equation}
for length scales $r \gg \xi$, where $\xi:=\max\{\xi_{\mathrm{corr}},\xi_{\mathrm{Markov}}\}$ is the larger of the correlation length and the Markov lengths, and $\mu$ is an ${\cal O}(1)$ constant. This scaling is observed in our data for realistic decoherence (Fig.~\ref{fig:violations}; see also Appendix~\ref{app:violation}).

We consider the 
regime that the length scales $r\gg \xi$, and describe the system by a coarse-grained honeycomb lattice; see Fig.~\ref{fig:deformations-explicit}(c). 
The Levin-Wen partition is applied to a single coarse-grained unit cell. We then ask how the associated topological entropy varies with a length scale $r$. Concretely, we compare two adjacent placements of the partition that differ by a one-cell translation on the coarse-grained lattice (denoted ``left'' and ``right'') and seek an $r$-dependent bound  $|\gamma_{\textrm{LW}}^{\rm(left)} - \gamma_{\textrm{LW}}^{\rm(right)}| \le f(r)$, where $f(r) \rightarrow 0$ as $r \rightarrow \infty$.

\begin{Proposition}\label{prop:deformation}
The difference between the Levin-Wen TE on the left $\gamma_{{\rm LW}}^{\rm(left)}$ and the partition on the right $\gamma_{{\rm LW}}^{\rm(right)}$ is bounded as 
\begin{equation}
    |\gamma_{\rm{LW}}^{\rm(left)} - \gamma_{\rm{LW}}^{\rm(right)}| \le 8\mu \, e^{-r/\xi}.
\end{equation}
In particular, the error becomes exponentially smaller at larger length scales even though the two partitions are separated by a larger distance.
\end{Proposition}

\begin{proof}
We need to estimate the entropy difference of each pair of nearby configurations; there are 8 such pairs, as in Fig.~\ref{fig:deformations-explicit}(c). 

Now let us show that 
the deformation in the 1st step gives a change at most $O(1) \, e^{-r / \xi}$. The difference of Levin-Wen topological entropy for the 1st step is the difference for the partition $A,B,C$ and $A,B,CC'$ in
\begin{equation}
    \includegraphics[width=0.54\linewidth]{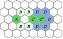}
\end{equation}
which is computed as
\begin{equation}
\begin{aligned}
    |\delta \gamma_{\textrm{LW}}^{\rm step-1} | &:= I(A:CC'|B) - I(A\,{:}\,C|B) \\
    & = (S_{AB} + S_{BCC'} - S_B - S_{ABCC'}) \\
    &~\,\,\,\, - (S_{AB} + S_{BC} - S_B - S_{ABC}) \\
    &= S_{ABC} + S_{BCC'} - S_{BC} - S_{A BC C'} \\
    &= I(A:C'|BC)\\
    & \le I(AB:C'|C) \\
    & \le \delta(r).
\end{aligned} 
\end{equation}
The last two lines use the strong subadditivity, and the last line uses the error of {\bf M1} is given by $\delta(r)$. The same analysis applies to steps 2-8, with $|\delta \gamma_{\textrm{LW}}^{\rm step-k} | \le  \delta(r)$ for all $k=2,...,8$. 
While we omit the details, we list the equalities that are used,
\begin{equation}
\begin{aligned}
     I(A:CC'|B)- I(A\,{:}\,C|B)  
      = I(A:C'|BC), 
\end{aligned}
\end{equation}
and
\begin{equation}
\begin{aligned}
     &I(A\,{:}\,C|B)- I(A:C|BB')  \\
     & \quad = I(A:B'|B) - I (A:B'|BC). 
\end{aligned}
\end{equation}
With these, we 
concludes the proof.
\end{proof}

\begin{corollary}[Stability of topological entropy] \label{Coro:invariant}

Consider a fixed point density matrix $\rho_\textrm{fix}$ that exactly satisfies three axioms. Assume $\rho$ belongs to the same phase in that one can construct a smooth spatial interpolation between $\rho$ and $\rho_\textrm{fix}$ where the violation of axioms is exponentially suppressed with scale $r$ as in \eqnref{eq:violation_approx}, see \figref{fig:deformations-explicit}(b). Then the topological entropies of $\rho$ and $\rho_\textrm{fix}$, computed on two Levin-Wen partitions separated by a constant number of deformation steps, are identical up to an exponentially small error at most $\calO(1) e^{-r/\xi}$.
\end{corollary}

Therefore, our axioms indeed play the role of gap---as long as two states are connected while maintaining these axioms either exactly or approximately, topological entropies remain invariant up to exponentially small corrections. In fact, all other topological invariants described in the isomorphism theorem can be shown to exhibit only exponentially small corrections. We provide a physical sketch in Appendix \ref{app:robust-iso}; the details will appear in a future dedicated work.

\section{Secret Sharing: Hierarchy of Topological Memory}\label{sec:Secret-sharing}

In the broad family of topological mixed states, both classical
 and quantum logical information can be encoded in a topological manner~\cite{Dennis_2002, Castelnovo_2007}. This encoding implies that the logical information cannot be retrieved through local operations; rather, one must access a topologically non-trivial subregion of the system to recover it. A familiar example is the toric code~\cite{Dennis_2002}, where logical operators are supported on nontrivial cycles of the torus. Put differently, if we partition the system as in \figref{fig:secret-sharing}(a) among four parties, recovery of the logical information requires collaboration between at least two parties. 

Remarkably, even at the level of classical information, different topological mixed states exhibit distinct capacities for secret protection under this partition of the torus. For instance, in the canonically decohered Abelian quantum double defined in Definition~\ref{def:canonical}, access to any two pairs (e.g., $AB$ and $AC$, $AD$ and $BC$, etc.) suffices to fully reconstruct the logical information stored along the two nontrivial cycles of the torus. However, as we shall see, for decohered non-Abelian topological mixed states, this level of access is insufficient. Even more intriguingly, a hierarchy of secret-sharing capacities emerges among different non-Abelian topological mixed states. We will illustrate this phenomenon through a detailed analysis of the $S_3$ and $A_4$ quantum doubles.

\subsection{Non-Abelian quantum double}

Before going into details, we define the decohered states of non-Abelian quantum doubles~\cite{Kitaev1997,Bombin2008}. 
For a non-Abelian group $G$, one can define the following projectors in a way analogous to the Abelian case in \eqnref{eq:DFT}:
\begin{itemize}[leftmargin=13pt]
    \item $A_v^R$ on the charge $R$ (irreducible rep.) at vertex $v$
    \item $B_p^C$ on the flux $C$ (conjugacy class) at plaquette $p$
\end{itemize} 
The notation of quantum double models is reviewed in Appendix~\ref{app:QD}. For both the trivial irreducible representation and the conjugacy class of identity element, which correspond to trivial charge and flux, respectively, we use the label $\bm{1}$. The quantum double ground state with maximally mixed logical space is then given as $\prod_v A_v^{\bmcharge{1}} \prod_p B_p^{\bmflux{1}}$.

In the Abelian case, we can start from the quantum double state and decohere it maximally in charge or flux, respectively, to obtain different topological mixed states (Def.~\ref{def:canonical}). 
However, it turns out that the situation is different in non-Abelian cases.

\vspace{5pt} \noindent \emph{Decohering into $\rho^m$.} We apply the channel described by Kraus operators $\{T_e^g|g\in G\}$ where $T^g_e$ is a projector onto $|g \rangle$ on edge $e$; therefore this channel corresponds to decohering non-Abelian charge. 
As $[B_p^{\bmflux{1}}, T^g_e] = 0 \quad\forall h, g\in G$, $B_p^{\bmflux{1}}$ survives under this channel, while only the identity part survives in the expansion of $\prod_v A_v^{\bmcharge{1}}$.
Thus, we eliminate $\prod_v A_v^{\bmcharge{1}}$ from the pure state and obtain $\rho^m= \prod_p B_p^{\bmflux{1}}$.

As one might expect, this charge-decohered non-Abelian state $\rho^m$ satisfies three axioms {\bf P0}, {\bf M0}, and {\bf M1}.  The argument to prove the axioms is similar to that of the Abelian case. The only subtlety is that the stabilizer measurements can give non-Abelian flux excitations when we prove {\bf P0} and {\bf M1}. However, partial trace always gives the constraint (c.f. \eqnref{eq:traceout} and \eqref{eq:traceout-2}) that these fluxes must fuse into a trivial total flux as in the Abelian case. Therefore, the original stabilizers can always be recovered.

The topological mixed state $\rho^m$ represents a classical mixture of loops subject to non-Abelian conservation laws, specifically the zero flux condition:
\begin{equation}\label{eq:fp-classical-loop}
  \rho^m = \prod_p B_p^{\bmflux{1}} \propto \sum_c |c\rangle \langle c|,
\end{equation}
where $|c\rangle$ is a product state in the group element basis satisfying $B_p^{\bmflux{1}} |c\rangle =|c\rangle$. Since the state $|1 \rangle^{\otimes N}$ ($1\in G$) satisfies this condition, all other loop configurations can be generated by applying products of unitary operators $A_v^g$. For a manifold like $\mathbb{T}^2$, we can further introduce two fluxes around non-contractible cycles, which are trivial for loops generated in this way. 
To create a loop configuration with non-trivial flux through non-contractible cycles, a special operator should be applied along the cycle, as shown in \figref{fig:loopNA}(a). It is important to note that fluxes are only defined up to conjugacy class, meaning that two fluxes, $g$ and $h^{-1} g h$, are equivalent.

\begin{figure}[h]
    \centering
    \includegraphics[width=0.96\linewidth]{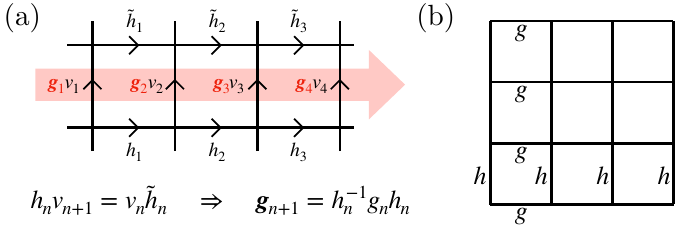}
    \caption{{\bf Decohered Non-Abelian state.} (a) The operator along the \emph{dual} horizontal non-contractible cycle $L$, $\prod_{e \in {\color{red} L}} L_e^{g_e}$ where $g_{1} = g$ and $g_{{n+1}} = h_n^{-1} g_n h_n$.  
    Applied on a loop state with trivial flux along the \emph{vertical} cycle, this generates a conserved loop state with non-trivial flux $C_g$ through the vertical non-contractible loop. (b) Initial product state $\lambda_{g,h}$ to prepare an extreme point of the information convex set $\Sigma(\mathbb{T}^2|\rho^m)$. The square lattice is on a torus with a periodic boundary condition. Unlabeled links are in the identity element.}
    \label{fig:loopNA}
\end{figure}

We are ready to study the information convex set. On a lattice of the torus, define $\lambda_{g,h}$ to be the zero flux configuration shown in Fig.~\ref{fig:loopNA}(b), for $gh=hg$.
We further define
\begin{equation} \label{eq:Cgh}
    C_{g,h} := \{ (r g r^{-1}, r h r^{-1}) \,|\, r \in G \},  
\end{equation}
which is the set of simultaneous conjugacy classes of pairs $(g,h)$ with $g,h \in G$. 

\begin{Proposition}\label{prop:simplex-classical-loop}
    The information convex set $\Sigma(\mathbb{T}^2| \rho^m)$ on the torus $\mathbb{T}^2$ is a simplex with extreme points in one-to-one correspondence with $C_{g,h} $ for $g h = hg$. These extreme points are mutually orthogonal and have the explicit form  
    \begin{align}\label{eq:all-extreme-pts}
    \rho^{C_{g,h}} := \frac{1}{|G|^{N_v}} \sum_{\{ g_v \} } \, \prod_v A^{g_v}_v \lambda_{g,h} A^{g_v \dagger}_v
\end{align} 
Furthermore, these extreme points can be prepared by a finite-depth local circuit from the product state.
\end{Proposition}  
    
\begin{proof}
    The proof is analogous to the one in the Proposition~\ref{prop:convexset}. Following similar logic, we can show that for any $\rho \in \Sigma(M|\rho^m)$, it satisfies
    \begin{align}
        T^g_e \rho &= \rho T_e^g \nonumber \\
        B_p^{\bmflux{1}} \rho &=  \rho B_p^{\bmflux{1}} = \rho  \nonumber \\
        A_v^{g} \rho &= \rho A_v^{g}.
    \end{align}
    For any $\alpha \in \Sigma$, the first condition implies that $\alpha = \sum_c \lambda_c|c\rangle \langle c|$,
    where $|c\rangle$ is a configuration of group elements on each link. The second condition implies that every $|c \rangle$ must have zero flux locally, that is $B_p^{\bmflux{1}} |c\rangle =|c\rangle$. The third condition implies that $\alpha = \calA(\alpha)$ where
    \begin{equation}
        \calA(\calO) :=  \frac{1}{|G|^{N_v}} \sum_{\{ g_v \} } \, \prod_v A^{g_v}_v  \calO A^{g_v \dagger}_v.
    \end{equation}
    Therefore, for any $\alpha \in \Sigma$, we can write
\begin{equation}
    \alpha = \sum_C \lambda_c \, \calA(|c\rangle \langle c|).
\end{equation}
Next, $\calA(|c\rangle \langle c|) =\calA(|c'\rangle \langle c'|)$ if $|c \rangle$ and $|c'\rangle$ can be related by some {``gauge transformation''} $\prod_v A_v^{g_v}$, which we denote as $|c\rangle \sim |c'\rangle$. For any configuration $|c \rangle$ with zero local flux, one can always find a gauge transformation into $\lambda_{g,h}$ (the product state of the form in \figref{fig:loopNA}(b)), for some commuting pair $gh=hg$; this is a fact we explain in Appendix~\ref{app:non-Abelian-justify}.
Thus, any element in $\Sigma(M|\rho^m)$ is given as a convex combination of states of the following extreme points:
\begin{align} 
    \rho^{C_{g,h}} :=   \calA( \lambda_{g,h} ).
\end{align}
Since the conjugation by $A_v^r$ generates all possible configuration in the class $C_{g,h}$ in \figref{fig:loopNA}(b), the above state is indeed labeled by the class $C_{g,h}$ in \eqnref{eq:Cgh}. Thus, \eqref{eq:all-extreme-pts} describes every extreme point of $\Sigma(\mathbb{T}^2| \rho^m)$. The extreme points $\rho^{C_{g,h}}$ are a classical mixture of configurations, and they are mutually orthogonal. Furthermore, since $\lambda_{g,h}$ is a product state and the above operation is equivalent to applying a mixed unitary channel~\footnote{A mixed unitary channel is a quantum channel that can be expressed as a convex combination of unitary operations.} at every vertex, this provides a classical preparation protocol. 
\end{proof}

From the formula \eqref{eq:all-extreme-pts}, we see that extreme points are classical separable states. If the group $G$ is Abelian, then $C_{g,h}=\{(g,h)\}$, and the number of extreme points would be $|G|^2$. If the group $G$ is non-Abelian, the number of $C_{g,h}$ equals the number of anyons in the quantum double models, which is less than $|G|^2$~\cite{Bombin2008}. It reflects the fact that non-Abelian fluxes along two non-contractible cycles cannot be fully independent.  The entropies of the extreme points are 
\begin{equation}
    S(\rho^{C_{g,h}}) = S(\rho^{C_{1,1}}) + \log | C_{g,h}|.
\end{equation}
This formula is obtained by computing the sizes of equivalent classes of the configurations, as detailed in Appendix~\ref{app:non-Abelian-justify}.
While the entropy of each extreme point follows a volume law, the entropy difference $S(\rho^{C_{g,h}}) - S(\rho^{C_{1,1}}) = \log | C_{g,h}|$ has a topological origin.

This further implies that the memory capacity
\begin{equation}\label{eq:capacity}
    \calQ(\mathbb{T}^2|\rho^m) = \log |G| + \log |G_{\rm cj}|,  
\end{equation}
where $|G_{\rm cj}|$ is the number of conjugacy classes of $G$. A quick intuitive way to see this is noticing that the maximum entropy is achieved by mixing all extreme points with proper weights. To maximize $S(\rho)$, we can do it by uniformly mixing all the valid elements in the extreme points. To minimize $S(\lambda)$, we pick the extreme point corresponds to $ C_{1,1}$, where $C_{1,1}= \{ (1,1) \}$ contains only one element. Therefore, we arrive at $ \calQ(M|\sigma) = \log N$ where 
\begin{equation}
    N = \{\#(g, h)| gh = hg, \;g,h\in G\}.
\end{equation}
For every element $c$ in a conjugacy class $\calC$, the number of elements in group $G$ that commutes with $c$ is $\frac{|G|}{|\calC|}$. And the number of conjugacy class is $|G_{\rm cj}|$, therefore:
\begin{equation}
    N = \sum_{\calC\in G_{\rm cj}}\frac{|G|}{|\calC|}|\calC| = |G||G_{\rm cj}|.
\end{equation}
For more systematic derivations, see~\appref{app:derivation_memory}.

The Levin-Wen topological entropy for this model is
\begin{equation}\label{eq:LW-TE-G}
    \gamma_{\textrm{LW}} = \log \sqrt{|G|}.
\end{equation} 
The question reduces to computing the entanglement entropy of the fixed point density matrix for disks and annuli, as is done in Appendix~\ref{app:Levin-Wen}.  
Because $|G_{\rm cj}|< |G|$ for any non-Abelian model, the bound on memory capacity in Eq.~\eqref{eq:bound-mixed} is saturated if and only if $G$ is Abelian.

\vspace{5pt} \noindent \emph{Decohering into $\rho^e$.} Unlike $\rho^m$, obtaining $\prod_v A_v^{\bmcharge{1}}$ is not straightforward. Naively, we attempt to decohere fluxes in a fashion analogous to the Abelian case, applying mixed unitary channels in terms of $L^g_e$ where $L^g_e |h \rangle = |gh \rangle$ on edge $e$. Since this does not commute with $A^{R}_v$, we define $L^R_e := \frac{\dim R}{|G|}\sum_{g\in G } \chi_R(g) L^g_e$ for each irreducible representation $R$ of $G$, which gives $[L^R_e, A^{R}_v] = 0$. Therefore, we may use $\{ L^R_e \}$ as the local channel performing projective measurements. 
But applying this channel generically cannot remove $\prod_p B^{\bmflux{1}}_p$ unlike the Abelian case. The difference of charge decoherence and flux decoherence is also evident in the fusion rules of charges and fluxes. In non-Abelian quantum double models, it is no longer true that fusing two fluxes always results in a flux (e.g., in the $S_3$ quantum double \cite{Beigi2011}), whereas fusing two charges still always yields a charge.

While we leave a systematic study of generic flux-decohered non-Abelian states to future work, it is worth noting a related example. Ref.~\cite{sala2024stability} studied a non-Abelian quantum double model under a specific flux-decohering channel and found evidence that non-Abelian topological order can remain robust even at maximal decoherence. We expect that this robustness is tied to the nontrivial fusion structure of the particular fluxes being proliferated. More specifically, we conjecture that projecting onto the trivial charge pattern described by $\rho^e$ requires proliferating all flux types, rather than only a single flux species. At present, however, we do not have a general proof of this claim.

\subsection{Secret sharing models}

What distinction between (charge) decohered Abelian and non-Abelian models can we make? One notable difference is the secret sharing~\cite{shamir1979share,cleve1999share,Kato2015,Fiedler2017}, which refers to the scheme to distribute information (or secret) among a group in such a way that only nontrivial collaboration of multiple parties can decode the information.
Let us partition the torus $\mathbb{T}^2$ into four disk-like patches $ABCD$ as shown in Fig.~\ref{fig:secret-sharing}. The question is, what can we learn given pieces of the state?

As a topological memory, two properties immediately follow: $(i)$ Given any three parties (e.g. $\rho_{ABC}$), we can recover the entire state, and the knowledge is complete. $(ii)$ No single party can recover any information. 
But what if we allow only the collaboration between two parties $(A,B)$, $(B,C)$ and $(C,A)$? The two parties can do any joint operations. Suppose we have enough copies of the state, then this corresponds to knowing the reduced density matrices $\rho_{AB}$, $\rho_{BC}$, and $\rho_{CA}$. To be concrete, we shall consider the torus partitioned into $A,B,C$, and $D$ in Fig.~\ref{fig:secret-sharing}.

\begin{figure}[!t]
    \centering
\includegraphics[width=0.75\linewidth]{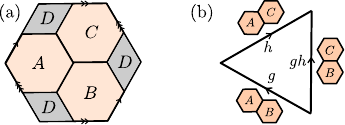} 
\caption{(a) Partition a torus into 4 disk-like pieces $A,B,C$ and $D$. The torus is visualized as a big hexagon with opposite edges identified, as in Fig.~1.2 of Ref.~\cite{thurston2000three}. (b) A schematic drawing of the conjugacy classes detected by the collaboration of two parties (thickened non-contractable loops of the torus), $(A,B) \to C_g$, $(A,C)\to C_h$ and $(B,C)\to C_{gh}$.}\label{fig:secret-sharing}
\end{figure}

This can be converted to a group-theoretic computation. Recall that 
\begin{align}
    C_g &:=\{rgr^{-1} \,|\,r \in G \} \nonumber \\
    C_{g,h} &:= \{ (r g r^{-1}, r h r^{-1}) \,|\, r \in G \}
\end{align}
where $C_{g,h}$ is defined for $gh\,{=}\,hg$.
Suppose we store the information in the extreme point $\rho^{C_{g,h}}$. $AB$ knows $C_g$, $AC$ knows $C_h$, and $BC$ knows $C_{gh}$. When $G$ is Abelain, $C_{g,h}$ is determined from $g,h$ and thus knowing $\rho_{AB}$ and $\rho_{BC}$ is enough.

However, for non-Abelian cases, the collaboration between any two parties may not determine the message. Finding an example is equivalent to solving the following problem. Consider a finite group $G$ which provides a triple $(g, h, h')$, $g, h, h' \in G$, such that
\begin{equation}
    gh = hg, \quad g h' = h' g,
\end{equation}
and $h' \in C_h$, $gh' \in C_{gh}$, which nonetheless, has $C_{g,h} \ne C_{g,h'}$.  

Finding such a group $G$ corresponds to finding a model with information not decodable with the collaboration of any two parties. 

  One such group is the 4-element alternating group $A_4$.  It has 12 group elements, and there are three elements of order 2:
    \begin{equation}
        \alpha :=(12)(34), \quad \beta :=(13)(24) ,\quad \gamma:= (14)(23).
    \end{equation}
    $\alpha^2=\beta^2=\gamma^2=1$,  $\{\alpha,\beta,\gamma\}$ are mutually commuting and \begin{equation}
        \alpha \beta = \gamma, \quad \gamma \alpha =\beta, \quad \beta \gamma =\alpha.
    \end{equation}
    $\{\alpha,\beta,\gamma\}$ forms a conjugation class of $A_4$.  
One can check:
$C_{\alpha,\beta}= \{ (\alpha,\beta), (\beta,\gamma), (\gamma,\alpha)\}$ and $C_{\alpha,\gamma}= \{ (\alpha,\gamma), (\beta,\alpha), (\gamma,\beta)\}$. 
Thus, if we let $g=\alpha$, $h = \beta$, and $h'=\gamma$, all the conditions that enable the nontrivial secret sharing are satisfied. This means that the two extreme points associated with $ C_{\alpha,\beta}$ and $C_{\alpha,\gamma}$ can be distinguished only by the collaboration of 3 parties.

\begin{table}[!t]
    \centering  
    \renewcommand{\arraystretch}{1.1}
    \begin{tabular}{|c|c|c|c|c|}
     \hline
        $G$ & $\rho_{AB}$  & $\rho_{AB}, \rho_{BC}$ & $\rho_{AB}, \rho_{BC}, \rho_{CA}$ & $\rho_{ABC}$ \\
        \hline  
        Abelian  & No & Yes  & Yes  & Yes \\ \hline
        $S_3$ & No  & No & Yes & Yes \\ \hline
         $A_4$  & No & No & No & Yes \\
         \hline
    \end{tabular}
    \caption{Secret sharing properties for models with different groups $G$. ``Yes'' and ``No'' refer to the answers to whether the complete logical information on the torus can be learned.}
    \label{tab:secret-sharing}
\end{table}

Do all decohered non-Abelian quantum doubles have this property? In fact, the decohered $S_3$ quantum double provides another interesting example sitting between any Abelian group and $A_4$. In this model, we do not need $\rho_{ABC}$; instead, the complete logical information can be learned given $\{ \rho_{AB}, \rho_{BC}, \rho_{CA} \}$, but not from any two of the three. The result is summarized in Table~\ref{tab:secret-sharing}.

\section{Higher Dimension}\label{sec:3D}

Let us consider the 3D models based on an Abelian group $G=\mathbb{Z}_2$, adopted from the 3D toric code. 
We consider two reference states of the decohered 3D toric code
\begin{equation}
    \rho^e = \prod_v (1+ \prod_{e \ni v} X_e) , \quad \rho^m = \prod_p (1+ \prod_{e \in p} Z_e).
\end{equation}
These reference states satisfy the 3D version of the axioms. By the 3D version of the axioms, we mean the conditions centered on bounded radius 3D balls, which are the rotated versions of $BC$ or $BCD$ in the 2D axiom figure.

     \begin{figure}[!t]
         \centering
         \includegraphics[width=0.80\linewidth]{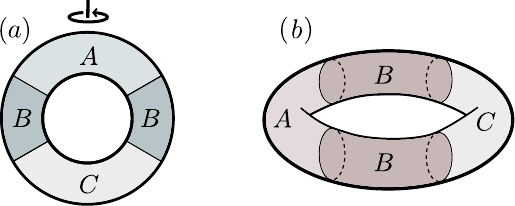}
         \caption{Definition of 3D TEs (a) A partition of a sphere shell that defines $\gamma_{\text{particle}}:= I(A\,{:}\,C|B)/2$ and (b) A partition of the solid torus that defines $\gamma_{\text{flux}}:= I(A\,{:}\,C|B)/2$.}
         \label{fig:TE-3D-definition}
     \end{figure}

     These two fixed points $\rho^e$ and $\rho^m$ must be in two different phases. As discussed in \secref{sec:MS-phases}, if the two states can be separated by a domain wall that preserves all the axioms {\bf P0}, {\bf M0}, and {\bf M1} (i.e. if the domain is transparent), the TE must match. However, the topological entropies summarized in Table~\ref{tab:TE-3D} show the mismatch, implying that two states belong to different topological mixed-state phases. 

     The physical meaning of the two types of topological entropy ($\gamma_{\text{particle}}$, $\gamma_{\text{flux}}$) in Fig.~\ref{fig:TE-3D-definition} is suggested by their names, and is studied in ground states of topological order \cite{Grover2011,Huang2021,Swingle2016}. $\gamma_{\text{particle}}$ measures the total quantum dimension of point particles, and $\gamma_{\text{flux}}$ measures the total quantum dimension of flux loop excitations. Importantly, for any pure state topological order in 3-dimensions, $\gamma_{\text{particle}}=\gamma_{\text{flux}}$. The values of TEs in Table~\ref{tab:TE-3D}, not only suggest that the two phases represented by the fixed points $\rho^e$ and $\rho^m$ are distinct mixed state phases, they also suggest that they are different from any pure state topological order in 3D, by our definition.

\begin{table}[h]
    \centering
    \renewcommand{\arraystretch}{1.3}
     \begin{tabular}{|c|c|c|c|c|}
    \hline
      & $\rho^{\textrm{t.c.}}$ & $\rho^e$ & $\rho^m$ & \, $\otimes_i \rho_i$ \, \\
     \hline
       $\gamma_{\text{particle}}$   & \, $\log \sqrt{2}$ \, & \, $\log \sqrt{2}$ \,  & $0$  & $0$ \\
       \hline
       $\gamma_{\text{flux}}$   & \, $\log \sqrt{2}$ \, & $0$ & \, $\log \sqrt{2}$ \, & $0$  \\
       \hline
    \end{tabular}
    \caption{Topological entropies for 3D models.}
    \label{tab:TE-3D}
\end{table}

   Now let us study the behavior of  $\rho^e$ and  $\rho^m$ as memory on closed manifolds. Consider the 3D manifold $M = S^1 \times S^2$. A simple computation shows 
    \begin{itemize}[leftmargin=13pt]
        \item $\Sigma(M|\rho^e)$  is a simplex with two extreme points
        \item $\Sigma(M|\rho^m)$ is a simplex with two extreme points
    \end{itemize} 
    In both sets, the two extreme points are orthogonal, and the entropy of each extreme point is identical. They both represent a classical memory.
  Can we tell the difference between the two classical memories? There is a simple answer: The secret-sharing schemes are different! 

    \begin{figure}[h]
        \centering
        \includegraphics[width=0.75\linewidth]{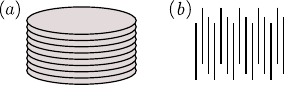}
        \caption{Partition $M=S^2\times S^1$ into pieces to create secret sharing schemes.  (a) Partition $M$ into thickened $S^2$ by dividing $S^1$ into pieces, and (b) Partition $M$ into thickened $S^1$ by dividing $S^2$ into small disks.}
        \label{fig:3D-fibers}
    \end{figure}

We can divide the manifold $S^1 \times S^2$ in two different ways, as illustrated in Fig.~\ref{fig:3D-fibers}.
    For the classical memory $\Sigma(M|\rho^e)$, the information can be decoded on each thickened $S^2$, but no information can be decoded in thickened $S^1$. For the classical memory $\Sigma(M|\rho^m)$, the information can be decoded on each thickened $S^1$, but no information can be decoded in thickened $S^2$.
    
It is further possible to consider fracton models with decoherence. However, we expect such states to violate axiom {\bf M1}. The intuition behind this is that {\bf M1} guarantees a topological sense of smoothness (the information of excitation can be transported in a topological manner), but for fractons, the excitations are free to move only in submanifolds.

\section{Conclusion}\label{sec:conclusion}

In this work, we developed a systematic framework for understanding topological mixed states by formulating a set of fixed-point axioms. These axioms, grounded in the spirit of the entanglement bootstrap~\cite{Shi2019fusion} and the recent study of Markov length~\cite{Sang2024}, encapsulate fundamental features we expect of any mixed state exhibiting topological correlation or entanglements. Each axiom imposes a constraint reflecting the local recoverability, stability, and consistency of the topological mixed states.

From the bootstrap axioms, we could derive several topological invariants that partially characterize a given topological mixed state phase. The primary diagnosis is the information convex set $\Sigma(M|\sigma)$, a set of density matrices locally indistinguishable from the fixed point $\sigma$ on $M$. On a closed manifold $M$, $\Sigma$ tells us the classical or quantum nature of information encoded in the topological mixed states. 
We identified classes of solvable mixed-state fixed points, showing that they satisfy all the axioms. We can prepare classes of classical fixed points efficiently via finite-depth local channels.
The associated $\Sigma$ is shown to exhibit classical topological memory, marking a sharp contrast to the fully coherent structure of pure topological states.
Another important diagnosis is the topological entropy, which is a straightforward generalization of Levin-Wen topological entanglement entropy. Interestingly, the topological entropy gives a fundamental bound to the memory capacity.

Crucially, our axiomatic framework allows us to go beyond fixed points: by requiring the axioms to hold only under coarse-graining, we explored phases of topological mixed states and provided a notion of equivalence via transparent domain walls (Def.~\ref{def:p.f.-equivalence}).
This naturally leads to the concept of topological channel connectivity and generalizes traditional notions of phase equivalence. It successfully distinguishes between decohered Abelian topological orders and trivial product states—something the previous naive definition of two-way connectivity via FDLC failed to capture. Furthermore, as long as \emph{approximate} axioms hold, we proved that topological entropy remains invariant across the system up to exponentially small corrections. This further substantiates that our axioms play the role of a gap in mixed states.

We further investigated the richness of non-Abelian topological mixed states. While many of the classical mixed states are unentangled and separable, we demonstrated that they may nevertheless exhibit nontrivial features such as multi-party classical secret-sharing and nonlocal storage of classical information. These properties are absent in decohered Abelian phases, highlighting a new kind of nontriviality in the mixed-state setting. Our results also suggest that there may exist an interesting hierarchy among decohered topological mixed states, which is different from that proposed in \cite{Tantivasadakarn_2023} based on the preparation complexity using measurements and feedback.

This paper takes a first important step toward the systematic understanding of topological mixed states, which inspires numerous interesting directions. 
\begin{itemize}[leftmargin=13pt]
   \item Our approach naturally motivate the study of quantities such as conditional mutual information on a sequence of coarse-grained partitions, each of which is further coarse-grained from the previous one. This class of regions is featured by the uniform rescaling of \emph{all} regions involved, e.g., in concentric partitions centered at a disk, and its generalization related to axiom {\bf M1}. See Sec.~\ref{subsec:coarse-graining} for details.
   We note that, this way of scaling the regions, which naturally reflects the need for the study of RG behavior, is different from the three cases commonly studied in literature, i.e., ``pairwise'', ``local'' and ``global''; see e.g. Ref.~\cite{Chen2025}. We hope future literature will provide stronger bounds on such a partition for Gibbs states, states under decoherence, and ground states.

   \item Incorporating global symmetries: Def.~\ref{def:TCC} can be naturally extended to incorporate global symmetries: given a topological channel, we can additionally impose a symmetry to be conserved on its truncated boundary. This will open a route to study mixed-state versions of symmetry-enriched topological orders as well as symmetry-protected topological states. 
    \item Higher dimensions: While the axioms naturally extend to higher dimensions and we addressed some examples, a comprehensive classification of 3D and beyond remains an open challenge.
    \item RG flow and numerical study: More detailed numerical studies of RG behavior under the axioms could uncover critical points and universality classes in mixed-state settings. While path-integral formalism suggests that this criticality would be that of the same-dimensional classical models, a more systematic understanding of critical behavior is lacking, such as how finite-size scaling works. 
    
    \item Information robustness: 
    Inspired by recent experimental realizations of the $D_4$ topological order \cite{shortestrouteD4_2023,D4trappedion_2024}, which remain coherent under maximal $X$-decoherence \cite{sala2024D4decoherence, sala2024stability}, we suggest that tools in our framework, e.g. information convex set and memory capacity, may offer a more direct way to understand such robustness. For instance, our method may show that the convex set remains a Bloch sphere at maximal decoherence.

    \item Domain walls and mixed-pure interfaces: The study of domain walls between mixed and pure topological phases may yield insights into emergent quasiparticles and braiding, especially when viewed through the lens of information convex sets. This may explain what the right mixed state analog of anyon condensation \cite{Kitaev-Kong2012,Kong2014,Kong2024higher} is, by deriving the fundamental constraints on excitations tunneling through the domain walls.

    \item Chiral mixed state phases: Can there exist mixed-state analogs of chiral topological order? If so, they might necessarily violate the axioms by an amount that decays with system size, akin to how chiral fixed points cannot be exactly realized on finite-dimensional Hilbert spaces~\cite{strict-J-2024}. However, the modular commutator stops being a topological invariant in mixed state setups; see Appendix~\ref{app:non-topo-invariant} (and also~\cite{Kim2022c-minus,chen2023symmetryenforced}).  One may need another chirality measure~\cite{Vardhan2025}. Furthermore, should domains between chiral and nonchiral mixed state phases break more assumptions than {\bf M1}? If so, perhaps one will be able to identify new fixed point conditions~\cite{Lin2023}. 
    
    \item Generalized symmetries: Our approach offers a potential route to derive emergent generalized symmetries, such as higher-forms or non-invertible ones, directly from the axioms rather than postulating them externally. 

    \item Gapless analogs of mixed-state phases: Our framework provides an information-theoretic approach to classifying gapped mixed-state phases of matter. An important open direction is to extend this framework beyond the gapped setting, toward a theory of mixed-state criticality and the recently discovered \emph{information-critical phases}~\cite{vijay2026informationcriticalphasesdecoherence}.

    \item Classical preparability and braiding: In general, what obstructions exist for preparing certain mixed states using finite-depth local channels? We expect the answer lies in a clearer understanding of the emergent particles and their braiding. One possible route is to consider the information convex set for embedded subregions and locally embedded (i.e., immersed) regions. In parallel to the pure state entanglement bootstrap, we expect that insight into braiding can be gained by deforming the region back to its original position, passing through a nontrivial path, and observing the nontrivial isomorphism (Appendix~\ref{app:ICS-Omega}) of the information convex set. In contrast to the braiding non-degeneracy derived from area law axioms, we expect the volume law axioms to allow braiding degeneracy as the outcome. 

    \item With the crucial assumption of spatial uniformity, we also propose a sharpened time-based definition of mixed-state topological phases; see Conjecture~\ref{conj:t}. Rather than invoking a single diverging Markov length, this refinement requires the finiteness of three distinct Markov (or correlation) lengths, thereby revealing a richer structure underlying mixed-state phase transitions. This perspective suggests a richer landscape of critical behavior, distinguished by which subset of the three length scales diverges. Identifying such transitions and uncovering their physical origin would sharpen our understanding of the mechanisms that separate mixed-state phases.
    \end{itemize}

\noindent \emph{Note added}. Ref.~\cite{mixedstatelocalrev2025}, which appeared after our preprint, independently studies the information convex set of mixed state phases.

\section*{Data Availability}
Processed data and code used to reproduce the figures are openly available in Ref.~\cite{mixed_topo_data}.

\section*{Acknowledgments}
J. Y. L dedicates this paper to his daughter Seol Lee, catalyst and companion to this work from conception to birth. 
We thank Soonwon Choi, Tyler Ellison, Haeum Kim, Isaac Kim, Hyukjoon Kwon, Michael Levin, Yizhou Ma, and Shengqi Sang for discussions and feedback. B. S. thanks Tarun Grover, Meng Cheng, Timothy Hsieh, and Chong Wang for inspiring discussions about mixed states. 
The computational work is conducted through the Illinois Campus Cluster Program at the University of Illinois, Urbana-Champaign. 
T.-H. Y. is supported by the Taiwan-UIUC Scholarship program and Elite Dream Project Grant of Veterans and Dependents Foundation, Republic of China. 
B. S. acknowledges support from NSF under award number National Science Foundation under Award No. PHY-2337931 at UC Davis, and support in part by Perimeter Institute for Theoretical Physics. 
B. S. and J. Y. L. are supported by the IQUIST fellowship and faculty startup grant at the University of Illinois, Urbana-Champaign.
J. Y. L. acknowledges support from KIAS through the Quantum Universe Center scholar program.

\appendix

\section{More on mixed-state bootstrap}

\subsection{Global version of the axioms}\label{app:axioms-global}

In the main text, we stated {\bf M0} and {\bf M1} in local patches, making {\bf P0} the only axiom that needs the global state. In fact, {\bf P0} implies a global version of the other two axioms. By the global version of the axioms {\bf M0} and {\bf M1} we mean the following:
\begin{itemize}[leftmargin=13pt]
    \item {\bf M0*}: $I(AA':C)=0$ for the partition in Fig.~\ref{fig:axiom-alternative}(a), where $AA'$ is the complement of $BC$ in the global system $M$.
    \item {\bf M1*}: $I(AA':C|B)=0$ for the partition in Fig.~\ref{fig:axiom-alternative}(b), where $AA'$ is the complement of $BCD$ in the global system $M$.
\end{itemize}

Here we explain the details.

\begin{figure}[h]
    \centering
    \includegraphics[width=0.85\linewidth]{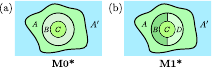}
    \caption{Alternative axioms of {\bf M0*} and {\bf M1*}. Here, $ABC$ and $ABCD$ are the local patches that appear in the definition of {\bf M0} and {\bf M1}. $AA'$ is the complement of the local region $BC$ or $BCD$ on the entire system.}
    \label{fig:axiom-alternative}
\end{figure}

\begin{Proposition}
     A reference state $\sigma_M$ satisfies {\bf P0}, {\bf M0} and {\bf M1} if and only if it satisfies {\bf P0}, {\bf M0*} and {\bf M1*}.
\end{Proposition}

\begin{proof} $(\Leftarrow)$
    This follows immediately from the strong subadditivity, which says $I(AA':C) \ge I(A\,{:}\,C)$ and $ I(AA':C|B) \ge I(A\,{:}\,C|B).$

     $(\Rightarrow)$ {\bf P0} implies the existence of a local channel on $BC$ to recover a missing disk $C$. (See Fig.~\ref{fig:axioms}(a).) It implies the vanishing of conditional mutual information $I(C:E|B)=0$. Noticing that the condition is symmetric under the exchange of $C$ and $E$, we immediately see that the state on region $E$ can be recovered by a channel on $BE$. Applying this to the configurations in Fig.~\ref{fig:axiom-alternative}, namely  
     for the partition of the system into $ABC A'$ where $A'= M\setminus ABC(D)$, we find the following.   {\bf P0} implies that, $\exists$ a quantum channel $\calE_{A\to AA'}$ such that
     \begin{equation}
         \calE_{A\to AA'} (\sigma_{ABC}) = \sigma_{AA'BC},
     \end{equation}
     for both Fig.~\ref{fig:axiom-alternative}(a) and (b).
     This implies that 
     \begin{itemize}[leftmargin=13pt]
         \item for the case in Fig.~\ref{fig:axiom-alternative}(a), $I(A\,{:}\,C)_\sigma \ge I(AA':C)_\sigma.$ This is because mutual information cannot increase when applying a quantum channel to one party. Then, {\bf M0} implies that $I(AA':C)_\sigma=0$. Thus, {\bf M0*} holds.
         \item for the case in Fig.~\ref{fig:axiom-alternative}(b), 
         $$ I(A\,{:}\,C|B)_\sigma \ge I(AA':C|B)_\sigma, $$  
         because quantum channels on an unconditioned party cannot increase the conditional mutual information. Then, {\bf M1} implies that $I(AA':C|B)_\sigma=0$. Thus, {\bf M1*} holds.
         \end{itemize}
        This completes the proof.
         \end{proof}
   
\subsection{Mixed versus Pure}\label{app:local-pure}

In this appendix, we talk about two things relating the ``mixed-state'' to the ``pure-state'' entanglement bootstrap. 

First, we called the entanglement bootstrap for gapped system as ``pure state'' bootstrap. It is not because that the reference state need to be pure, but for the reason that the state satisfies a local purity condition {\bf A0}
\begin{equation}
  \text{axiom }{\bf A0:} \quad   (S_{BC}+ S_{C} - S_B)_\sigma =0, \quad \text{for Fig.~\ref{fig:axiom-alternative}(a)}.
\end{equation}
This condition implies that the state is
``locally pure'' in the sense that $\sigma_{BC} = \rho_{B_L} \otimes | \varphi\rangle_{B_RC}\langle \varphi | $, for some decomposition of the Hilbert space $\calH_{B} = (\calH_{B_L} \otimes \calH_{B_R}) \oplus \cdots$.
The state $|\varphi\rangle$ purifies the reference state density matrix in $C$.

Second, as long as the local purity condition {\bf A0} holds, the axiom {\bf A1} and {\bf M1} can be used interchangeably: 

\begin{lemma}
Consider a state $\sigma$ satisfying {\bf A0} on regions centered at $CD$ and $D$ of Fig.~\ref{fig:axiom-alternative}(b), and suppose 
\begin{equation}\label{eq:M1-app}
    I(A:C|B)_\sigma =0.
\end{equation}
Then $(S_{BC} + S_{CD} - S_{B} - S_D)_\sigma =0$; namely, {\bf A1} holds on configurations centered at $C$.   
\end{lemma}
\begin{proof}
The condition {\bf A0} centered at $CD$ and $D$ give 
\begin{eqnarray}
    S_{ABCD} + S_{CD} &=& S_{AB}, \quad \text{and} \\
    S_{ABCD} + S_{D} &=& S_{ABC}.
\end{eqnarray}
They imply $S_{AB} - S_{ABC} = S_{CD} - S_D$ on $\sigma$. Plugging this into Eq.~\eqref{eq:M1-app}, we get $(S_{BC} + S_{CD} - S_{B} - S_D)_\sigma =0$.
\end{proof}
    
\begin{corollary}
    The axiom set $\{ {\bf A0}, {\bf M1} \}$ defines the same entanglement bootstrap setup as $\{ {\bf A0}, {\bf A1} \}$.
\end{corollary}
Thus, when the local purity condition holds, the axiom {\bf A1} and {\bf M1} can be used interchangeably.

\subsection{Axioms on a discrete lattice}\label{app:discrete-axioms}

On a discrete lattice, such as a square lattice or a honeycomb, the set of partitions for the axioms can be stated precisely. The partition centered at a (coarse-grained) site is a finite set. We explicitly draw this set of partitions for the square lattice. The single configuration in Fig.~\ref{fig:toric-code-cases}(a) is the region to check for {\bf P0} and {\bf M0}. The five configurations in Fig.~\ref{fig:square-M1-complete} ones obtained from them by rotations, reflections and $B\leftrightarrow D$ is the complete list to check for axiom {\bf M1}.

\begin{figure}[h]
    \centering
    \includegraphics[width=0.87\linewidth]{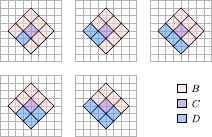}
    \caption{The five partitions together with ones obtained from them by rotations, reflections and swapping $B$ and $D$ is the complete list to check for {\bf M1}.}
    \label{fig:square-M1-complete}
\end{figure}

\subsection{Information convex set is convex}\label{app:convex}

We show that $\Sigma(M|\sigma)$ defined in Def.~\ref{def:ICS-mixed} must be a compact convex set when $M$ is a closed manifold. Here $\sigma$ is a fixed-point state on $M$, which satisfies all the axioms {\bf P0}, {\bf M0}, {\bf M1}.
We note that any state $\rho_M \in \Sigma(M|\sigma)$ has an identical reduced density matrix as $\sigma$ on any ball.

\begin{lemma}\label{lemma:Sigma(Disk)}
For any $\rho_M \in \Sigma(M|\sigma)$ and any disk-like region
    $\mathfrak{D} \subset M$, 
    \begin{equation}
        \rho_{\mathfrak{D}} = \sigma_{\mathfrak{D}}.
    \end{equation}
\end{lemma}

This follows from {\bf M1} and the definition of $\Sigma(M|\sigma)$. The proof is identical to that in Lemma~3.1 and Prop.~3.5 of \cite{Shi2019fusion}, and thus we only recall the essence.   
The axiom {\bf M1} and the monotonicity of conditional mutual information under partial trace imply
\begin{equation}
    I(A\,{:}\,C|B)_\rho =0 \quad \text{for}\quad  \vcenter{\hbox{\includegraphics[width=0.08\textwidth]{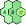}}} 
\end{equation}
where the $ABC$, $AB$, and $BC$ shown in the figure are disk regions in the bulk and $BC$ is contained in a bounded radius disk. Thus, $\rho_{ABC}$ is a quantum Markov state, and it is uniquely determined by $\rho_{AB}$ and $\rho_{BC}$. We use the same idea to divide the state into smaller pieces, each of which is contained in some bounded radius disk, as appeared in Def.~\ref{def:ICS-mixed}. As both $\rho_\mathfrak{D}$ and $\sigma_\mathfrak{D}$ are determined (in the same way) by the same density matrices on the marginals, they must be identical.

\begin{Proposition}[convex set]
   The $\Sigma(M|\sigma)$ in \eqnref{eq:ICS} is a compact convex set.     
\end{Proposition}

\begin{proof}
    First, $M$ is a finite lattice, and thus the state space (i.e., the set of all density matrices) of a finite-dimensional Hilbert space is a finite-dimensional compact
convex set $S$. The set $\Sigma(M|\sigma)$ is compact because it is obtained by adding a set of constraints that are linear equalities to the set $S$. Let us analyze the conditions one by one:
\begin{itemize}[leftmargin=13pt]
    \item The first constraint is the local indistinguishability. For any bounded radius ball $b\subset M$, we require $\rho_{b}\,{=}\,\sigma_b$. These are clearly a set of linear constraints.
    \item This second constraint, local recoverability, is the requirement that {\bf P0} is satisfied for an $\rho_M \in \Sigma(M|\sigma)$. Noticing that $\rho_M$ reduced to an arbitrary disk must be identical to the reference state (Lemma~\ref{lemma:Sigma(Disk)}), we see that the constraint can be rewritten as
    \begin{equation}
        \calE^{\sigma}_{B\to BC} (\rho_{AB}) = \rho_{ABC} , \quad \forall \rho_{ABC} \in \Sigma(M|\sigma).
    \end{equation}
    where $\calE^{\sigma}_{B\to BC}(\cdot)$ is the Petz map, which is fixed by the reference state $\sigma_{BC}$ on local disk $BC$~\cite{PETZ_2003} and is independent of the state $\rho$.
    These are another set of linear constraints.
\end{itemize}
Thus, $ \Sigma(M|\sigma) $ is a compact convex set.
\end{proof}

\subsection{Role of {\bf P0} in Information convex set}\label{app:info-robustness}

We proposed that the information convex set $\Sigma(M|\sigma)$ captures the notion of physically robust information storage. But why is this the appropriate choice?
In particular, one might ask: why not relax the local recoverability ({\bf P0}) in Def.~\ref{def:ICS-mixed}? If we were to do so, we would arrive at an alternative set, but as we will show, this modified set fails to retain key features essential for topological robustness.

Consider the alternative set
\begin{equation}
    \tilde{\Sigma}(M|\sigma) = \{ \rho_M|\rho_b = \sigma_b,\,\forall b \subset M\},
\end{equation}
where $b$ is any bounded radius disk and $M$ is a closed manifold. $\tilde{\Sigma}$ is the set of states on $M$, which are locally indistinguishable with the reference state.
It is easy to show $\tilde{\Sigma}(M|\sigma)$ is a compact convex set that contains ${\Sigma}(M|\sigma)$. When the pure state bootstrap axioms ({\bf A0} and {\bf A1}) apply, these two sets are identical. This is because {\bf A0} on the local ball implies {\bf P0}. However, in mixed state setups, 
 $\tilde{\Sigma}(M|\sigma)$ can be strictly larger, as we explain in examples below.

\begin{exmp} 
Let $\sigma = \prod_{i\in \calS} 1_i$. Here $\calS$ is the set of all qubits of the system.
In this case, $\Sigma(M|\sigma)= \{ \sigma_M\}$. On the other hand, 
\begin{equation}\label{eq:lost1}
    \frac{1 \pm \prod_{i\in \calS} X_i}{2} \in \tilde{\Sigma}(M|\sigma).
\end{equation}
    This is because, as long as we trace out any single site, the state will be the maximally mixed state. Thus $\tilde{\Sigma}$ is strictly larger than $\Sigma$.
\end{exmp}

\begin{exmp} Let $M= \mathbb{T}^2$ and let $\sigma =  \prod_v A^{\bm{1}}_v.$ be the dephased toric code, as in Example~\ref{exmp:toric-code-fixed-point}. In this case, $\Sigma(M|\sigma)$ is a simplex with 4 extreme points (Prop.~\ref{prop:convexset}). These extreme points are
\begin{equation}
    \rho^e_{ab} \propto \prod_v A_v^{\bm{1}} (1 \pm \bar{X}_1)  (1 \pm \bar{X}_2), 
\end{equation}
where $\bar{X}_1$ and $\bar{X}_2$ are non-contractible $X$ string on the dual lattice along the noncontractible $x,y$ cycles of the torus respectively.
On the other hand, we have more states in $\tilde{\Sigma}(M|\sigma)$, e.g.
\begin{equation}\label{eq:state-Z}
   (1 + \bar{Z}_j)\prod_v A_v^{\bm{1}}   (1 + \bar{Z}_j) \in \tilde{\Sigma}(M|\sigma).
\end{equation}
where $\bar{Z}_j$, with $j=1,2$ are non-contractible $Z$ strings along the non-contractible $x,y$ cycles of the torus respectively. 
The reason is that the noncontractible loop $\bar{Z}_j$ disappears after the partial trace to a local region, and thus the state in Eq.~\eqref{eq:state-Z} is locally indistinguishable from $\sigma$. The distinction between $\bar{X}_j$ and $\bar{Z}_j$ is that the former can be smoothly deformed when acting on $\rho\in \Sigma(M|\sigma)$, whereas the latter cannot. 
\end{exmp}

We have seen from examples above that $\tilde{\Sigma}(M|\sigma)$ can be strictly larger than ${\Sigma}(M|\sigma)$. We are ready to explain why such a convex set does not represent the set of ``topologically protected'' information. Suppose that the environment is capable of removing any local ball and resetting the state to a local product state. Then the information represented by states in $\tilde{\Sigma}(M|\sigma)$ cannot be recovered. Explicitly, 
\begin{itemize}[leftmargin=13pt]
    \item For the state in Eq.~\eqref{eq:lost1}, tracing out any qubit will result in the loss of the information.
    \item For the state in Eq.~\eqref{eq:state-Z}, tracing out any qubit touching the (nondeformable) $\bar{Z}_j$ string will result in the loss of the information.
\end{itemize}
On the other hand, due to {\bf P0}, the information in ${\Sigma}(M|\sigma)$ is always preserved by such a process. In fact, the original states can be recovered by the local channel surrounding the hole. This is the main physical reason we prefer Def.~\ref{def:ICS-mixed} as the definition of information convex set, which captures the physical information storage that is robust.

\subsection{Bound on relative entropy change}\label{app:relative-S}

Consider a system partition into $BCE$ where $BC$ is a bounded region and $E$ may be large, allowing $\dim \calH_{E}\to \infty$. In general, it is possible to induce a global change in the density matrix by applying a local quantum channel $\calE_C$, even if such a global change is inaccessible (negligible) if we trace out the complement of the local neighborhood of $C$. In other words,
\begin{equation}
 D(\rho_{EBC} \Vert \calE_C[\rho_{EBC}])  \gg D(\rho_{BC} \Vert \calE_C[\rho_{BC}]),
\end{equation}
may happen. This phenomenon happens in the context of SWSSB phase transition.

We explain why {\bf P0} prevents this from happening. First, the conditional mutual information is monotonic under the action of quantum channels on an unconditioned region. Thus, \begin{equation}
    0 = I(E:C|B)_{\rho} \geq I(E:C|B)_{\calE_C[\rho]} = 0.
\end{equation}
Importantly, $\rho_{BCE}$ and $\calE_C[\rho_{BCE}]$ can be recovered from  $\rho_{BC}$ and $\calE_C[\rho_{BC}]$ respectively by the \emph{same} quantum channel. Such a quantum channel can be chosen to be the Petz map~\cite{PETZ_2003} $\calE^{\text{Petz}}_{B\to BE}(\cdot):= \rho_{BE}^{\frac{1}{2}}\, \rho_B^{-\frac{1}{2}} (\cdot)\rho_{B}^{-\frac{1}{2}} \rho_{BE}^{\frac{1}{2}}$.
By the monotonicity of relative entropy we must have $ D(\rho_{BC} \Vert \calE_C[\rho_{BC}]) \le D(\rho_{EBC} \Vert \calE_C[\rho_{EBC}])$, where the quantum channel is the partical trace $\Tr_E$. The existence of the recovery map $\calE^{\text{Petz}}_{B\to BE}$, on the other hand, implies $D(\rho_{EBC} \Vert \calE_C[\rho_{EBC}]) \le D(\rho_{BC} \Vert \calE_C[\rho_{BC}])$. Thus,
\begin{align}
    D(\rho_{EBC} \Vert \calE_C[\rho_{EBC}]) = D(\rho_{BC} \Vert \calE_C[\rho_{BC}]).
\end{align}
Therefore, the local channel acting on $C$ cannot induce a global change in the density matrix.

\subsection{Sphere has zero memory capacity}\label{app:S2}

We show that the information convex set on the sphere has a unique element. Thus, no robust information can be stored on a sphere. The memory capacity $\calQ(S^2|\sigma)=0$.

\begin{Proposition}
    Let $S^2$ be a sphere, and $\sigma$ be a reference state, then 
    \begin{equation}
        \Sigma(S^2|\sigma) = \{ \sigma_{S^2}\}.
    \end{equation}
\end{Proposition}

\begin{proof}
The proof is an analog of Prop.~3.7 of pure state entanglement bootstrap~\cite{Shi2019fusion}, where in the latter case one can also show the reference state on the sphere is pure. Let the sphere $S^2$ be partitioned into $A,B$, and $C$ such that $B$ is the thickening of the equator. Thus, $AB, BC$ are two disks. Suppose $\rho_{ABC},\sigma_{ABC} \in \Sigma(S^2)$. As both $AB$ and $BC$ are disks, Lemma~\ref{lemma:Sigma(Disk)} implies that $\rho_{AB} = \sigma_{AB}$ and  $\rho_{BC} = \sigma_{BC}$.  Then, by {\bf P0}, we have $I(A\,{:}\,C|B)=0 $ for both $\rho$ and $\sigma$. Thus, the states are determined by the marginals on $AB$ and $BC$. Thus, the two global states $\rho_{ABC} = \sigma_{ABC}$. This completes the proof.
\end{proof}

\section{Information convex set beyond closed manifolds}\label{app:ICS-Omega}

In Appendix.~\ref{app:convex}, \ref{app:info-robustness}, and \ref{app:S2}, we discussed the properties of information convex sets on closed manifolds. Similar definition applies to regions with boundaries.
When we discuss `boundaries' in the bootstrap context, there are two types: ($i$) physical boundary, beyond which the system is in a product state and does not mediate correlation nor entanglement, and ($ii$) correlation boundary (generalizing the notion of entanglement boundary discussed in \cite{Shi2019fusion,Braiding2023}), which is not. The boundary we shall be interested in are the correlation boundary.

In general, we allow $\Omega$ to be a region embedded in the background physical system equipped with a fixed-point state $\sigma$ satisfying {\bf P0}, {\bf M0} and {\bf M1}. Let us make a few remarks about the definition:
\begin{enumerate}[leftmargin=13pt]
    \item We always consider a region that contains a finite number of coarse-grained sites. Each site has a finite Hilbert space dimension. Again, one can show that the set $\Sigma(\Omega)$ is a compact convex set.

    \item Information convex sets associated with different regions are related to each other nontrivially. Notably  
 \begin{equation}
     \forall \rho_{AB} \in \Sigma(AB), \quad \Tr_B \rho_{AB} \in \Sigma(A).
 \end{equation}
  \item Let $D$ be a disk, then $\Sigma(D) = \{ \sigma_D\}$. This generalizes Lemma~\ref{lemma:Sigma(Disk)}.
  
  \item  The purpose of having the full set of conditions near the entanglement boundary is to ensure that qubits located in far-separated parts of the entanglement boundary do not entangle with each other. This definition is a close analog to Def.~C.1 of~\cite{Shi2019fusion}. To see the similarity more clearly, we note that, the properties inherited from the fixed-point state and the argument in Appendix~\ref{app:axioms-global} imply the properties listed in Fig.~\ref{fig:ICS-properties}. This further guarantees that the states on region $\Omega$ can be smoothly extended, leading to the important \emph{isomorphism} property, Eq.~\eqref{eq:iso-thm}, as we discuss below.
    
    \item More broadly speaking, the definition can be further generalized to immersed (i.e., locally embedded) regions. The reason is that the first condition in the Def.~\ref{def:ICS-mixed} puts consistency on local balls only. 
    While we omit the discussion of immersed regions for simplicity, we point out that the information convex sets for such regions is known to give nontrivial constraints for braiding statistics~\cite{Braiding2023,Shi2024-figure8} in pure state entanglement bootstrap.  
\end{enumerate}

 \begin{figure}[h]
     \centering    \includegraphics[width=0.9\linewidth]{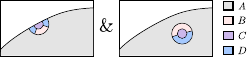}
     \caption{Some properties of $\rho \in \Sigma(\Omega)$. $I(A\,{:}\,C)_\rho =0$ and $I(A\,{:}\,C|B)_\rho =0$ for the configurations shown. Here, the region $B$ is disk-like. In the left figure, $B$ may fill in the blank region on the left, right or both.}
     \label{fig:ICS-properties}
 \end{figure}

Let us explain the notion of ``smooth deformation'' of an embedded region. 
We say embedded regions  $\Omega $ and $\Omega'$ can be smoothly deformed into each other (denoted as $\Omega \sim \Omega'$) if they are related by a finite sequence of elementary steps. Here, an elementary step is the process of adding or deleting a small ball near the entanglement boundary in a way that preserves the topology of the region, as shown in Fig.~\ref{fig:elementary-step}. 

\begin{figure}[h]
    \centering\includegraphics[width=0.60\linewidth]{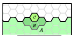}
    \caption{What we mean by an elementary step. $\Omega=AB$ and $\Omega'=ABC$ are related by an elementary step. $BC$ is within a bounded radius ball, and only part of $C$ is shown.}
    \label{fig:elementary-step}
\end{figure}

The important statement (known as the isomorphism theorem \cite{Shi2019fusion}) is the statement that
   \begin{equation}\label{eq:iso-thm}
        \forall\, \Omega \sim \Omega', \quad \Sigma(\Omega) \cong \Sigma(\Omega').
    \end{equation}
    Here $\Omega \sim \Omega'$ (reads $\Omega$ is smoothly deformable to $\Omega'$) if two subsystems $\Omega$ and $\Omega'$ can be related by a finite sequence of elementary steps (Fig.~\ref{fig:elementary-step}).
    The meaning of isomorphism ``$\cong$'' will be explained in a few paragraphs, and for now, we just mention that an isomorphism of the information convex sets is obtained by a sequence of quantum channels that is determined by the elementary steps in the deformation process $\Omega \sim \Omega'$. It is thus sufficient to explain what an elementary step does.~\footnote{The proof of the following statement is nontrivial, but it is essentially identical to the argument in Appendix~C of~\cite{Shi2019fusion}. We thus omit the proof here.}
    Consider the following elementary steps for $\Omega=AB$ and $\Omega'=ABC$ as in Fig.~\ref{fig:elementary-step}. Note that $I(A\,{:}\,C|B)_\rho =0$ for any $\rho \in \Sigma(ABC)$, following Def.~\ref{def:ICS-mixed}. The nontrivial statements are the operations on the convex set of quantum states
\begin{itemize}[leftmargin=13pt]
    \item Elementary step of extension:
    \begin{equation}
        \calE_{B\to BC}^\sigma(\rho_{AB}) = \rho_{ABC} \in \Sigma(ABC),  
    \end{equation}
    for any $ \rho_{AB} \in \Sigma(AB)$. Here $\calE_{B\to BC}^\sigma$ is the Petz map defined from the reference state.
    \item Elementary step of restriction:
        \begin{equation}
        \Tr_C(\rho_{ABC}) = \rho_{AB} \in \Sigma(AB),  
    \end{equation} 
    for any $ \rho_{ABC} \in \Sigma(ABC)$.
\end{itemize}
Importantly, the composition of $\calE_{B\to BC}\circ \Tr_C$ is an identity on the set $\Sigma(ABC)$ and $\Tr_C \circ \,\calE_{B\to BC} $ is an identity on the set $\Sigma(AB)$. 
Thus, the two sets must be isomorphic. This further implies the three conditions below, which underpin our notion of isomorphism ``$\cong$'':
\begin{enumerate}[leftmargin=13pt]
    \item (geometry preserving) The two sets must be identical as geometrical convex sets in the sense that the convex combination commutes with the mapping:
    \begin{equation}
        \Phi( p \rho + (1-p) \lambda) = p \,\Phi( \rho ) + (1-p) \Phi(\lambda)    
    \end{equation} 
    \item (distance preserving) The distance measure between elements is preserved
    \begin{equation}
        D(\rho,\lambda) = D(\Phi(\rho),\Phi(\lambda)),
    \end{equation}
    where $D(\cdot, \cdot)$ is any distance measure that is monotonic under quantum channels, and $\Phi$ is the isomorphism map between the sets.
    \item (entropy difference preserving)
    \begin{equation}
      S(\rho) - S(\lambda) =   S(\Phi(\rho)) - S(\Phi(\lambda)). 
    \end{equation}
\end{enumerate}
These properties imply further conclusions, including the conservation of memory capacity,
\begin{equation}
    \calQ(\Omega|\sigma) =\calQ(\Omega'|\sigma), \quad \forall \Omega \sim \Omega'.
\end{equation} 

Importantly, the isomorphism guarantees a sense of uniformity of many-body quantum states on the 2D manifold. This is because it is possible to deform a region $\Omega$ to another $\Omega'$ that is separated far apart from it.  That the structure captured by the information convex sets of $\Omega$ and $\Omega'$ (e.g., the topological charges) is isomorphic implies that the system has an emergent homogeneity without invoking any symmetry (e.g., translation).

\subsection{Annulus and memory capacity}\label{app:annulus-and}

Let $X$ be an annulus that is embedded in a larger disk on which we have the reference state. 
Under {\bf P0}, {\bf M0} and {\bf M1}, we prove the following. 

\begin{lemma}\label{lemma:vac}
The reference state reduced density matrix $\sigma_X:=\Tr_{\bar{X}} \sigma$ is an element of the information convex set of the annulus, $\Sigma(X)$. Moreover,
\begin{enumerate}[leftmargin=14pt]
    \item $\sigma_X$ is a state with (absolute) minimum von Neumann entropy in the set $\Sigma(X)$.
    \item $\sigma_X$ is an extreme point of $\Sigma(X)$.  
    
\end{enumerate}
\end{lemma}

\begin{proof}
   The fact that $\sigma_X \in \Sigma(X)$ follows immediately from the defintion of information convex set (Def.~\ref{def:ICS-mixed}). 
   Now consider an arbitrary element of $\Sigma(X)$, denoted as $\alpha_X$. For the partition in Fig.~\ref{fig:ABC-hole},
   \begin{equation}\label{eq:alpha-has-larger-S}
   \begin{aligned}
         I(A:C|B)|^\alpha_\sigma 
         & = S(\alpha_{AB}) + S(\alpha_{BC}) - S(\alpha_{B}) - S(\alpha_{ABC})\\
         & \,\,\,\, -  S(\sigma_{AB}) - S(\sigma_{BC}) + S(\sigma_{B}) + S(\sigma_{ABC}) \\
         &= S(\alpha_X) - S(\sigma_X).
   \end{aligned}
   \end{equation}
   This follows because of three facts. First, $S_{AB}|_\sigma^\alpha$ and $S_{BC}|_\sigma^\alpha$ are the same, and they are both the same with $S_{X}|_\sigma^\alpha$. This is a consequence of the isomorphism theorem. Second, $S_B|_\sigma^\alpha =0$ because all states in the information convex set must be identical on any disk-like subregion, e.g. $B$. Third, $S_{ABC}|_\sigma^\alpha = S_{X}|_\sigma^\alpha$. This is because creating a hole on the annulus $X$ gives $ABC$, and punching such a hole cannot change the entropy difference between $\alpha$ and $\sigma$. 

   From Eq.~\eqref{eq:alpha-has-larger-S} and the fact that conditional mutual information is non-negative, we see that $S(\alpha_X) \ge S(\sigma_X)$ for all states in $\Sigma(X)$. In other words, $\sigma_X$ is a state with the absolute minimum von Neumann entropy in the set $\Sigma(X)$. 

   Lastly, we show $\sigma_X$ is an extreme point. Suppose $\sigma_X$ is not an extreme point, then there exists $\alpha_X,\beta_X \in \Sigma(X)$ such that $\sigma_X = p \, \alpha_X + (1-p) \beta_X $ with $p \in (0,1)$ and $\alpha_X \ne \beta_X$. But then
   \begin{equation}
       \begin{aligned}
          S(\sigma_X) & > p S(\alpha_X) + (1-p) S(\beta_X) \\
          & \ge \min \{ S(\alpha_X), S(\beta_X) \}.
       \end{aligned}
   \end{equation}
   This leads to a contradiction. This completes the proof.
\end{proof}

\begin{figure}
    \centering
    \includegraphics[width=0.38\linewidth]{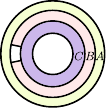}
    \caption{Delete a hole (white rectangle) from a contractable annulus $X$ to obtain $ABC \subset X$.} 
    \label{fig:ABC-hole}
\end{figure}
\noindent
With this we have:
 \begin{Proposition}\label{prop:annulus-TE}
For a contractable annulus $X$, the memory capacity of the information convex set equals twice the topological entropy,
\begin{equation}
    \calQ(X|\sigma)= 2 \gamma_{{\rm LW}}.
\end{equation}
\end{Proposition}
\begin{proof}
    We apply the same proof strategy as \cite{Shi2019fusion,Braiding2023}. The max-entropy state $\rho^\star_X$ of the information convex set $\Sigma(X|\sigma)$ is a Markov state having $I(A:C|B)_{\rho^\star} =0$ for the Levin-Wen partition of the annulus.
    On the other hand, the vacuum state $\sigma_X$ is a state in $\Sigma(X|\sigma)$ that reaches the absolute minimum of entropy by Lemma~\ref{lemma:vac}.
Thus,
\begin{equation}
\begin{aligned}
    \calQ(X|\sigma) &= \max_{\rho, \lambda} \left( S(\rho_X) - S(\lambda_X) \right)\\
    &= S(\rho^\star_X) - S(\sigma_X)  \\
                    &= I(A:C|B)_\sigma - I(A:C|B)_{\rho^\star}  \\
                    &= 2 \gamma_{\rm LW} - 0 \\
                    & = 2 \gamma_{\rm LW}.
\end{aligned}
\end{equation}
The first equality is the definition of memory capacity.
The third equality uses the fact that any state in the information convex set of $X=ABC$ reduces to the same state on $AB$, $BC$ and $B$. 
This completes the proof.
\end{proof}
In this perspective, the topological entropy $\gamma_{\rm LW}$ is a derived quantity that captures part of the richer structure. 

\section{Robustness of isomorphism}\label{app:robust-iso}

To analyze the robustness of the isomorphism away from the fixed points of topological mixed states, we allow relaxations in which the axioms may take small decaying values rather than being required to vanish exactly. In addition, we also relax the strict local equivalence condition of the information convex set. We therefore define a notion of approximate information convex set as a refinement to Def.~\ref{def:ICS-mixed}:

\begin{definition}[Approximate information convex set]\label{def:ICS-mixed-robust}
  For a subsystem $\Omega$ and a relaxation $\mathfrak{q} \ge 0$, we define the approximate information convex set as
  \begin{equation}
      \Sigma_\mathfrak{q}(\Omega):= \{\rho_\Omega| (1),(2)\}
  \end{equation}
  \begin{enumerate}
      \item [(1)] $||\rho_b - \sigma_b||_1 \le \mathfrak{q}$ for bounded radius balls $b \subset \Omega$.
       
      \item [(2)]  $I(E\,{:}\,C|B)_\rho \le \mathfrak{q}$ for configurations in Fig.~\ref{fig:ICS-def-short}. 
  \end{enumerate}
\end{definition}

\begin{figure}[h]
    \centering    \includegraphics[width=0.95\linewidth]{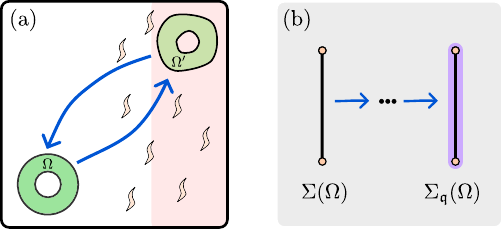}
    \caption{On the robust (approximate) isomorphism of information convex sets. 
    }
    \label{fig:robust-iso}
\end{figure}

With this definition stated, we now analyze whether deformations of the information convex set within the physical system preserve its structure under errors decaying with the length scale. If so, the  topological data extracted from a local region remains invariant across the system, thereby ensuring that the information convex set is stable under small, decaying perturbations.  

Now we discuss a toy problem that the axioms are satisfied exactly on the left part of the system and is noisy on the right side of the system with axioms violated with errors up to \(\delta(r)\le \mathcal{O}(1)e^{-r/\xi}\), as shown the geometry of Fig.~\ref{fig:robust-iso}(a). Suppose we transport the states on region $\Omega$ to the noisy region (to $\Omega'$) and back by quantum channels
\(\Phi_{\mathrm{forward}}\) and \(\Phi_{\mathrm{backward}}\). We ask if the information is approximately preserved. This corresponds to an estimation of the error $\mathfrak{q}$ in
\begin{equation} \Phi_{\mathrm{backward}}\circ\Phi_{\mathrm{forward}}: \quad \Sigma(\Omega)\to\Sigma_{\mathfrak q}(\Omega).
\end{equation}
Using tools of quantum information theory like Fawzi-Renner inequality \cite{FawziRenner_2015}, Alicki-Fannes-Winter inequality \cite{Audenaert_2007,Winter_2016}, and Fuchs-van de Graaf inequality\cite{fuchs1998cryptographicdistinguishabilitymeasuresquantum}, and the idea of merging quantum Markov states~\cite{Kato2015,Kim-unpublished}, we are able to come up with an estimation
\begin{equation}
\label{eq:q-decay}
\mathfrak q(r)\;\le\; 
\textrm{poly}(r)\,\exp({-\frac{r}{\mathfrak{c}\,\xi}}).
\end{equation}
where $\mathfrak{c}\in \mathbb{R}_+$ is an $O(1)$ constant that is depend on the details of the deformation path but is independent of the length scale $r$. 
In particular, the structure of the information convex set is preserved up to exponentially suppressed corrections, so that violations of the axioms only induce a tiny “blurring" of the convex set geometry as shown in Fig.~\ref{fig:robust-iso}(b).
We refer to a future work~\cite{Shi2026} for details. 
Importantly, the degree of $\textrm{poly}(r)$ in equation \eqref{eq:q-decay} is still unclear. If one can bound the degree of the polynomial, then the convex set structure is still preserved under certain milder power-law decay violations of the axioms.

\section{Bound on memory capacity}\label{app:bound-C-TEE}

In this appendix, we prove a general theorem (Thm.~\ref{thm:capacity-general}) which gives an upper bound of memory capacity in terms of the topological entropy $\gamma_{\textrm{LW}}$. This implies the bound on capacity in the main text (Prop.~\ref{prop:torus-capacity}). We start with a useful lemma.

\begin{lemma}\label{lemma:LW}
Consider regions $A,B,C$ and $D$ embedded in $M$ as in Fig.~\ref{fig:M-ABCD}(a). Then,
\begin{equation}\label{eq:LW-global}
     I(A\,{:}\,C|B)_{\rho}= 2 \gamma_{\textrm{LW}}, \quad \forall \rho \in \Sigma(M|\sigma).
\end{equation}
\end{lemma}

\begin{figure}[h]
    \centering
    \includegraphics[width=0.96\linewidth]{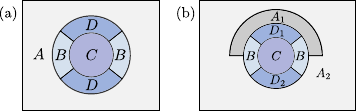}
    \caption{Partition the closed manifold $M$ into a ball $BCD$ (with topology as shown) and the complement $A$.}
    \label{fig:M-ABCD}
\end{figure}

\begin{proof}
Let us relabel $D$ as $D_1 D_2$ and $A$ as $A_1 A_2$ as shown in Fig.~\ref{fig:M-ABCD}(b). First, we note that
\begin{equation}\label{eq:bound<}
\begin{aligned}
   I(A\,{:}\,C|B)_\rho & \ge I(A_1:C|B)_\rho \\
   & =2 \gamma_{\textrm{LW}}.  
   \end{aligned}
\end{equation}
The inequality follows from the strong subadditivity. 
Secondly,  
\begin{equation}\label{eq:bound>}
\begin{aligned}
    I(A\,{:}\,C|B)_\rho - 2 \gamma_{\textrm{LW}} &= I(A_1 A_2 :C |B)_\rho - I(A_1:C |B)_\sigma \\
     &  = I(A_1 A_2 :C |B)_\rho - I(A_1:C |B)_\rho \\
     & = I(A_2: C | A_1 B)_\rho \\
     & \le I(A_2: C D_1 |A_1 B)_\rho  \\
     & = 0.
\end{aligned}
\end{equation}
    The ``$=$" in the first line uses the definition of $\gamma_{\textrm{LW}}$, noticing that $A_1, B, C$ is the Levin-Wen partition of an annulus $A_1 B C$; the annulus is embedded in a disk. The 2nd line uses that $\rho_M \in \Sigma(M)$ is identical to the reference state $\sigma$ on the disk containing $A_1 B C$. The third line is obtained by the chain rule (i.e. expansion and trivial rewriting) of the conditional mutual information. The fourth line is due to the strong subadditivity. The last line is by axiom {\bf M1} with the choice of the local ball $(A_1 B, C D_1, D_2)$.

   By Eq.~\eqref{eq:bound>}, $I(A\,{:}\,C|B)_\rho \le 2 \gamma_{\textrm{LW}}$, while Eq.~\eqref{eq:bound<} says $I(A\,{:}\,C|B)_\rho \ge 2 \gamma_{\textrm{LW}}$. Thus Eq.~\eqref{eq:LW-global} holds.
\end{proof}

\begin{theorem}\label{thm:capacity-general}
    Suppose $\sigma_M$ is a reference state on closed 2D manifold $M$, satisfying all the axioms {\bf P0}, {\bf M0} and {\bf M1}.
    The memory capacity satisfies 
    \begin{equation}
        \calQ(M|\sigma ) \le 2 n(M) \,\gamma_{\textrm{LW}},
    \end{equation}
    where $\gamma_{\textrm{LW}}$ is the topological entropy and $n(M)$ is an integer, such that
    \begin{itemize}[leftmargin=13pt]
        \item If $M$ is orientable, i.e. $M= \# g \mathbb{T}^2$, $n(M) = 2 g$.
        \item If $M$ is non-orientable, $M= \#h \mathbb{RP}^2$, $n(M) = h$.
    \end{itemize}
\end{theorem} 

\begin{proof}
 Let us start with a basic topology observation. Let $M^*$ be the manifold with boundary obtained by removing a ball from $M$.
Then, $M^*$ can be converted to a disk by cutting $n(M)$ handles. See Fig.~\ref{fig:cut-handle}(a) for the meaning of cutting a handle.  

\begin{figure}[h]
    \centering
    \includegraphics[width=0.8\linewidth]{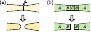}
    \caption{Cut a handle and how it affects the entropy difference. In (b), only part of the region $A$ is shown.}
    \label{fig:cut-handle}
\end{figure}

Now, let us discuss the quantum aspect of removing a ball and cutting a handle. 
   Removing an isolated ball does not change the entropy difference, 
    \begin{equation}
        S(\rho_M)- S(\lambda_M) = S(\rho_{M^*})- S(\lambda_{M^*}), \quad \rho,\lambda \in \Sigma(M).
    \end{equation}
    This follows from condition {\bf P0} on the two states $\rho$ and $\lambda$ and the fact that their local reduced density matrices are identical. 
    
    On the other hand, cutting a handle can lower the entropy difference by an amount. Let $ABC \subset M$ as shown in Fig.~\ref{fig:cut-handle}(b). For $\rho, \lambda \in \Sigma(M)$,
    \begin{equation}
    \begin{aligned}
          S_{ABC}|^\rho_\lambda - S_{AB}|^\rho_\lambda & = I(A\,{:}\,C|B)|^\lambda_\rho\\
          &\le I(A\,{:}\,C|B)_\lambda \\
          & \le 2 \gamma_{\textrm{LW}}.
    \end{aligned}      
    \end{equation}
    The first line is because $\rho$ and $\lambda$ are identical on $BC$. The 2nd line follows because $ I(A\,{:}\,C|B)_\rho$ is non-negative. The last line follows from Lemma~\ref{lemma:LW}. 
    
    This implies that the entropy difference $S_{ABC}|^\rho_\lambda$ and $S_{AB}|^\rho_\lambda$ may be different, but their difference is upper bounded, due to the fact that $AB$ is obtained by $ABC$ by cutting a handle.  
     But, on a disk, we have a unique state, by Lemma~\ref{lemma:Sigma(Disk)}. Thus, the memory capacity of $M$ can be at most $2 n(M) \gamma_{\textrm{LW}}$. This completes the proof.  
\end{proof}

This proof is elementary. The idea is inspired by the equally elementary proof by Kim~\cite{Kim2013storage} as well as the uniqueness of the groundstate of invertible state~\cite{Kitaev-2011}. In the pure state context, Kim's bound is stricter than ours, and the bound saturates for Abelian anyon theory; however, it is not optimal when there are non-Abelian anyons. The more advanced pure state bootstrap technique can give a precise prediction on the memory capacity~\cite{Braiding2023}. The analog of such results in the mixed state setup remains an open problem.

 \section{Quantities that fail to be invariant on topological mixed states}\label{app:non-topo-invariant}
 
We demonstrate why certain topological invariants on the ground state of topological order are no longer invariant under smooth deformation of the regions for topological mixed states. 
We show that {\bf P0}, {\bf M0}, and {\bf M1} cannot guarantee the invariance of 
 \begin{itemize}[leftmargin=15pt]
     \item The Kitaev-Preskill entropy~\cite{KitaevPreskill_2006}  
\begin{equation}
    \gamma_{\rm KP} := S_{AB} + S_{BC} + S_{CA} - S_{A} - S_B - S_C - S_{ABC}
\end{equation}
     \item The modular commutator~\cite{Kim2022c-minus}
     \begin{equation}
         J(A,B,C)= i \Tr(\rho_{ABC} [\ln \rho_{AB}, \ln \rho_{BC}]),
     \end{equation}
 \end{itemize}
 where the partition $A,B,C$ in question is that in Fig.~\ref{fig:coarse-grain}(a). 
We construct the counterexample as follows. 
We start with a fixed-point state $\sigma$ on a certain coarse-grained lattice, where $\sigma$ satisfies all three axioms on this coarse-grained lattice as exemplified in Fig.~\ref{fig:coarse-grain}(b). 
Let us further suppose that $\sigma$ has $\gamma_{\rm KP}$ and $J(A,B,C)$ to be the same value anywhere on the coarse-grained lattice.

Next, we stack a mixed state $\lambda_{abc}$ on the three coarse-grained sites $a,b,c$ of Fig.~\ref{fig:coarse-grain}(b) to obtain
\begin{equation} \label{eq:tensorstack}
    \tilde{\sigma} =\sigma \otimes \lambda_{abc}.
\end{equation}
By choosing $\lambda_{abc}$ appropriately, we can shift $\gamma_{\rm KP}$ and $J(A,B,C)$ for $\tilde{\sigma}$ relative to those for $\sigma$. In particular, 
\begin{align}
    \lambda_{abc}^{\gamma} &= \frac{1}{2} \big( \ket{000}\bra{000}+\ket{111}\bra{111} \big) \nonumber \\
    \lambda_{abc}^J &= \frac{1}{8} \Big[ I+\frac{1}{2}(X_aX_b+Z_bZ_c+X_aY_bZ_c) \Big]
\end{align}
decrease $\gamma_{\rm KP}$ by $\log 2$ and $J$ by $\simeq 0.3$, respectively. 
These shifts occur only for tripartitions $(A,B,C)$ whose triple-intersection point coincides with the center of the $abc$ cell.

\begin{figure}[!t]
     \centering
     \includegraphics[width=0.85\linewidth]{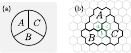}
     \caption{A coarse-grained lattice where $A,B,C$ form a useful partition of the disk. $a\,{\subset}\,A$, $b\,{\subset}\,B$ and $c\,{\subset}\,C$.}
     \label{fig:coarse-grain}
 \end{figure}

Furthermore, for any state $\lambda_{abc}$ supported on the three adjacent coarse-grained cells, the stacked state $\tilde{\sigma} =\sigma \otimes \lambda_{abc}$ satisfies the axioms on the coarse-grained lattice. 
Because $a,b,c$ are adjacent cells, when it comes to the calculation of mutual information or conditional mutual information between \emph{separated} regions, their correlation never appears; furthermore, $\lambda_{abc}$ is in a tensor product with $\sigma$, which already satisfies the axioms.

 Finally, it is easy to see why in the pure state version of the entanglement bootstrap, such counterexamples do not exist. In this context, we are only allowed to stack pure state $|\lambda_{abc}\rangle$ for preserving axioms the pure state axioms; suppose that we tensor with a mixed state at a local region, then the axioms {\bf A0} centered at the region will be violated by twice the entropy of the mixed state. Such a pure state does not give any correction to $\gamma_{\rm KP}$ and $J$. In particular, $J(A,B,C)_{|\psi_{ABC}\rangle} = i \langle [\ln \rho_{AB},\ln \rho_{BC}|\rangle =i \langle [\ln\rho_C, \ln \rho_A]\rangle =0$ for any tripartite pure state $|\psi_{ABC}\rangle$~\cite{Modular-commutator-Gapped}.

\section{Decoherence and entropy conditions}\label{app:violation}

In this appendix, we give details on the calculations of axiom violations. In Appendix~\ref{app:fixed-pts}, we analytically compute the violations of various entropy conditions at the mixed state fixed point in the vicinity of a domain wall between $\rho^e$ and $\otimes_i \rho_i$. Such a fixed point state is a stabilizer state, and thus we borrow well-known tools to compute the entanglement entropy~\cite{Hamma2004}. In Appendix~\ref{app:num}, we explain the numerical method to compute the entropy for arbitrary error rate $p \in [0,1]$ from the nontrivial state, following the technique of~\cite{Sang2024}. This also explains the method to obtain Fig.~\ref{fig:violations}. In Appendix~\ref{app:fragile} we explain why the quantum channel from a trivial state to maximally dephased toric code $\rho^e$ must be fragile.

\subsection{Fixed point of the domain wall and the violation of {\bf M1}}\label{app:fixed-pts}

\emph{Stabilizer state entropy:~} Let us first make a general remark on the computation of entanglement entropy of stabilizer states. Let $\{S_i\}_{i=1}^k$ be a set of stabilizer generators. Each of them is a product of Pauli operators, $S_j^2=1$; they mutually commute $[S_i,S_j]=0$, and they are independent as generators of the group by multiplication. A stabilizer state associated with $\{S_i\}_{i=1}^k$ on $n$ qubits is 
\begin{equation}
    \rho[\{S_i\}_{i=1}^k] := \frac{1}{2^{n}} \prod_{j=1}^k (1+ S_j).
\end{equation}
The entanglement entropy of a stabilizer state is $S(\rho)= (n-k) \log 2$. In other words, computing the entropy corresponds to finding the number of independent stabilizer generators.

\emph{Domain wall configuration:~} We consider the dry boundary and the wet boundaries  in Fig.~\ref{fig:dry-wet-bdy}.  On the upper half plane, we have the fixed point $\alpha=\prod_v (1 + A_v)$ and on the lower half plane, we have $\beta=\prod_i |+\rangle^{\otimes N}$. As both states are stabilizer states, we only need to specify the stabilizer generators. This is shown in Fig.~\ref{fig:dry-wet-stabilizer}. On the upper half plane, each vertex has a stabilizer generator $A_v= \prod_{e \ni v} X_e$. On the lower half plane, the stabilizer generator is $X_e$ on each link $e$. On the domain wall lies the difference, 
\begin{itemize}[leftmargin=13pt]
    \item For each dry boundary site, we do not associate any stabilizer generator; Fig.~\ref{fig:dry-wet-stabilizer}(a).
    \item For the wet boundary, for each boundary vertex $v'$ we have a stabilizer generator that is the product of $X_e$ on three adjacent links pointing left, right and up; see Fig.~\ref{fig:dry-wet-stabilizer}(b).
\end{itemize}

\begin{figure}[!t]
    \centering
    \includegraphics[width=0.85\linewidth]{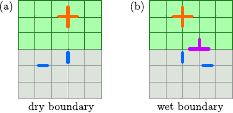}
    \caption{The stabilizer generators for the dry and wet boundaries.  
    There are three different types of stabilizer generators. Each of them is a product of Pauli $X$ on the colored links. The red ones are centered at a vertex in the upper half plane. The blue ones are on a single (horizontal or vertical) link on the lower half plane. Compared to the dry boundary, the wet boundary has an additional purple stabilizer for each boundary site. }
    \label{fig:dry-wet-stabilizer}
\end{figure}

Based on this knowledge of stabilizer generators, one can compute the von Neumann entropy of any finite region, and thus we can compute the correction to the axioms explicitly. Nonetheless, we mention an alternative way to compute the conditional mutual information $I(A\,{:}\,C|B)_\rho$, which may be more intuitive. 

Consider a stabilizer state $\rho_{ABC}$. Its reduced density matrices $\{\rho_{AB}, \rho_{BC}, \rho_B\}$ must all be stabilizer states. Then,
\begin{equation}
    I(A\,{:}\,C|B)_\rho = (k_{B} + k_{ABC}  - k_{AB} - k_{BC}) \log 2 
\end{equation}
due to the cancellation of the Hilbert space dimensions, where $k_X$ is the number of stabilizer generators associated with region $X$. Note that the stabilizer generators associated with $\rho_{AB}$ can be chosen to be the set $\calS_{A|B} \cup \calS_B$, where $\calS_B$ are the stabilizer generators associated with $\rho_B$ and $\calS_{A|B}$ is the subset of stabilizers associated with $\rho_{AB}$ which are supported on $AB$ and cannot be reduced to $B$ by multiplication. Similarly, the set of stabilizer generators associated with $\rho_{BC}$ can be chosen to be  $\calS_{C|B} \cup \calS_B$. $|   \calS_{A|B}|= k_{AB}- k_B$ and $|   \calS_{C|B}|= k_{BC}- k_B$.
Also note $\calS^{\text{loc}}_{ABC}:= \calS_{A|B} \cup \calS_B \cup \calS_{C|B}$
is a set of independent stabilizer generators, $\calS^{\text{loc}}_{ABC} \subset \calS_{ABC}$. Physical speaking $\calS^{\text{loc}}_{ABC}$ is the subset of stabilizer generators of $\rho_{ABC}$ that one can infer from local regions $AB$ and $BC$.
By simple counting 
\begin{equation}
    k_{ABC} + k_{B} - k_{AB} - k_{BC} = |\calS_{ABC}| -|\calS^{\text{loc}}_{ABC}|,
\end{equation} 
one can see that the computation of conditional mutual information
$I(A\,{:}\,C|B)_\rho$ of the stabilizer state $\rho_{ABC}$ is boiled down to finding the number of stabilizer generators that cannot be inferred from local regions.
 
From this reasoning, we verified the value of $\gamma_{\rm dry}$ and $\gamma_{\rm wet}$ in Table~\ref{tab:TE-dry-wet} and Fig.~\ref{fig:bdy-TE}. The preserved axioms on the boundary, shown in Fig.~\ref{fig:bdy-axiom}, are verified with the same consideration.

\subsection{Numerical method and data}\label{app:num}

We explain the numerical method to obtain the data in Fig.~\ref{fig:violations} of the main text. 
    We need to calculate the von-Neumann entropy of a local region at any error rate $p \in [0, 1/2]$. Especially, we want to be able to do such a computation at the spatial boundary of decoherence. 
    Following the method of~\cite{Sang2024}, since states with different anyonic configurations are orthogonal to each other, we have
    \begin{equation}\label{eq:S}
        S\Big(\sum_ip_i\rho_i \Big) = \sum_ip_iS(\rho_i)+H(\{p_i\}),
    \end{equation}
    where $H(\{p_i\}):= -\sum_i p_i \log p_i$ denotes the Shannon entropy and subscript $i$ denotes the anyon configuration excited by errors. $p_i$ is the error probability that gives the anyon configuration labeled by $i$. The von-Neumann entropy is the same across each $\rho_i$. For the no-error event corresponding to the state $\rho_0$, the axioms are already satisfied. The first terms in Eq.~\eqref{eq:S} would be canceled when it comes to the calculation of the conditional mutual information; therefore, even for a noisy case, this contribution would disappear. 
    To calculate the second term, we can use the Monte Carlo method:
    \begin{equation}
        H(\{p_i\}) = -\sum_i p_i\log(p_i) = -\mathbb{E}_{p_i}\left[\log(p_i)\right]. 
    \end{equation}
    Thus, we can sample the anyonic configuration from the error probability $p$ at each edge.  
    After getting the configuration $i$, we can get $p_i$ by contracting a tensor network. The sampling process is exact, and no Markov Chain Monte Carlo procedure is needed. Note that the variance is reduced if we calculate the violations of $\bf P0$ and $\bf M1$ for the same samples $\{p_{i,ABC}\}$:
    \begin{equation}
        \begin{aligned}
    I(A\,{:}\,C|B) = & \mathbb -\mathbb{E}_{p_i}\bigl[\log(p_{i,AB})+\log(p_{i,BC})
    \\ 
    &
    -\log(p_{i,B})-\log(p_{i,ABC})\bigr],
    \end{aligned}
    \end{equation}
    where  $\{p_i,_{AB}, p_{i,B}, p_{i,BC}\}$ are computed from the same samples as $\{p_{i,ABC}\}$. 

    The computation of probabilities for a given anyon configuration is done by the contraction of a tensor network, as in~\cite{Sang2024}. To illustrate how the data in Fig.~\ref{fig:violations} is obtained, we write the explicit tensors for $r=3$ for the calculation of probabilities related to the $\bf M1$ violation $\delta_1$ of Fig.~\ref{fig:violations}:
    \begin{equation}
        p_{i,ABC} =
        \raisebox{-0.5\height}{\includegraphics[width=0.3\textwidth]{violation_ABC.pdf}}
        \label{eq:piABC}
    \end{equation}
    
    \begin{align}
      p_{i, AB} &=
      \raisebox{-0.5\height}{\includegraphics[width=0.3\textwidth]{violation_AB}} \nonumber \\[12pt]
      &= \frac{1}{2}
      \raisebox{-0.5\height}{\includegraphics[width=0.3\textwidth]{violation_AB_mod}}
      \label{eq:pi_AB}
    \end{align}
    The upper middle tensor in the first line is $
        \delta_{\sum_ia_i=1\; \text{(mod2)}}$.
    Now consider a valid set of $\{a_i\}$, the number of the degree left for the inner links is only two  and both of their weight are one so it's two times the original tensor network.
    \begin{equation}
      p_{i, B} =
      \raisebox{-0.5\height}{\includegraphics[width=0.2\textwidth]{violation_B.pdf}}
      \label{eq:pi_B}
    \end{equation}
    
    \begin{equation}
      p_{i, BC} =
      \raisebox{-0.5\height}{\includegraphics[width=0.2\textwidth]{violation_BC.pdf}}
      \label{eq:pi_BC}
    \end{equation}
       
    Even though the partition $AB$ crosses the boundary, because region $AB$ is not simply connected as shown in Fig.~\ref{fig:violations}(a), we still need to take into account of the possible excitation of this closed loop stabilizer in the lower half plane. The sample sizes for $r=[4, 5, 6, 7]$ for Fig.~\ref{fig:violations} are $[10^7, 2.5\times 10^6, 2.5\times10^5, 5\times10^4]$ respectively.
    
    Finally, the local tensors are:
    \begin{equation}
      \raisebox{-0.5\height}{\includegraphics[width=0.08\textwidth]{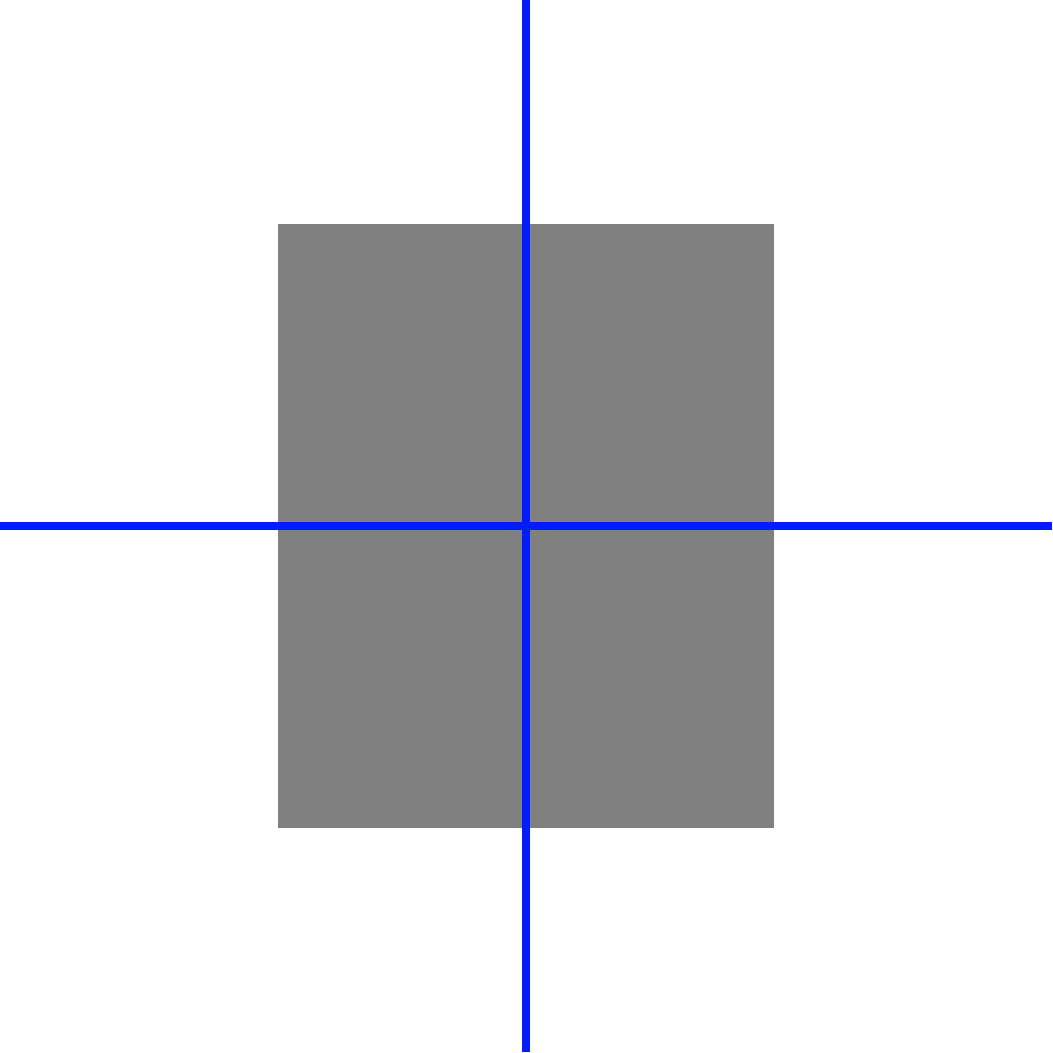}}
      = \delta_{i+j+k+l\equiv 0 \; (\text{mod} \; 2)}
     \label{eq:four_leg_unexcited}
    \end{equation}

    \begin{equation}
      \raisebox{-0.5\height}{\includegraphics[width=0.08\textwidth]{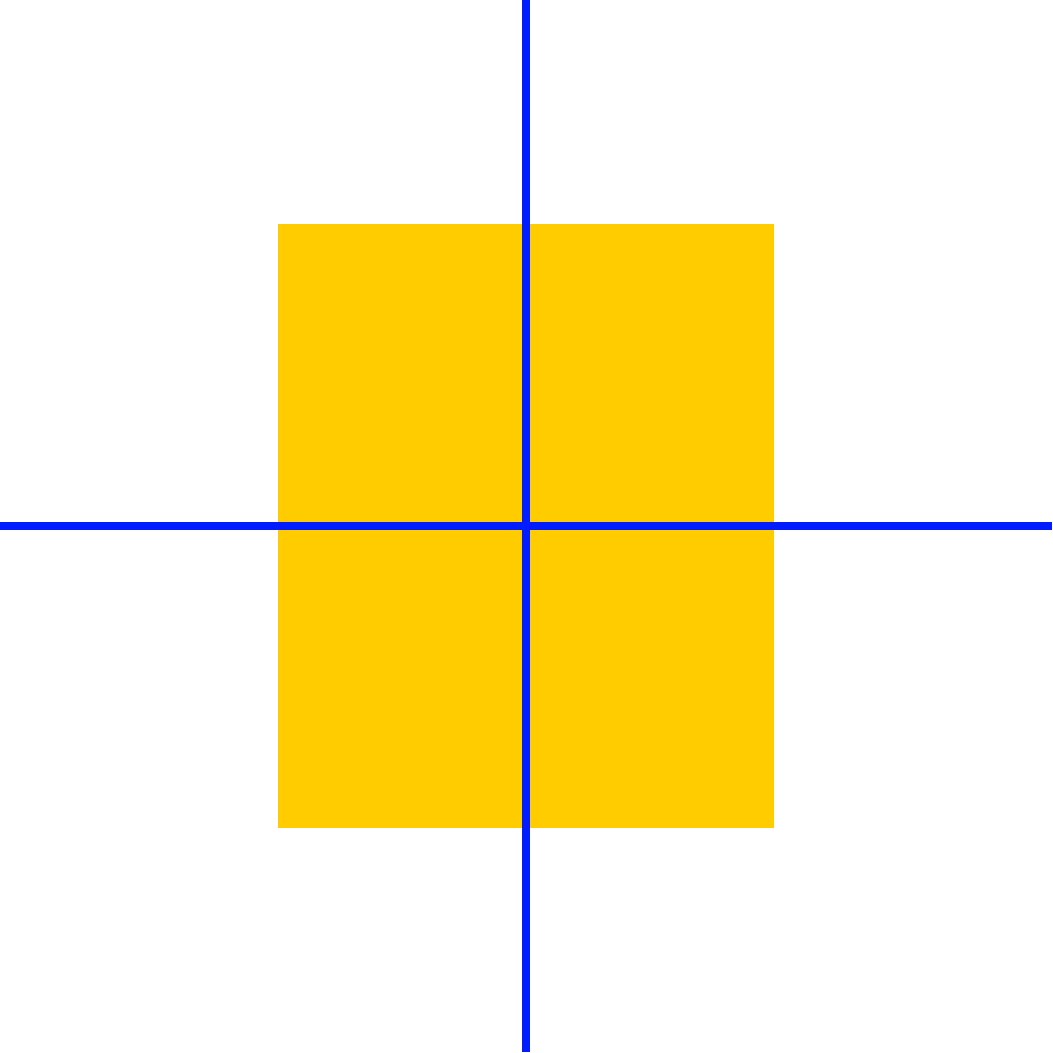}}
      = \delta_{i+j+k+l\equiv 1 \; (\text{mod} \; 2)}
      \label{eq:four_leg_excited}
    \end{equation}

    \begin{equation}
      \raisebox{-0.5\height}{\includegraphics[width=0.1\textwidth]{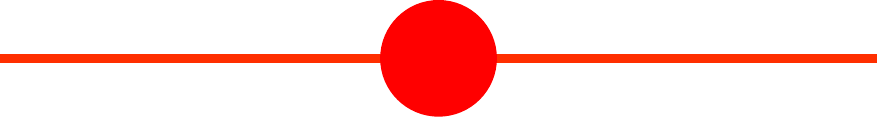}}
      = \begin{pmatrix}
            1 - p & 0 \\
            0 & p
         \end{pmatrix}.
      \label{eq:two_leg}
    \end{equation}
    
    For three-leg tensors, they're $\delta_{i+j+k\equiv 0, 1\; \text{mod 2}}$ depending whether there is an anyon there. Finally, the uncontracted tensor legs will be sumed over in the end.
    No approximation is used for Fig.~\ref{fig:violations}, but approximate frameworks such as boundary matrix product states (MPS) or tensor renormalization group methods~\cite{HOTRG,TRG} may help us investigate larger system sizes.

    \begin{figure*}[t]
  \begin{minipage}{\linewidth}
  \centering
  \centering
         \begin{tikzpicture} 
    \node[anchor=south west, inner sep=0] (img1) at (0,0) {\includegraphics[width=0.27\columnwidth]{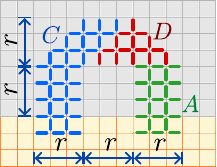}};
    \node at (2.3, -0.5) {(a)};  
    \node[anchor=south west, inner sep=0] (img2) at (5.2,0) {\includegraphics[width=0.29\columnwidth]{violation_renyi.pdf}};
    \node at (7.57, -0.5) {(b)};
     \node[anchor=south west, inner sep=0] (img3) at (10.6,-0.73) { \includegraphics[width=0.35\linewidth]{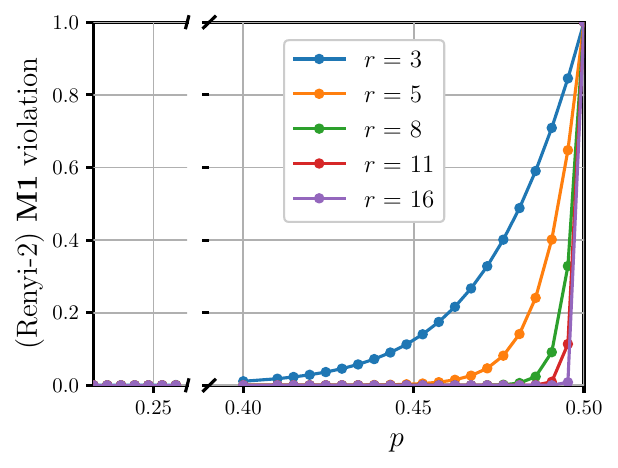}};
    \node at (14, -0.5) {(c)};
  \end{tikzpicture}      
        \caption{ {\bf Fragility of the wet boundary.} (a) Decoherence on the upper half plane, where the $Z$-decoherence acts on the gray links. The subsystems $A,C,D$ are shown at $r=3$. For  $p=0.5$, the  decohered state~\eqref{eq:fragile-rho} has a wet boundary between $\rho^e$ and the product state. (b) The tensor network to calculate the Renyi-2 version of $I(A\,{:}\,C|D)$, which diagnoses a R\'enyi-2 analogue of the {\bf M1} violation at the boundary, $\delta^{(2)}_1:= S^{(2)}_{AD} + S^{(2)}_{DC} - S^{(2)}_{D} - S^{(2)}_{ADC}$ (known as $2\gamma_{\text{wet}}$ in Sec.~\ref{subsec:channel-wet-dry}). (c) A plot of the Renyi-2 {\bf M1} violation $\delta^{(2)}_1$, in units of $\log 2$, as a function of the error rate $p$ and subsystem sizes $r$.
        }
  \label{fig:fragile}
  \end{minipage}
\end{figure*}

\subsection{From product state to $\rho^e$ and fragility}\label{app:fragile}

    Starting from $\ket{+}^{\otimes N}$, the density matrix under $\prod_{e \in p} Z_e$ decoherence on every plaquettes for the toric code is:
    \begin{equation}\label{eq:fragile-rho}
        \rho_\textrm{u} \propto \sum_{\bm{l}_\perp \in {\cal S}_d }  \qty[ \prod_{v \in \partial \bm{l}_\perp} (1-2p) ]X^{\bm{l}_\perp}  
    \end{equation}
    where $p$ is the error rate of $B_p^f$ and  ${\cal S}_d$ is the set of all possible strings on the dual lattice. 
    In this appendix, we investigate if the state is in the same phase as the dephased toric code $\rho^e$, above a certain criticality. 
    
    For this purpose, we are interested in computing $I(A\,{:}\,C|D)_{\rho_\textrm{u}}$ in Fig.~\ref{fig:fragile}(a).
    Although there is no efficient ways to get the von-Neumann entropy, we can calculate the Renyi-2 entropy as follows:
    \begin{equation}\label{eq:ap_fragile}
       \exp(- S^{(2)}(\rho_\textrm{u})) = \Tr(\rho_\textrm{u}^2) = 2^{-N}\sum_{\bm{l}_\perp}(1-2p)^{2|\partial \bm{l}_\perp|}.
    \end{equation}
    Here the sum runs over all configurations $\bm{l}_\perp$.    

This computation can be done efficiently using a tensor network, where the $r = 3$ case is illustrated as Fig.~\ref{fig:fragile}(a) and (b). The local tensors are:    
    \begin{equation}
    \begin{aligned}
      \raisebox{-0.5\height}{\includegraphics[width=0.08\textwidth]{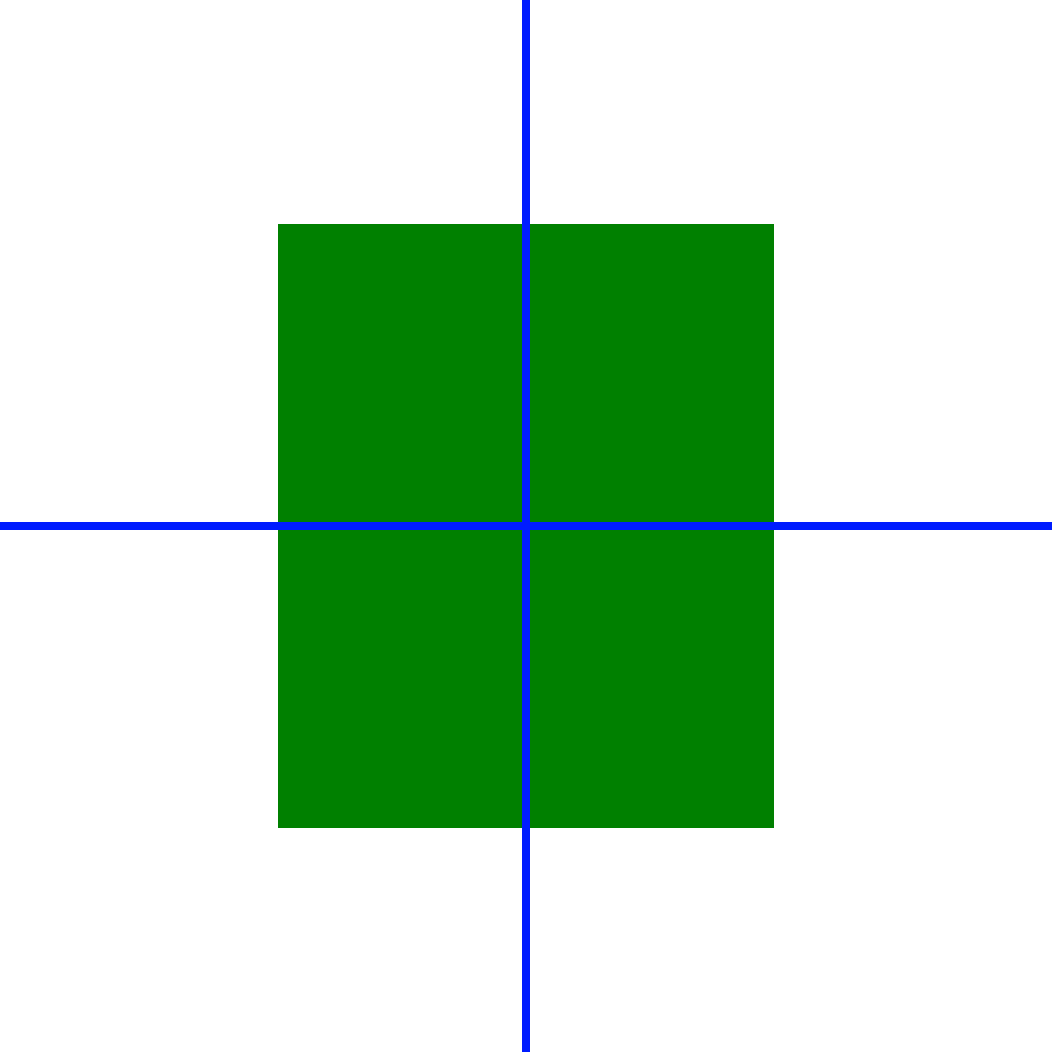}}
      = & \,\, \delta_{i+j+k+l\equiv 0 \; (\text{mod} \; 2)} \\ & +  (1 - 2p)^2 \delta_{i+j+k+l\equiv 1 \; (\text{mod} \; 2)}
    \end{aligned}
    \label{eq:four_leg_excited_renyi}
    \end{equation}
    
Note that here the computation is simpler than that in Sec.~\ref{app:num} as the weight of the excitations is on the anyons themselves instead of being on the edges so we can sum over all the configurations with one tensor network without doing Monte Carlo. For three-leg tensors, we just replace the delta functions with $\delta_{i+j+k\equiv 0, 1\; \text{mod 2}}$ respectively. Same for two-leg tensors. Finally, the indices of uncontracted legs are sum over as before.

Contracting this tensor network numerically yields the Rényi-2 entropy \(S^{(2)}\) as a function of the probability \(p\) of applying $\prod_{e\in}Z_e$ to each plaquette.  The result is shown in Fig.~\ref{fig:fragile}(c).

\subsubsection{Another argument, with bulk $\gamma_{\text{LW}}$}

Another indirect way to justify the fragility of the channel preparing the state $\rho^e$ is through the theoretical computation of the Renyi-2 LW topological entropy of the bulk. 
First, the Renyi-2 entropy of \eqnref{eq:ap_fragile} can be rewritten as 
\begin{align}
    S^{(2)}(\rho) = N\log 2 - \ln {\cal Z}, \quad {\cal Z} = \sum_{\bm{l}_\perp\in {\cal S}_d} e^{-\beta |\partial \bm{l}_\perp|},
\end{align}
where $\beta = \ln \frac{1}{(1-2p)^2}$. The additive volume term $N\log 2$ cancels in the conditional-mutual-information and topological-entropy combinations used below. Note that ${\cal Z}$ is equivalent to the partition function of the two-dimensional Ising gauge theory since
\begin{align}
    {\cal Z} = \sum_{\bm{l}_\perp \in {\cal S}_d} e^{-\beta |\partial \bm{l}_\perp|} = \sum_{\bm{\sigma}} e^{-\beta \sum_p \prod_{e \in p} \sigma_e}.
\end{align}
where $\bm{\sigma}=\{\sigma_e \}$ lives in the links of the dual lattice.
Let $F_R = - \ln {\cal Z}_R$ be the free energy of the region $R$. Then the (Renyi-2) topological entropy is given as 
\begin{align}
    \gamma^{(2)}_{\textrm{LW}} = \frac{1}{2} \qty( F_{AB} + F_{BC} - F_B - F_{ABC} )
\end{align}
where the partition is that in \eqnref{eq:TE}.
However, since the Wilson-loop expectation value obeys an area law, at any finite temperature $\gamma^{(2)}_{\textrm{LW}}$ should vanish for large enough system size. At zero temperature, the exact two-fold topological degeneracy (from loop) provides $F_{ABC}$ smaller than $F_{AB} + F_{BC} - F_B$ by $\log 2$, but not at any finite temperature.

\section{Application of phase equivalence}\label{app:phase-equiv}

We discuss a few applications of the mixed state phases equivalence as in Def.~\ref{def:p.f.-equivalence}. 





\subsection{Phase distinction beyond topological entropy}\label{app:phase-equiv-coherent}

We provide examples of mixed state fixed points which can be distinguished even though they have identical topological entropy. 

The key idea is is to look at the information convex set of a certain region and use the isomorphism theorem. If two mixed state fixed points $\alpha \sim \beta$, they can be jointed by the transparent wall (Def.~\ref{def:p.f.-equivalence}). Therefore, according to the isomorphism theorem (Eq.~\eqref{eq:iso-thm}), the information convex sets for a given region embedded (or immersed) in these states must be isomorphic.
We argue that such statements should also apply to the approximate information convex sets Def.~\ref{def:ICS-mixed-robust}. The following examples are based on this approach.

\begin{exmp}[Toric code versus Ising anyon]

As a warm up example, we explain why the toric code fixed point $\rho^{\rm T.C.}$~\cite{Kitaev1997} and the Ising anyon fixed point $\rho^{\rm Ising}$~\cite{Kitaev2005,Rowell2007} are different by their information convex sets. Note that both theories have identical topological entropy $\gamma_{\rm LW} = \log 2$ and thus topological entropy cannot distinguish them. 

Consider the information convex set of the annulus.
$\Sigma(\text{annulus}|\rho^\text{T.C.})$ is a simplex with 4 extreme points. In contrast, $\Sigma(\text{annulus}|\rho^\text{Ising})$ has 3 extreme points. Because of this difference, there cannot be a transparent wall connecting these two fixed points, when we treat them as mixed state fixed points. It follows that
\begin{equation}
    \rho^{\rm T.C.} \not\sim \rho^{\rm Ising}.
\end{equation}
\end{exmp}

\begin{exmp}[semion versus $\rho^e$] We show that the chiral semion fixed point $\rho^s$ and the decohered toric code fixed point $\rho^e$ are in different phases, 
\begin{equation}\label{eq:semion-appendix}
    \rho^s \not\sim \rho^e.
\end{equation}
This does not follow from the topological entropy because $\gamma_{\rm LW} = \log \sqrt{2}$. The information convex set of the annulus cannot distinguish them either because 
\begin{equation}
    \Sigma(\text{annulus}|\rho^s) \cong \Sigma(\text{annulus}|\rho^e) \cong \text{interval}. 
\end{equation}
An interval is a convex set with a pair of extreme points. Such extreme points are mutually orthogonal, and due to the Abelian nature of the theory, not distinction can be extracted from the entropy difference either.
The justification of \eqref{eq:semion-appendix} lies in the information convex set of an immersed punctured torus $W$:
\begin{equation}
  \includegraphics[width=0.82\linewidth]{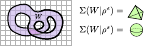}
\end{equation}
For $\rho^e$ the information convex set associated with $W$ is purely classical, with 4 extreme points. In contrast, for $\rho^s$, the information convex set forms a Bloch ball with a sphere of extreme points. This completes the argument.
\end{exmp}

More generally, there is a clear contrast of the pure-state bootstrap reference state $\sigma^{\rm pure}$ satisfying {\bf A0} and {\bf A1} and the decohered $G$ quantum double (i.e. classical loop state) $\rho^m$ in Eq.~\eqref{eq:fp-classical-loop}.   
\begin{itemize}[leftmargin=13pt]
    \item The exists a state in $\Sigma(W|\rho^m)$ such that, its reduced density matrix to annuli $X\subset W$ and $Y \subset W$ (as Fig.~\ref{fig:XY-of-W}) can simultaneously be identical to the vacuum state $\rho^m|_X, \rho^m|_Y$.
    \item For $\sigma^{\rm pure}$, there cannot be any state in  $\Sigma(W|\sigma^{\rm pure})$ which simultaneously agrees with the vacuum upon reduced to both $X$ and $Y$, unless the topological entropy $\gamma_{\rm LW} =0$. This is due to the entropic uncertainty relations forced by information theoretic constraints~\cite{Braiding2023,Zhang2012}. 
\end{itemize}
This distinction arises because braiding is necessarily nondegenerate in the pure-state case~\cite{Braiding2023}.

\begin{figure}[h]
    \centering
    \includegraphics[width=0.9\linewidth]{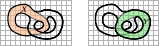}
    \caption{Two annuli $X, Y$ are subsets of the immersed punctured torus $W$. Both $X$ and $Y$ are embedded on the physical system. $X$ and $Y$ correspond to two noncontractable annuli of the punctured torus ($X,Y\subset W$).}
    \label{fig:XY-of-W}
\end{figure}

The double semion and the toric code have isomorphic information convex sets for all immersed regions. Thus, they cannot be distinguished by the shapes of the information convex sets. This is expected as the information convex set only detect the fusion properties rather than braiding statistics. Nonetheless, braiding statistics \emph{can} be detected by the automorphism of the information convex sets induced by a nontrivial deformation of the immersed region. Applying the idea in Ref.~\cite{Shi2024-figure8}, we can distinguish this pair of phases. (The extraction of topological $S$-matrix and the proof of braiding non-degeneracy in \cite{Shi2019Verlinde} needed both the string operators and immersed regions which can distinguish the two phases in the context of pure-state bootstrap. However, for mixed state phases, the approach in Ref.~\cite{Shi2024-figure8} Appendix~B is useful to bypass the need of string operators.)

\begin{exmp}[$\rho^e$ versus $\rho^f$]
    Is it true that the $e$- and $f$-decohered toric code belong to two different phases? While previous literature argued an answer based on the strong symmetry~\cite{Wang_intrinsicmixedstateTO_2023}, no proof has been provided in the context when there are perturbations breaking the exact 1-form strong symmetry. Our approach opens a possibility to prove such a statement. We left it as an open problem for now, and details will be available in \cite{Shi2026}.
\end{exmp}

\section{Quantum double calculations}\label{app:QD}

Our notation of quantum double is similar to previous literature, e.g., \cite{Bombin2008}. 
The quantum double model~\cite{Kitaev1997} is a generalization of the toric code model, which allows a non-Abelian group $G$.  (See Ref.~\cite{Bombin2008} for a clear review.) The lattice model has a tensor product Hilbert space on the edges of a lattice. The Hilbert space on each edge $e$ is $\calH_e := \{ | g\rangle\, |\, g \in G \}$, which forms an orthonormal basis for the Hilbert space of dimension $|G|$. For visual representation convenience, we put an orientation on each edge such that upon flipping the orientation, we require
\begin{equation}\label{eq:g-gbar}
\vcenter{\hbox{\includegraphics[width=0.45\linewidth]{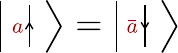}}}.
\end{equation}
On each oriented edge $e$, we define the following operators:
\begin{eqnarray}
       T_e^g : &&\quad T_e^g | h\rangle =  \delta_{g,h}|h\rangle, \\
       L_e^g: &&\quad L_e^g |h \rangle = | gh\rangle,\\
       \bar{L}_e^g: &&\quad  \bar{L}_e^g |h \rangle = | h \bar{g}\rangle.
\end{eqnarray}
where $\bar{g}:=g^{-1}$. For an Abelian group, $\bar{L}_e^g = (L_e^g)^\dagger = L_e^{\bar{g}}$. However, for a non-Abelian group, in general, $\bar{L}_e^g \neq L_e^{\bar{g}}$.

Equipped with single-edge operators, we define operators acting on the oriented 2-cellular complex where $\rd e = \sum_{v \in e} \epsilon_{e,v} v$ and $\rd p = \sum_{e \in p} \epsilon_{e,p} e$, where $\epsilon = \pm 1$ and $\rd^2 p = 0$. For instance, if $e \in E$ is oriented from vertex $v_1$ to $v_2$, one takes $\epsilon_{e,v_1} = 1$ and $\epsilon_{e,v_2} = -1$.  
With this understanding, define  
\begin{align}\label{eq:QD-terms}
    \includegraphics[width=0.7\linewidth]{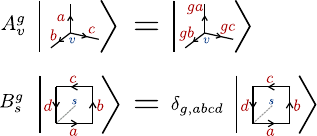}
\end{align}
Here, $v$ refers to a vertex, and $s$ is a face $p$ with a starting position chosen. What if the orientation of the edges are different from the figure? This is explained by Eq.~\eqref{eq:g-gbar}. For example, if the edge orientation for $b$ is toward $v$ instead of away from $v$, we can first reverse the orientation, apply $A_v^g$, and reverse it back as follows:
\begin{align}
    A_v^g |a_+, b_-,c_+ \rangle &= A_v^g | a_+, \bar{b}_+,c_+ \rangle = |ga_+, g\bar{b}_+, gc_+ \rangle  \nonumber \\
    &= |ga_+, b \bar{g}_-, gc_+ \rangle
\end{align}
where the subscript is $\epsilon_{e,v}$. The case for $B^g_s$ can be understood similarly.

While $A_v^g$ is unitary, $B_s^g$ is a projector. Using these operators, we define 
commuting projectors  
\begin{equation}\label{def:projector-A-B}
    A_v^R := \sum_{g\in G} \frac{\dim R}{|G|} \chi_R(g^{-1}) A^g_v, \quad B_p^C := \sum_{g\in C} B_s^{g},
\end{equation}
where $R \in (G)_{\text{ir}}$ and $C \in (G)_{\text{cj}}$ are labels for irreducible representation and conjugacy class of $G$. 
The projector properties of $B_p^C$
\begin{equation}
    B_p^C B_p^{C'}= \delta_{C,C'} B_p^C
\end{equation}
is obvious, while the projector property of $  A_v^R $ that is
\begin{equation}
    A_v^R A_v^{R'} = \delta_{R,R'} A_v^R
\end{equation}
will be verified in Appendix~\ref{app:projector}. 
These two sets of projectors $\{ A_v^R \}$ and $\{ B_p^C \}$ are commuting:
\begin{equation}
    [A_v^R, B_p^C]=0, \quad \forall v, p.
\end{equation}  

Of importance are the projectors corresponding to the identity irreducible representation $\bmcharge{1}\in (G)_{\text{ir}}$ and $\bmflux{1} \in (G)_{\text{cj}}$.
\begin{equation}  
A_v^{\bmcharge{1}} = \frac{1}{|G|} \sum_{g \in G} A_v^g, \quad \text{and} \quad  B_p^{\bmflux{1}} = B_s^{{1}}.
\end{equation}
where $1 \in G$ is the identity element.

For a given 2D lattice, we can write the quantum double Hamiltonian as
\begin{equation}
     H = - \sum_v A^{\bmcharge{1}}_v  - \sum_p B^{\bmflux{1}}_p,
\end{equation}
If we take $G=\mathbb{Z}_2$, we get the toric code.

\subsection{Projector properties}\label{app:projector}

Given a group $G$, for all elements $g\in G$, denote the unitary matrix of the irreducible representation of $g$ by $\Gamma_R(g)$. We have the great orthogonality relation:  
\begin{equation}
    \sum_{g\in G}\Gamma^{ij}_R(g)\bar\Gamma_{R'}^{i'j'}(g)=\frac{|G|}{\dim R}\delta_{R,R'}\, \delta_{i,i'} \,\delta_{j,j'},
\end{equation}
where the bar on top of $\Gamma$ means complex conjugation. By summing over $(i,i',j,j')$, we obtain the first orthogonality relation for the characters $\chi_R(g):=\Tr \Gamma_R(g)$:
\begin{equation}
    \frac{1}{|G|} \sum_{g\in G} \chi_{R}(g)\bar\chi_{R'}(g)= \delta_{R,R'}.
\end{equation}

Finally, we also have the second orthogonality relation:
\begin{equation}
     \sum_{R\in(G)_{ir}}\chi_R(C)\bar\chi_R(C') = \frac{|G|}{|C|}\delta_{C,C'}.
\end{equation}

Consider a set of operators $\{\calO_g\}_{g\in G}$ such that $\calO_g \calO_h =\calO_{gh}$. We define 
\begin{equation}
    e_R^{ij}=\frac{\dim R}{|G|}\sum_{g\in G}\bar\Gamma_R^{ij}(g) \,\calO_g.
\end{equation}
With straightforward but lengthy algebra, we see
\begin{align}
\begin{aligned}
     e_R^{ij}\, e_{R'}^{i'j'} &=\frac{\dim R \cdot \dim {R'}}{|G|^2}\sum_{g,g'} \bar\Gamma^{ij}_R(g) \, \bar\Gamma_{R'}^{i'j'}(g')\calO_g \calO_{g'} \\
    & = \frac{\dim R \cdot \dim {R'}}{|G|^2}\sum_{g,r} \bar\Gamma^{ij}_R(g) \, \bar\Gamma_{R'}^{i'j'}(g^{-1}r) \calO_r \\
    & = \frac{\dim R \cdot \dim {R'}}{|G|^2}\sum_{g,r,k}\bar\Gamma^{ij}_R(g)\Gamma_{R'}^{ki'}(g)  
    \bar\Gamma_{R'}^{kj'}(r) \calO_r \\
    & = \frac{\dim R}{|G|}\sum_{r,k}\delta_{R,R'}\delta_{i,k}\delta_{j,i'}\bar\Gamma_R^{kj'}(r) \, \calO_r \\
    & = \delta_{R,R'}\delta_{j,i'}e_R^{ij'}.
    \end{aligned}
\end{align}
Now, let us define $e_R=\sum_{i}e_R^{ii}$. It follows that
\begin{align}\label{Rorthogonal}
    e_R \, e_{R'}= \delta_{R,R'} \, e_R.
\end{align}
With Eq.~\eqref{Rorthogonal}, we can see that each $\{A_v^R\}_{R\in (G)_{\text{ir}}}$ of Eq.~\eqref{def:projector-A-B} is an orthogonal set of projectors because $A_v^g$ is a valid choice of $\calO_g$.

\subsection{Classical loop configurations}\label{app:non-Abelian-justify}

Consider the quantum double model on a lattice $M$ with a number of faces, links, and vertices, where the Hilbert space on each link is $\calH_e =\text{span}\{|g\rangle | g\in G\}$ on each link.  
 Let $|c \rangle$ be a product state in the $\{ |g\rangle_e\}\rangle$ basis, and with zero flux on every plaquette (i.e., $B_p^{\bmflux{1}} |c \rangle = |c \rangle$, $\forall p$). We say $|c \rangle \sim |c'\rangle$ if $|c'\rangle = \prod_v A_v^{g_v} |c \rangle$ by the action of unitary vertex terms $A_v^{g_v}$, with some $ g_v\in G$. Using local zero-flux condition, the following can be checked:
 \begin{itemize}[leftmargin=13pt]
\item Suppose $M$ is a connected tree, i.e., it has no faces and no closed loops formed by links; see Fig.~\ref{fig:trees}(a). Then it is always true that $|c \rangle \sim |1,1,\cdots, 1\rangle$. Furthermore, the number of configurations in the equivalent class is $|G|^{V-1}$ as the number of vertices ($V$) of a connected tree equals the number of links plus one.

Suppose $M$ is a union of $k$ connected trees, then it is still true that $|c \rangle \sim |1,1,\cdots, 1\rangle$. But, the number of configurations in the equivalent class is $|G|^{V-k}$.

    \begin{figure}[h]
       \centering
       \includegraphics[width=0.9\linewidth]{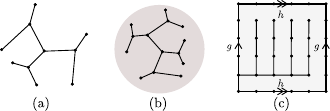}
       \caption{An illustration of trees and lattices. (a) a connected tree. (b) A tree on a sphere obtained by removing links according to a certain rule. (c) A tree that is attached to a rectangular frame on a torus.}
       \label{fig:trees}
   \end{figure}
 
     \item When $M$ is a disk or a sphere, we always have $|c \rangle \sim |1,1,\cdots, 1\rangle$. Thus, there is a unique equivalence class under $\sim$. 

   To see why this is true, we first turn the lattice into a tree in the following way. We remove a number of links in such a way that every step removes a face until we are left with 0 or 1 face. The resulting links form a connected tree, as in Fig.~\ref{fig:trees}(b). The important observation is that the group elements on the links in the tree determine the full configuration on the lattice of the disk or sphere. This is because the group elements on the removed links can be inferred from the tree by zero flux constraints of the faces.  From the property of connected trees, we see it is always possible to convert $|c \rangle \sim |1,1,\cdots, 1\rangle$ by the action of vertex terms. 
   Furthermore, the number of configurations in the equivalence class equals $|G|^{V -1}$, where $V$ is the number of vertices of the original lattice. 

   Similarly, if $M$ is the union of $k$ disjoint disks, $|c \rangle \sim |1,1,\cdots, 1\rangle$, and the number of configurations in the equivalent class is $|G|^{V-k}$.

     \item When $M$ is the torus, the equivalence class of $|c \rangle$ is labeled by $C_{g,h}$. Here, the set $C_{g,h} := \{ (r g \bar{r}, r h \bar{r})| r \in G\}$, defined when $gh = h g$, $g,h\in G$. Let $\calS_{g,h}$ be the set of configurations equivalent to that shown in Fig.~\ref{fig:gh}. Furthermore, the sizes of such sets satisfy
     \begin{equation}
         |\calS_{g,h}| = |C_{g,h}| \cdot |\calS_{1,1}|, \quad |\calS_{1,1}| =|G|^{V -1}.
     \end{equation}

 \begin{figure}[!t]
     \centering
     \includegraphics[width=0.24\textwidth]{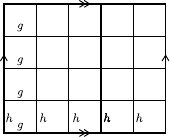}
     \caption{A representative configuration of the torus in the equivalent class $\calS_{g,h}$, where $g, h \in G$ and $gh = hg$. On each link, we assign a group element as shown. We omit the label of the group element whenever it is the identity $1\in G$. Each horizontal link is assigned a right arrow and each vertical link is assigned an up arrow, and we omit the drawing of them.}\label{fig:gh}
 \end{figure}

    To see why these statements are true, we first remove a number of links in such a way that every step removes a face until we are left with a unique face of the torus. The remaining graph is a tree attached to a square frame, with periodic boundary conditions, as Fig.~\ref{fig:trees}(c). It is then easy to see that the nontrivial loop and the zero flux constraints give rise to the equivalence classes $\calS_{g,h}$ stated. The condition $gh=hg$ follows from the zero flux constraint. The application of star terms $\prod_v A_v^r$  can result in $(g,h)\to (rg r^{-1}, r h r^{-1})$. That is why $C_{g,h}$ captures the equivalence classes.  
    The presence of the tree does not provide extra equivalence classes and it only increases all $|\calS_{g,h}|$ by the same factor. This justifies the statement. 
 \end{itemize}

 \subsection{The memory capacity of the torus} \label{app:derivation_memory}

Recall that the definition of the memory capacity in (\ref{memory_cap})
\begin{equation}
    \calQ(M|\sigma) \equiv \hspace{-5pt} \max_{\rho,\lambda \in \Sigma(M|\sigma)} \hspace{-5pt} S(\rho) - S(\lambda).
\end{equation}

We explain the group theoretical computation that leads to Eq.~\eqref{eq:capacity}, namely $\calQ(\mathbb{T}^2|\rho^m) = \log |G| + \log |G_{\rm cj}|$, where $|G_{\rm cj}|$ is the number of conjugacy classes of $G$.
First, given that the extreme points are orthogonal, and that the entropies of the extreme points are given by $S(\rho^{C_{g,h}}) = S(\rho^{C_{1,1}}) + \log |C_{g,h}|$, the minimum entropy is $S(\rho^{C_{1,1}})$ and the maximum entropy is $S(\rho^\star)=S(\rho^{C_{1,1}})+ \log\big( \sum  |C_{g,h}| \big)$, where the sum run over different $g, h\in G$ such that $gh = hg$ which give inequivalent $C_{g,h}$. This maximum entropy can be obtained by optimizing the convex sum $\rho = \sum_{(g,h)} p_{g,h} \rho^{C_{g,h}}$ over $p_{g,h}$.

To proceed, let $E_g$ be the centralizer group $\{ x \in G \,|\,x g x^{-1} = g\}$.
Note that two elements $h_1, h_2 \in E_g$ with $h_2 = s h_1 s^{-1}$ for some $s \in E_g$  
yield the same set $C_{g,h_1} = C_{g,h_2}$. Hence, one only sums \(\lvert\Cgh\rvert\) over $h$ in a choice of representatives of the conjugacy classes in \(\Eg\), and then sums over $g$ over a choice of representatives of the conjugacy classes in $G$.

Let $H:=\{h^{(i)} \}_{i=1}^k$ be the set of representatives of those conjugacy classes, $k=|(E_g)_{\rm cj}|$, then 
\begin{align}
    \sum_{h \in H} |C_{g,h} | = \sum_{i=1}^k \frac{|G|}{|E_{g,h^{(i)}}|}
\end{align}
where we used the orbit-stabilizer theorem, where $E_{g,h}:=\{ x \in G \,|\, xg x^{-1} = g, \, x h x^{-1} = h\}$.
On the other hand, the size of the conjugacy class for $h^{(i)}$ in $E_g$ is
\begin{align}
    | \text{conj. class of } h^{(i)} \text{ in } E_g | = \frac{|E_g|}{|E_{g,h^{(i)}}|}.
\end{align}
Since the union of these conjugacy classes exhausts $E_g$, we have 
\begin{align}
    \sum_{i=1}^k \frac{|E_g|}{|E_{g,h^{(i)}}|} = |E_g|.
\end{align}
Therefore, 
\begin{align}
    \sum_h |G_{g,h}| = \frac{|G|}{|E_g|} \sum_{i=1}^k \frac{|E_g|}{|E_{g,h^{(i)}}|} = |G|.
\end{align}
Let $|G_{\rm cj}|$ be the number of conjugacy classes of $G$. Then the sum $\sum_{g,h} |C_{g,h}|$ equals $|G|\cdot |G_{\rm cj}|$.

\subsection{Levin-Wen topological entropy}\label{app:Levin-Wen}

We derive the Levin-Wen topological entropy $\gamma_{\rm LW} = \log \sqrt{|G|}$ (Eq.~\eqref{eq:LW-TE-G} of the main text), for the fixed point $\rho^m=\prod_p B^{\bmflux{1}}_p$ of the quantum double.
To do so, we apply the tools in Appendix~\ref{app:non-Abelian-justify}.
The important observation is that the wave function $\rho^m$ is an equal weight convex combination of zero-flux classical loop configurations. We thus have 
\begin{itemize}
    \item Let $D$ be the union of $k$ disjoint disks. Then, $S(\rho^m_D) =  (V_D-k) \log |G|$ where $V_D$ is the number of vertices in $D$. 

    \item Let $X$ be an annulus (embedded in a disk), then $S(\rho^m_X) = (V_X -1) \log |G|$.   The reason is that the zero flux conditions of the hole of the annulus effectively make the region a single disk.  
    \end{itemize}
    Now we are ready to compute the Levin-Wen topological entropy, 
    \begin{equation}
        \gamma_{\rm LW}= \frac{1}{2} I(A\,{:}\,C|B)_{\rho^m},
        \quad \text{for} \quad \vcenter{\hbox{\includegraphics[width=0.08\textwidth]{figs/fig_LW-TEE-v2.pdf}}} 
    \end{equation}
    According to the discussion above, we can first shrink the lattice on the annulus $ABC$ into a noncontractible loop and a tree attached to it. Let us call the resulting graph as $\mathbb{E}$, which is arranged in such a way that when $\mathbb{E}$ is restricted to $AB, BC$, it becomes a single tree, and when it is reduced to $B$ it becomes two disjoint trees. 
    
    We then compute 
    \begin{equation}
    \begin{aligned}
        I(A\,{:}\,C|B)_{\rho^m} &= (V_{AB}-1)\log |G| + (V_{BC}-1  ) \log |G| \\
        &-(V_B - 2) \log |G| - (V_{ABC}-1) \log |G| \\
        &=E_{AB} \log |G| + E_{BC}  \log |G| \\
        & - E_B \log |G| - (E_{ABC}-1) \log |G| \\
        &= \log |G|. 
    \end{aligned}   
    \end{equation}
    Here the first equality follows from the entropy counting from the number of vertices as explained above, where $V_Y$ is the number of vertices in region $Y$. About the 2nd equality, $E_{AB}$, $E_{BC}$, $E_B$ and $E_{ABC}$ are the number of links of the graph $\mathbb{E}$; they are related to the number of vertices due to the tree topology (with 1 or 2 components or with an extra loop). The last equality follows from $E_{AB} + E_{BC} = E_B + E_{ABC}$. Thus, the topological entropy $\gamma_{\rm LW} = \log \sqrt{|G|}$.

\bibliography{Bibliography}
 
\end{document}